\documentclass[cmp,11pt,tightenlines,referee,nofootinbib]{revtex4}

\usepackage{pre_all}
\usepackage{pre_GSF}
\DeclareFontFamily{U}{bbold}{}
\DeclareFontShape{U}{bbold}{m}{n}{<-> bbold10}{}
\usepackage{xurl}
\usepackage{microtype}
\usepackage{upquote}
\allowdisplaybreaks

\hypersetup{
pdfstartview = {FitH},
}
\hypersetup{
	colorlinks=true,
	linkcolor=brown,
	citecolor=red,
	urlcolor=blue
}

\newcommand{\Vsig}{\mathbb V_{\sigma}}
\newcommand{\Msig}{\mathcal M_{\sigma}}
\newcommand{\etasig}{\eta_{\sigma}}
\newcommand{\Osig}{\mathcal O_{\sigma}}
\newcommand{\ambip}[2]{(#1,#2)_{\sigma}}
\newcommand{\tanip}[2]{\langle #1,#2\rangle}
\newcommand{\trspin}{\operatorname{tr}_{2}}
\newcommand{\trvec}{\operatorname{Tr}_{3}}
\newcommand{\Inertia}{\operatorname{In}}
\newcommand{\impart}{\operatorname{Im}}
\newcommand{\sgnop}{\operatorname{sgn}}

\definecolor{currentrevision}{RGB}{0,120,65}

\usepackage{titlesec}

\begin{document}

\title{Holonomy Reconstruction and Character Varieties of Lorentzian Curved Tetrahedra}

\author{{\bf Hongguang Liu}}\email{liuhongguang@westlake.edu.cn}
\affiliation{Institute for Theoretical Sciences, and Departments of Physics and Astronomy, Westlake University, Hangzhou 310030, China}
\affiliation{Institute of Natural Sciences, Westlake Institute for Advanced Study, Hangzhou 310024, China}

\author{{\bf Qiaoyin Pan}}\email{qpan@tsinghua.edu.cn}
\affiliation{Yau Mathematical Sciences Center, Tsinghua University, Jingzhai, Haidian district, Beijing 100084, China}

\date{\today}

\begin{abstract}
We give a recognition and reconstruction theorem for finite strictly convex
tetrahedra in $\dS^3$ and $\AdS^3$ with spacelike, timelike, or null faces.
{The input consists of four nontrivial based $\SO^+(1,2)$ holonomies
satisfying the closure relation. We construct their Gram data and give a
global condition ensuring that the reconstructed tetrahedron lies
within the required geometric domain. This condition is strict copositivity
of the signed inverse Gram form on the outward branch. Together with
nondegeneracy of the Gram data, it guarantees a unique finite tetrahedron up to
ambient isometry, with exactly the prescribed Levi--Civita face holonomies.
We express this criterion as finitely many polynomial inequalities in
trace coordinates, leading to an identification of the finite tetrahedra and a subset of the relative
$\SL(2,\mathbb R)$ character varieties of the four-holed sphere. We also
discuss the degenerate Gram strata, vector closure in the
curvature-zero limit, and the projective dual tetrahedra, which include ideal and hyperideal tetrahedral sectors. The results provide classical geometric input for quantizing
Lorentzian curved tetrahedra and for quantum gravity models with a nonzero
cosmological constant.}
\end{abstract}

\maketitle

\tableofcontents

\section{Introduction}

\medskip

The classical theorem of Minkowski is one of the basic reconstruction results of
convex geometry. If a collection of nonzero vectors
$\boldsymbol{A}_{i}=A_i\boldsymbol{n}_i\in\mathbb R^3$, with
$A_i>0$ and $\boldsymbol{n}_i$ pairwise distinct unit normals, spans $\mathbb R^3$ and satisfies
$$
  \sum_i \boldsymbol{A}_i=0\,,
$$
then these vectors are the outward area-normal vectors of a convex Euclidean
polyhedron, unique up to translation \cite{Minkowski:1897,Schneider:2014}.
For a tetrahedron, this statement converts a finite set of algebraic data into a
complete geometry.
Lorentzian Minkowski problems and convex polyhedra in flat and
constant-curvature Lorentzian space forms have been studied in
\cite{Schlenker:2001,Fillastre:2015,BonsanteFillastre2017}.
In discrete gravity, such area-normal vectors are the classical face
fluxes, and closure is the local constraint expressing their compatibility
as the boundary of a polyhedral cell. With a nonzero cosmological constant,
constant-curvature {complexes} provide  a corresponding discretization of
geometry \cite{BahrDittrich2010}. The reconstruction of these cells from
finite parallel transports is therefore a natural prerequisite for their
use as gravitational data. %
{In this paper, we restrict attention to tetrahedra, the
three-dimensional simplices that serve as elementary cells in
simplicial discretizations of geometry. }

At nonzero constant curvature, however, area-normal vectors at different
faces do not naturally live in a single vector space. Comparing them requires
parallel transport, and the additive closure relation
must be replaced by a relation between based holonomies.
The natural finite
curvature datum attached to a triangular face is instead its Levi--Civita
holonomy.
A nontrivial face holonomy determines an invariant normal line and a
stabilizer element. In the elliptic and hyperbolic sectors, the latter is
described by a rotation angle or a boost parameter. In the parabolic sector,
the based holonomy determines a nilpotent logarithm vector.
In the compact real form, closing holonomies
were used by Haggard, Han, and Riello to reconstruct spherical and hyperbolic
tetrahedra \cite{Haggard:2015ima}.
The corresponding compact phase space and its quantization have been
studied in \cite{HanHsiaoPan:2024,HsiaoPan:2025}.
The Lorentzian problem involves the noncompact real form
$\SL(2,\mathbb R)$ and all three causal types of faces. It also requires a
global test ensuring that the reconstructed vertices bound a finite
tetrahedron in the selected de Sitter or anti-de Sitter domain.

The same group-valued closure relation also has a moduli-space meaning. Four
closing based holonomies, modulo a common change of frame, are precisely
boundary monodromy data for a flat connection on a four-holed sphere. After
the boundary conjugacy classes are fixed, the natural parameter space is a
relative $\PSL(2,\mathbb R)$ or $\SL(2,\mathbb R)$ character variety.
This connects the reconstruction problem with the Hamiltonian
description of Chern--Simons theory and with its role in three-dimensional
gravity \cite{AchucarroTownsend1986,Witten1988Gravity,Witten:1989}.
 { Note that the flat connection on the four-holed
sphere encodes their closure relation, not the vanishing of the curvature
of the tetrahedron.}
The question is then which points of this phase space describe finite
Lorentzian tetrahedra.

This paper answers both the pointwise reconstruction question and the
recognition question on the character variety. We prove a holonomy
Minkowski theorem for finite tetrahedra in $\dS^3$ and $\AdS^3$, including
spacelike, timelike, and null faces. We then express its hypotheses as
finite trace-polynomial tests on the relative character surface. These
tests distinguish the two curved reconstruction loci and identify the
geometric information that survives on their degenerate strata.

We use the unit-radius quadrics
$\mathcal M_{+1}=\dS^3\subset\mathbb R^{1,3}$ and
$\mathcal M_{-1}=\AdS^3\subset\mathbb R^{2,2}$.
Both have tangent signature $(1,2)$, so their based face holonomies lie in
the same group $\SO^+(1,2)$, with spin cover $\SL(2,\mathbb R)$. 
 Precisely speaking, the curved tetrahedra considered in our reconstruction
theorem are finite quadric sections of simplicial cones whose every
nonzero ray has the sign of the corresponding quadric
(Definition~\ref{def:generalized_tetrahedron}). {Their} faces are
ordinary triangular patches of totally geodesic supporting surfaces.
Null faces are allowed, but null vertex rays and through-infinity or
two-sheeted branches are outside the reconstruction theorem, compared to \cite{Haggard:2015ima}.

\medskip
{\it 1. Lorentzian holonomy data and the first main result.}

Fix an abstract tetrahedron with labeled vertices
$V_1,\ldots,V_4$. Choose a common base vertex and a distinguished edge to base the data of the face opposite that vertex.
The orientations and basing of the four face-boundary loops
follow the simple-path convention of
Definition~\ref{def:simple_path_convention}. The four resulting face loops
give the based closure relation for face holonomies
$$
  O_4O_3O_2O_1=\id\,,
  \qquad O_i\in\SO^+(1\,,2)\,.
$$
Every nontrivial $O_i$ has an invariant normal line. Choose a representative
$n_i$ on that line, with squared norm $-1$ or $+1$ when it is non-null.
Elliptic, hyperbolic, and parabolic holonomies correspond, respectively, to
spacelike, timelike, and null faces. For a parabolic holonomy write its
logarithm vector as $\kappa_i=\Theta_i n_i$. The rescaling
$n_i\mapsto c_i n_i$, $\Theta_i\mapsto\Theta_i/c_i$, with $c_i\neq0$,
leaves the holonomy unchanged.

The simple-path convention determines a symmetric reconstructed
matrix $G$ { from the invariant normal representatives $n_i$, with the prescribed parallel transports inserted in their pairings.} It is constructed from the holonomies
before an ambient tetrahedron is known to exist. Four ordered Lorentzian
triple products $\chi_i$ of normal representatives $n_i$, with the same transport convention, determine
the normal signs. Their squares satisfy
$\det G_{\hat i}=-\chi_i^2$, where $G_{\hat i}$ is obtained by deleting
row and column $i$. When these minors are nonzero, fixing global chirality
allows a unique choice of signs for which every $\chi_i$ is positive.

\begin{theorem}[Theorem~\ref{theorem:unified_minkowski}]
Let $O_1,\ldots,O_4\in\SO^+(1,2)$ be nontrivial and satisfy
$O_4O_3O_2O_1=\id$. Fix the simple-path convention and suppose that
$G$ and every $G_{\hat i}$ are nonsingular. Select the normal signs so that
every $\chi_i>0$, and denote the resulting matrix again by $G$.
If {
\be  \lambda^\top\bigl[-\sgnop(\det G)G^{-1}\bigr]\lambda>0
  \qquad\text{for every }0\neq\lambda\in\mathbb R_{\geq0}^4\,,
  \label{eq:copositivity}
\ee}
the data determine a unique finite strictly convex tetrahedron in
$\mathcal M_{\sigma_G}$, up to ambient isometry, where
$$
  \sigma_G=-\sgnop(\det G)\,.
$$
Its outward face-normal Gram matrix is $G$, and its based Levi--Civita
face holonomies are exactly the prescribed $O_i$.
\end{theorem}

The inequality \eqref{eq:copositivity} is called strict copositivity of the signed
inverse Gram form. Its geometric meaning is that every nonzero nonnegative
combination of the four selected vertex rays stays in the required
quadric domain. 
We therefore also call this the domain condition required for the reconstruction. The determinant sign selects $\dS^3$ when $\det G<0$
and $\AdS^3$ when $\det G>0$. The principal minors ensure four finite
vertices, while the triple products select outward representatives.
Neither of these tests guarantees the domain condition.
Lemma
\ref{lemma:finite_domain_admissibility_criterion} makes the domain
condition effective by reducing it to tests on the six edges, four faces,
and interior of the coefficient simplex.

The proof first factors $G$ in the selected ambient signature.
Strict copositivity gives a finite tetrahedron, and stabilizer-parameter
uniqueness establishes that its based face holonomies equal the input
tuple, completing the reconstruction.

For noncentral spin lifts $H_i\in\SL(2,\mathbb R)$ with
$H_4H_3H_2H_1=\epsilon\id$, $\epsilon=\pm1$, the projections satisfy the
same criterion (Corollary~\ref{cor:SL_reconstruction}). This central sign
is independent of the outward-normal branch. A consistent choice of spin lifts for all edge transports realizes the prescribed tuple when $\epsilon=+1$. For
$\epsilon=-1$, changing one central sign produces a closing tuple with
unchanged projected geometry. Central holonomies do not determine normal
lines and are excluded, whereas nontrivial parabolic holonomies are included.

\medskip
{\it {2.} Character varieties and the second main result.}

After choosing closing spin lifts, the relation
$H_4H_3H_2H_1=\id$ is precisely the relation among the boundary
holonomies of a flat $\SL(2,\mathbb R)$ connection on a four-holed
sphere. This observation allows us to formulate the reconstruction
problem within the moduli space of flat connections. The second
main result identifies, using trace coordinates, which points of
this moduli space describe finite curved tetrahedra.

Let $S_{0,4}$ be an oriented four-holed sphere, with boundary loops $\gamma_i$
ordered so that $\gamma_4\circ\gamma_3\circ\gamma_2\circ\gamma_1=1$. For four prescribed
noncentral conjugacy classes $\boldsymbol C=(C_1,\ldots,C_4)$ in
$\SL(2,\mathbb R)$, consider
$$
  \mathcal X_{\boldsymbol C}
  :=\left\{(H_1\,,\ldots\,,H_4)\in\prod_{i=1}^4 C_i
  \,\middle|\,H_4H_3H_2H_1=\id\right\}
  \big/\SL(2\,,\mathbb R)\,,
$$
where the quotient is by simultaneous conjugation.   Its adapted trace coordinates are
$$
  t_i:=\trspin(H_i)\,,\qquad
  x:=\trspin(H_2H_1)\,,\quad
  y:=\trspin(H_3H_2)\,,\quad
  z:=\trspin(H_1H_3)\,.
$$
They satisfy the Fricke equation $\mathcal F_{\mathbf t}(x,y,z)=0$ of
\eqref{eq:adapted_fricke_cubic}. It is easy to see that this relative character
variety is generically two-dimensional.  Real points of the Fricke surface need not arise from $\mathrm{SL}(2,\mathbb R)$ representations. We therefore restrict to the trace-coordinate image
$\mathcal Y$ of a connected
component of $\mathcal X_{\boldsymbol C}$.

Connected traces give the polynomial matrix
$\mathcal G_{\mathbf t}(x,y,z)$ of
\eqref{eq:character_connected_gram}, including the exceptional transported
pairing. It is diagonally congruent to the reconstructed Gram matrix $G$, even
when some faces are null. Define
$$
  \Delta_{\mathbf t}:=\det\mathcal G_{\mathbf t}\,,
  \qquad
  m_i:=\det(\mathcal G_{\mathbf t})_{\hat i}\,.
$$
To express the outward domain condition polynomially, set
$$
  P_i:=-\frac{\partial\mathcal F_{\mathbf t}}{\partial t_i}\,,\qquad
D_P:=\operatorname{diag}(P_1\,,\ldots\,,P_4)\,,\qquad
  A_{\mathbf t}:=-D_P\operatorname{adj}(\mathcal G_{\mathbf t})D_P\,.
$$
The $P_i$ encode the transported oriented triple products, and the
factors $D_P$ incorporate the outward branch choice into $A_{\mathbf t}$
(Lemma~\ref{lemma:character_trace_triples}).
On the Fricke surface, $m_i=-4P_i^2$. When $\Delta_{\mathbf t}\neq0$
and every $m_i\neq0$, strict copositivity of $A_{\mathbf t}$ is equivalent
to the domain condition in the first theorem.

\begin{theorem}[Theorem~\ref{theorem:character_geometric_chambers}]
For $\sigma\in\{+1,-1\}$, the finite nondegenerate reconstruction locus
in $\mathcal Y$ for $\mathcal M_\sigma$ is
$$
\mathcal X_\sigma:=\bigl\{p\in\mathcal Y\,\big|\,
 {-\sigma\Delta_{\mathbf t}(p)>0}\,,\quad
 m_i(p)\neq0\ \text{for every }i\,,\quad A_{\mathbf t}(p)\ \text{strictly copositive}\bigr\}\,.
$$
Every point determines a finite strictly convex tetrahedron uniquely up to
ambient isometry. The locus
$$
  \mathcal X_{\rm flat}
  :=\{p\in\mathcal Y\mid\Delta_{\mathbf t}(p)=0\,,\quad
  m_i(p)\neq0\ \text{for every }i\}
$$
determines a labeled tetrahedron in flat Minkowski space
$\mathbb R^{1,2}$, with scale fixed by the face area-vector data,
uniquely up to Lorentzian isometries and translations.
These loci may be empty. Their defining conditions are invariant under
closing-lift sign changes and descend to the corresponding relative
$\PSL(2,\mathbb R)$ character locus.
\end{theorem}

Proposition~\ref{prop:character_polynomial_domain_criterion} replaces
strict copositivity in this statement by finitely many polynomial
inequalities, with simplifications determined by the elliptic, hyperbolic,
or parabolic boundary types. The complementary
strata retain more limited geometric information. More precisely, for
$\Delta_{\mathbf t}\neq0$, the equation $m_i=0$ identifies a null candidate
vertex line{,} while on $\Delta_{\mathbf t}=0$, vanishing principal minors $m_i$ 
distinguish deficient rank-three normal relations from lower-rank collapse.
This gives a full geometrical  classification of the reconstructed Gram data.

These geometric loci are different from the topological and dynamical
decompositions of character varieties. Benedetto--Goldman classified the
topology of nonsingular real Fricke surfaces
\cite{BenedettoGoldman1999}, while Maloni--Palesi--Tan studied the
mapping-class-group action and domains of discontinuity
\cite{MaloniPalesiTan2015}. Our conditions instead recognize a finite
Lorentzian tetrahedron within a selected character component. They do not
assert that each reconstruction locus is connected, or classify how its
components meet as the boundary conjugacy classes vary.

The following table summarizes the geometric roles of the holonomy and
Gram data in the curved reconstruction theorem. In the last row, $G$ is
computed from outward normal representatives. 
\begin{center}
\renewcommand{\arraystretch}{1.2}
\begin{tabular}{@{}ll@{}}
\hline
\parbox[t]{0.40\linewidth}{\raggedright\strut\textbf{Holonomy or Gram datum}\strut}
& \parbox[t]{0.54\linewidth}{\raggedright\strut\textbf{Geometric meaning}\strut}\\
\hline
\parbox[t]{0.40\linewidth}{\raggedright\strut Elliptic, hyperbolic, or parabolic boundary holonomy\strut}
& \parbox[t]{0.54\linewidth}{\raggedright\strut Spacelike, timelike, or null face, respectively\strut}\\[0.35ex]
\parbox[t]{0.40\linewidth}{\raggedright\strut $\det G<0$ or $\det G>0$\strut}
& \parbox[t]{0.54\linewidth}{\raggedright\strut {Ambient model:} $\dS^3$ or $\AdS^3$, respectively\strut}\\[0.35ex]
\parbox[t]{0.40\linewidth}{\raggedright\strut $\det G_{\hat i}\neq0$ for every $i$\strut}
& \parbox[t]{0.54\linewidth}{\raggedright\strut Four real, finite, nondegenerate candidate vertices\strut}\\[0.35ex]
\parbox[t]{0.40\linewidth}{\raggedright\strut Common sign of the $\chi_i$\strut}
& \parbox[t]{0.54\linewidth}{\raggedright\strut {Selection of outward normal representatives}\strut}\\[0.35ex]
\parbox[t]{0.40\linewidth}{\raggedright\strut $-\sgnop(\det G)G^{-1}$ strictly copositive\strut}
& \parbox[t]{0.54\linewidth}{\raggedright\strut {Every nonzero ray of the selected simplicial cone stays
in the required quadric domain}\strut}\\
\hline
\end{tabular}
\end{center}

\medskip
{\it 3. Flat limits and projective duality.}

Restoring the curvature radius $R$ makes the relation with vector
closure explicit. 
As
$R\to\infty$, assume that the associated geometrical data and the based frames have regular limits, as specified in
Section~\ref{sec:flat_limit}.
The limiting face vector $\mathbf E_i^0$ is the integral of the oriented
vector-valued area form. This definition includes null faces without
requiring a unit null normal. The exact face-holonomy formula gives
$$
  O_i(R)=\id-\frac{\sigma}{R^2}J(\mathbf E_i^0)+o(R^{-2})\,,
  \qquad
  \sum_{i=1}^4\mathbf E_i^0=0\,,
$$
where $J:\mathbb R^{1,2}\to\mathfrak{so}(1,2)$ is the standard
identification with infinitesimal Lorentz transformations. If every triple
of limiting face vectors is independent, the classical Minkowski theorem,
applied as in Section~\ref{sec:flat_limit}, reconstructs the bounded affine
Lorentzian tetrahedron. The limiting normal Gram matrix has rank three.
This geometric contraction requires $O_i(R)\to\id$ and generally varies
the boundary conjugacy classes. By contrast, the construction on
$\mathcal X_{\rm flat}$ uses rank-three Gram data and does not realize
fixed noncentral holonomies as flat Levi--Civita face holonomies.

Projective duality provides a second interpretation of the causal sectors.
The projective lines $[N_i]$ of the extrinsic face normals are the vertices
of a projective dual tetrahedron, defined by polarity with respect to
the ambient bilinear form. Null faces give ideal dual vertices, so an
all-null tetrahedron has an ideal dual tetrahedron.
This parallels the lightlike/ideal correspondence studied in
\cite{MeusburgerScarinci2022}. In the all-timelike
$\AdS^3$ sector, the dual tetrahedron has four hyperideal vertex lines. It is a
standard hyperideal AdS tetrahedron precisely when unit representatives
$X_i$ can be chosen coherently with
$(X_i,X_j)_{-}<-1$ for every $i\neq j$
(Proposition~\ref{prop:edge_adapted_hyperideal_polar_locus}). Conversely,
the dual of a standard nondegenerate hyperideal AdS tetrahedron is finite
and domain-contained
(Proposition~\ref{cor:standard_hyperideal_polar_copositivity}). This duality identifies a proper subset of the finite all-timelike
AdS sector. We also give the comparison
with compact spherical and hyperbolic reconstruction in \cite{Haggard:2015ima} from changing the noncompact
$\SL(2,\mathbb R)$ holonomies in our case to compact
$\SU(2)$ holonomies therein.

\medskip
{\it 4. Geometric and physical scope.}

On its smooth locus represented by irreducible flat connections,
$\mathcal X_{\boldsymbol C}$ carries the Goldman--Atiyah--Bott symplectic
structure with fixed boundary conjugacy classes
\cite{AtiyahBott:1983,Goldman:1984,GuruprasadHuebschmannJeffreyWeinstein1997}.
It is the reduced phase space of $\SL(2,\mathbb R)$ Chern--Simons theory
on the four-holed sphere. Trace functions of curves are Wilson-loop
observables. The reconstruction formulas express the causal types,
dihedral pairings, curvature sign, and domain condition in this observable
algebra. They therefore identify which classical boundary states have a
finite tetrahedral interpretation, rather than merely satisfying group
closure.
A related curvature-dependent quantization of moduli spaces of
three-dimensional gravity was constructed in \cite{KimScarinci2024}.

This gives a concrete starting point for quantizing Lorentzian curved
tetrahedra. Combinatorial quantization and quantum trace constructions
provide noncommutative counterparts of curve observables
\cite{AlekseevGrosseSchomerus1995,BonahonWong2011}. The polynomial Gram
and domain expressions derived here are classical symbols that such a
construction should reproduce.
For quantum states describing finite tetrahedra, the spectral values of the relevant geometric operators should, in the semiclassical regime, satisfy the reconstruction conditions, including the copositivity inequalities.
 Implementing this restriction, together
with operator ordering and domains in noncompact representations, is a
further problem. The present results establish the classical geometric
criterion for a quantum state-space construction, {if one exists}.

There is also a connection with quantum Teichm\"uller theory and
tetrahedral amplitudes. Changes of pants decomposition are described by
fusion kernels and modular-double $b$-$6j$ symbols
\cite{PonsotTeschner2001,NidaievTeschner2013}, whose semiclassical
geometry includes hyperbolic tetrahedra
\cite{TeschnerVartanov:2014}. Related structures enter Virasoro TQFT
formulations of three-dimensional gravity
\cite{CollierEberhardtZhang2023,CollierEberhardtZhang2024}.
Recent asymptotic results also involve truncated hyperideal AdS tetrahedra
\cite{Liu:2025tzv} and decorated ideal AdS tetrahedra \cite{liu2026mathrm}. On the coherent hyperideal dual locus identified
above, the present theorem supplies finite Lorentzian tetrahedra dual to
that geometric class. It suggests studying the volume and co-volume of the dual
tetrahedron as functions of the reconstructed trace data.

Finally, four-holed-sphere flat connections occur as tetrahedral
boundary data in Chern--Simons formulations of four-dimensional spinfoam models with a
cosmological constant \cite{Han:2021tzw,Han:2025mkc}. The Lorentzian
criterion developed here addresses the local reconstruction problem when
such boundary data lie in the split real form. Its treatment of spacelike,
timelike, and null faces provides the geometric input for considering
causal sectors beyond purely spacelike tetrahedral boundaries. A full
spinfoam application must additionally impose gluing, shape matching, and
the four-dimensional critical-point equations.

\medskip

The paper is organized as follows. Sections~\ref{sec:unified_notation}
and~\ref{sec:tetrahedra} fix the ambient and tangent conventions and define
the tetrahedral geometry. Section~\ref{sec:holonomy_data} constructs the
Gram and triple-product data from based holonomies.
Section~\ref{sec:orientation_convexity} establishes the model-selection,
outward-sign, and finite domain criteria, and
Section~\ref{sec:main_reconstruction} proves exact holonomy reconstruction.
Section~\ref{sec:flat_limit} derives the curvature-zero contraction.
Section~\ref{sec:character_varieties_flat_connections} gives the
trace-polynomial reconstruction loci and analyzes null candidate vertex
lines and the determinant wall. Section~\ref{sec:dual_tetra} develops
projective duality and the comparison of real forms.
Section~\ref{sec:conclusion_discussion} discusses further geometric and
physical questions. 

\section{Unified Ambient and Tangent Conventions}
\label{sec:unified_notation}

We write $\sigma=+1$ for the de Sitter case and $\sigma=-1$ for the anti-de Sitter case.  The ambient model is
\be
\begin{array}{c|c|c|c|c}
\sigma & \Msig & \Vsig & \etasig & \Osig\\
\hline
+1 & \dS^3 & \mathbb R^{1,3}
& \mathrm{diag}(-1,1,1,1) & \mathrm O(1,3)\\
-1 & \AdS^3 & \mathbb R^{2,2}
& \mathrm{diag}(-1,-1,1,1) & \mathrm O(2,2)
\end{array}
\ee
with ambient inner product
\be
\ambip{X}{Y}:=X^\top\etasig Y\,,
\qquad
\Msig:=\{X\in\Vsig\mid \ambip{X}{X}=\sigma\}\,.
\label{eq:ambient_models}
\ee
Equivalently, in explicit form,
\be
\dS^3
:=\left\{
X=(X^0,X^1,X^2,X^3)\in\mathbb R^{1,3}
\ \middle|\ 
-(X^0)^2+(X^1)^2+(X^2)^2+(X^3)^2=+1
\right\}\,,
\label{eq:dS3_def}
\ee
and
\be
\AdS^3
:=\left\{
X=(X^0,X^1,X^2,X^3)\in\mathbb R^{2,2}
\ \middle|\ 
-(X^0)^2-(X^1)^2+(X^2)^2+(X^3)^2=-1
\right\}\,.
\label{eq:AdS3_def}
\ee
The induced metric on each space is the pullback of the corresponding ambient metric.  Throughout, we use the canonical identification $T_v\Vsig\simeq\Vsig$ given by translation, so that $v$ also denotes the radial tangent vector at the point $v$.  For $v\in\Msig$, the vector $v$ itself spans the normal line $\mathbb Rv$ to $\Msig$ at $v$.  We call this the radial direction.  Its squared norm is $\sigma$, so the symbol $\sigma$ is both the label of the geometry and the squared norm of the radial vector.

For $v\in\Msig$,
\be
T_v\Msig=v^\perp=\{Y\in\Vsig\mid \ambip{Y}{v}=0\}\,.
\ee

\begin{lemma}
\label{lemma:tangent_signature}
For every $v\in\Msig$, the induced metric on $T_v\Msig$ has signature $(1,2)$.
\end{lemma}

\begin{proof}
In $\dS^3\subset\mathbb R^{1,3}$ the radial vector has positive norm, so removing the radial line from the ambient signature $(1,3)$ leaves $(1,2)$.  In $\AdS^3\subset\mathbb R^{2,2}$ the radial vector has negative norm, so removing the radial line from $(2,2)$ also leaves $(1,2)$.
\end{proof}

After choosing an oriented pseudo-orthonormal frame in $T_v\Msig$, we identify it with a fixed model vector space 
\be
\mathbb R^{1,2}\,,\qquad
\eta_T=\mathrm{diag}(-1,1,1)\,,\qquad
\tanip{u}{w}:=u^\top\eta_Tw\,.
\ee
Here $\eta_T$ represents the restricted ambient metric on $v^\perp$ in the chosen frame.
The ambient pairing $\ambip{\cdot}{\cdot}$ and the tangent pairing $\tanip{\cdot}{\cdot}$ will always be kept distinct.  The {preceding} signature computation is standard in semi-Riemannian geometry \cite{ONeill:1983}, but we record it because it is the reason the same tangent holonomy group appears in both signs of curvature.

We use the Lorentzian cross product on $\mathbb R^{1,2}$ defined by
\be
(u\times w)^\mu=\eta_T^{\mu\lambda}\varepsilon_{\lambda\nu\rho}u^\nu w^\rho\,,
\qquad
\varepsilon_{012}=+1\,.
\label{eq:cross_product}
\ee
Indices on $\varepsilon$ are raised and lowered with $\eta_T$.
Then
\be
\tanip{u\times w}{z}=\det(u,w,z)
\ee
in every positively oriented pseudo-orthonormal tangent frame.  The identities used below are
\begin{subequations}
\begin{align}
u\times(w\times z)&=z\,\tanip{u}{w}-w\,\tanip{u}{z}\,,\label{eq:vector_triple_identity}\\
\varepsilon_{\mu\nu\rho}\varepsilon_{\mu'}{}^{\nu\rho}&=-2(\eta_T)_{\mu\mu'}\,.\label{eq:epsilon_contraction}
\end{align}
\end{subequations}

The map
\be
J:\mathbb R^{1,2}\longrightarrow \mathfrak{so}(1,2)\,,
\qquad
J(u)w:=u\times w
\label{eq:def-J}
\ee
is an isomorphism of vector spaces.  With the conventions above,
\begin{subequations}
\begin{align}
J(u)^3&=\tanip{u}{u}J(u)\,,\label{eq:J_cube}\\
\trvec\big(J(u)J(w)\big)&=2\tanip{u}{w}\,,\label{eq:trace_JJ}\\
\trvec\big(J(u)J(w)J(z)\big)&=\tanip{u\times w}{z}\,.\label{eq:trace_JJJ}
\end{align}
\end{subequations}
These formulas fix the sign convention for all triple products in the paper.

\section{Tetrahedra, Face Normals, and Gram Matrices}
\label{sec:tetrahedra}

A totally geodesic face in $\Msig$ is the intersection of $\Msig$ with a linear hyperplane through the origin of $\Vsig$.  We write
\be
\Pi_i:=N_i^\perp=\{X\in\Vsig\mid \ambip{X}{N_i}=0\}\,,
\qquad
       \widetilde{\Pi}_i :=\Pi_i\cap\Msig\,,
\ee
where $N_i\in\Vsig$ is a nonzero \emph{extrinsic}, or ambient, face normal. {We call $\widetilde{\Pi}_i$ a supporting surface.} Its causal sign is
\be
\nu_i:=\ambip{N_i}{N_i}\in\{-1,0,+1\}
\ee
after unit normalization in the non-null cases.  A timelike normal gives a spacelike face, a spacelike normal gives a timelike face, and a null normal gives a null face.  For null faces, no unit normalization exists, so one chooses any nonzero representative of the projective null-normal line.  Its inverse rescaling is absorbed by the parabolic stabilizer coordinate.

The corresponding \emph{intrinsic}, or tangent, face normal is denoted by $n_i$.  More precisely, if $p\in {\widetilde{\Pi}_i}$, then
\be
n_i(p):=N_i\in T_p\Msig\,,
\label{eq:intrinsic_extrinsic_normal}
\ee
where the equality means that the constant ambient vector $N_i$ is being viewed as a tangent vector at $p$.  This is legitimate because $\ambip{p}{N_i}=0$ on {$\widetilde{\Pi}_i$}, and hence $N_i\in T_p\Msig=p^\perp$.  The intrinsic and extrinsic squared norms agree:
\be
\tanip{n_i(p)}{n_i(p)}
=\ambip{N_i}{N_i}
=\nu_i\,.
\ee
After a base point $v$ and a path from $p$ to $v$ have been chosen, we write
\be
n_i:=P_{p\to v}\,n_i(p)\in T_v\Msig\simeq\mathbb R^{1,2}
\label{eq:transported_intrinsic_normal}
\ee
for the transported intrinsic normal representative.  Thus $N_i$ is the ambient object used in the four-dimensional Gram matrix, while $n_i$ is the tangent-space object preserved by the based face holonomy.  For faces passing through the base point and evaluated there, this distinction is only notational, while for a face not containing $v$, such as $F_4$ below, the transport path is part of the data.
Here $P_{p\to v}$ denotes Levi-Civita parallel transport along the chosen path from $p$ to $v$.

\begin{definition}[Tetrahedral configurations]
\label{def:tetrahedral_configuration}
A {\it tetrahedral configuration} is a linearly independent ordered quadruple $(N_1,N_2,N_3,N_4)$ of extrinsic face normals, together with the four hyperplanes $\Pi_i=N_i^\perp$ and {supporting surfaces} $\widetilde{\Pi}_i=\Pi_i\cap\Msig$.
Its candidate vertices opposite $\widetilde{\Pi}_i$ are the points of
$\widetilde{\Pi}_j\cap \widetilde{\Pi}_k\cap \widetilde{\Pi}_l$, whenever this intersection is nonempty.
\end{definition}
More explicitly, the three ambient hyperplanes opposite $\widetilde{\Pi}_i$ determine the real
line
\be
L_i:=\Pi_j\cap\Pi_k\cap\Pi_l\,,
\qquad
\{i,j,k,l\}=\{1,2,3,4\}\,.
\ee
Since
\be
\widetilde{\Pi}_j\cap \widetilde{\Pi}_k\cap \widetilde{\Pi}_l=L_i\cap\Msig\,,
\ee
this intersection is nonempty precisely when $L_i$ has the quadric sign
$\sigma$. In that case it is an antipodal pair
\be\label{eq:antipodal_pair}
L_i\cap\Msig=\{V_i^+,V_i^-\}\,.
\ee
A choice of oriented half-spaces selects one member of each pair, denoted
by $V_i$. After choosing four
oriented half-spaces, it may become a tetrahedron in the sense below.

For a {supporting surface} $\widetilde{\Pi}_i=\Pi_i\cap\Msig$ with normal $N_i$, a closed half-space bounded by $\widetilde{\Pi}_i$ means one of the two sign domains
\be
\Pi_i^\varepsilon
:=\left\{X\in\Msig\,\middle|\,\varepsilon\,\ambip{X}{N_i}\geq0\right\}\,,
\qquad
\varepsilon=\pm1\,.
\label{eq:curved_halfspace}
\ee
Its boundary in $\Msig$ is $\widetilde{\Pi}_i$.  Replacing $N_i$ by $-N_i$ exchanges the two choices. Choose the signs of $N_i$ so that the selected half-spaces are given by $\ambip{X}{N_i}\leq0$.

Define the associated ambient cone by
\be
C_N:=
\left\{
X\in\Vsig
\,\middle|\,
\ambip{X}{N_i}\leq0 
\text{ for all }i
\right\}\,.
\label{eq:ambient_cone}
\ee
The choice {of half-spaces} is called {\it domain-admissible} if
\be
\sigma\ambip{X}{X}>0
\qquad
\text{for every nonzero }X\in C_N\,.
\label{eq:domain_admissible}
\ee
Domain admissibility ensures that every positive ray of $C_N$ has a unique normalization in $\Msig$.

\begin{definition}[Tetrahedra]
\label{def:generalized_tetrahedron}

A {\it tetrahedron} in $\Msig$ is a tetrahedral configuration together
with a domain-admissible choice of one closed half-space $\Pi_i^-$ bounded
by each $\widetilde{\Pi}_i$. Its underlying region and face patches are
\be
T:=\bigcap_{i=1}^4\Pi_i^-\,,
\qquad
F_i:=T\cap \widetilde{\Pi}_i\,.
\label{eq:triangular_face_patch}
\ee
The points selected from the candidate antipodal pairs are its vertices
$V_1,\ldots,V_4$.

\end{definition}
Such a choice of outward normals $N_i$ fixes the notation
\be
\Pi_i^-=\left\{X\in\Msig\,\middle|\, 
\ambip{X}{N_i}\leq0
\right\}\,,
\label{eq:halfspace_convention}
\ee
and the opposite vertex satisfies
\be
\ambip{V_i}{N_i}<0
\label{eq:outward_support}
\ee
for a strictly convex tetrahedron.
Here strict convexity means the strict support inequalities above and each $F_i$ being the {face} determined by the other three vertices{. Here,} ``outward'' refers to the supporting covector $\ambip{\cdot}{N_i}$, positive on the exterior side.

\begin{definition}[Ambient Gram matrix]
\label{def:ambient_gram}
The ambient Gram matrix of a tetrahedral configuration is
\be
G_{ij}:=\ambip{N_i}{N_j}\,.
\label{eq:ambient_gram}
\ee
For non-null faces one may take $G_{ii}=\nu_i$.  For null faces $G_{ii}=0$, and the remaining entries depend on the chosen null representatives.
\end{definition}

\begin{definition}[Inertia]
\label{def:inertia}
For a nondegenerate real symmetric matrix $G$, we denote by%
\be
\Inertia(G):=(n_-(G),n_+(G))
\ee
the {numbers} of negative and positive eigenvalues of $G$, counted with multiplicity.  Equivalently, $\Inertia(G)$ is the signature of the corresponding symmetric bilinear form, written in the order ``negative, positive''.  With this convention,
\be
\Inertia(\eta_+)=(1,3)\,,
\qquad
\Inertia(\eta_-)=(2,2)\,.
\ee
\end{definition}

\begin{lemma}
\label{lemma:ambient_gram_reconstruction}
Let $G$ be a real symmetric nondegenerate $4\times4$ matrix.  There exist linearly independent vectors $N_i\in\Vsig$ satisfying $G_{ij}=\ambip{N_i}{N_j}$ if and only if
\be
\Inertia(G)=\Inertia(\etasig)\,.
\label{eq:inertia_condition}
\ee
The reconstructed quadruple is unique up to the global group $\Osig$.
\end{lemma}

\begin{proof}
This is Sylvester's law of inertia \cite{HornJohnson:2013}.  If $N$ is the matrix whose columns are the $N_i$, then $G=N^\top\etasig N$, so $G$ and $\etasig$ have the same inertia.  Conversely, if the inertias agree, a congruence factorization $G=N^\top\etasig N$ exists.  Any two such factorizations differ by a matrix $\Lambda$ with $\Lambda^\top\etasig\Lambda=\etasig$.
\end{proof}

\begin{remark}
For $\dS^3$, the required inertia is $(1,3)$, while for $\AdS^3$, it is $(2,2)$.  Hence $\det G<0$ for a nondegenerate $\dS^3$ tetrahedron and $\det G>0$ for a nondegenerate $\AdS^3$ tetrahedron.  For an arbitrary symmetric matrix, the full inertia is the invariant reconstruction condition.
\end{remark}

\begin{lemma}
\label{lemma:vertex_reality}
Assume that $G$ satisfies the inertia condition \eqref{eq:inertia_condition}, and let $G_{\hat i}$ denote the principal $3\times3$ submatrix obtained by deleting the $i$-th row and column.  For $\{i,j,k,l\}=\{1,2,3,4\}$, set
\be
U_i:=\operatorname{span}(N_j,N_k,N_l)\,.
\ee
If $G_{\hat i}$ is nondegenerate, then $N_j,N_k,N_l$ are independent and the three ambient hyperplanes $\Pi_j,\Pi_k,\Pi_l$ intersect in the real line
\be
L_i:=U_i^\perp\subset\Vsig\,.
\ee
This line gives a finite candidate vertex pair precisely when it contains two vectors $V_i^{\pm}$ satisfying
\be
\ambip{V_i^{\pm}}{V_i^{\pm}}=\sigma\,,
\ee
so that $V_i^{\pm}\in\Msig$.  This happens if and only if
\be
\Inertia(G_{\hat i})=(1,2)\,.
\label{eq:vertex_minor_inertia}
\ee
When this holds for all $i$, the four reconstructed hyperplanes have four candidate vertex pairs on the real manifold $\Msig$, with one point from each pair selected by the half-space orientation.
\end{lemma}

\begin{proof}
The restricted metric on $U_i$ is represented by $G_{\hat i}$.  If $G_{\hat i}$ is nondegenerate, then $U_i$ is three-dimensional and the common intersection of the three hyperplanes is the line $L_i=U_i^\perp$.
If $L_i$ contains a point $V_i\in\Msig$, then $V_i$ is non-null and $L_i$ has the quadric sign $\sigma$. According to \eqref{eq:antipodal_pair}, there is an {antipodal} pair $V_i \in \{V_i^{+}, V_i^{-}\}$. We then have $U_i=V_i^\perp=T_{V_i}\Msig$.  Lemma \ref{lemma:tangent_signature} gives
\be
\Inertia(G_{\hat i})=\Inertia(U_i)=(1,2)\,.
\ee
Conversely, suppose that $\Inertia(G_{\hat i})=(1,2)$.  Then $U_i$ is nondegenerate, so the ambient space splits orthogonally as
\be
\Vsig=U_i\oplus U_i^\perp\,.
\ee
Subtracting inertias from $\Inertia(\Vsig)=(1,3)$ for $\sigma=+1$ and $\Inertia(\Vsig)=(2,2)$ for $\sigma=-1$, the one-dimensional space $U_i^\perp$ is positive in the de Sitter case and negative in the anti-de Sitter case.  Thus any nonzero generator $w_i$ of $U_i^\perp$ has
\be
\sgn\ambip{w_i}{w_i}=\sigma\,.
\ee
After multiplying $w_i$ by a real scalar, we obtain $V_i\in U_i^\perp$ with $\ambip{V_i}{V_i}=\sigma$, hence $V_i\in\Msig$.
\end{proof}

\begin{remark}
\label{remark:wrong_causal_vertex_line}
The phrase ``vertex on $\Msig$'' is meant in this real pseudo-Riemannian sense.  The three linear equations always determine a real ambient line when the corresponding normals are independent, but that line may have the wrong causal sign or be null.  In that case it cannot be normalized to $\ambip{V_i}{V_i}=\sigma$, so it does not meet the real pseudo-sphere $\Msig$. A non-null line of the wrong causal sign admits such a normalization after complexification. A null line remains null even after complexification and never meets the quadric of squared norm $\sigma$.
This happens for the polar dual tetrahedra defined in Section \ref{subsec:polar_dual_ideal_hyperideal}. 
\end{remark}

\section{Holonomy Data and Dihedral Invariants}
\label{sec:holonomy_data}

\subsection{Simple paths and vector holonomies}
\label{sec:based_holonomies}

The holonomy around a face boundary is naturally an automorphism of the tangent space at the chosen base point of that boundary loop.  Since different faces have different convenient base points, the first task is to choose a common basing convention.  This is the purpose of the simple paths below: they turn the four face-boundary holonomies into four linear maps of one and the same tangent space, so that a closure product can be stated.

Label the vertices so that $F_i$ is the face opposite $V_i$.  Thus $F_1,F_2,F_3$ meet at the common base vertex
\be
v:=V_4=F_1\cap F_2\cap F_3\,.
\ee
The fixed labeled {tetrahedra} used in the following path convention are shown in Figure \ref{fig:dS-AdS-tetra}.  The same vertex labels and edge orientations are used for both signs of $\sigma$. The drawings only indicate the de Sitter and anti-de Sitter geometric realizations.

\begin{figure}[h!]
\centering
\begin{minipage}{0.3\textwidth}
\includegraphics[width=0.8\textwidth]{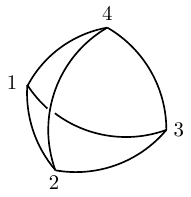}
\subcaption{}{}
\label{fig:dS_tetra}
\end{minipage}
\begin{minipage}{0.3\textwidth}
\includegraphics[width=0.8\textwidth]{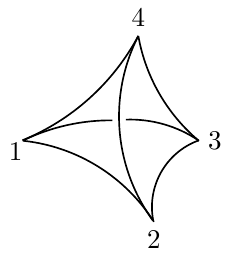}
\subcaption{}{}
\label{fig:AdS_tetra}
\end{minipage}
\caption{{\it (a)} Labeled tetrahedron in $\dS^3$; {\it (b)} labeled tetrahedron in $\AdS^3$. They are illustrated in Euclidean space. }
\label{fig:dS-AdS-tetra}
\end{figure}

For the edge with endpoints $V_a$ and $V_b$, we use the face-intersection notation
\be
E_{ab}:=F_c\cap F_d\,,
\qquad
\{a,b,c,d\}=\{1,2,3,4\}\,.
\ee
 When an orientation is needed, $(ab)$ denotes the oriented edge $E_{ab}$ with source $V_b$ and target $V_a$.  Its Levi-Civita parallel transport is
\be
o_{ab}:T_{V_b}\Msig\longrightarrow T_{V_a}\Msig\,,
\qquad
o_{ba}=o_{ab}^{-1}\,.
\ee
It is a proper, time-orientation-preserving Lorentzian isometry between tangent spaces,
\be
o_{ab}\in\operatorname{Iso}^+\big(T_{V_b}\Msig,T_{V_a}\Msig\big)\,.
\ee
After choosing oriented time-oriented orthonormal frames at the two endpoints, this is represented by a matrix in $\SO^+(1,2)$.  A based loop, whose source and target are both $v$, therefore gives an element of $\SO^+(T_v\Msig)\simeq\SO^+(1,2)$.
Composition is read from right to left.  For example, $(43)\circ(32)\circ(24)$ is the loop based at $V_4$ which successively follows $E_{24}$, $E_{23}$, and $E_{34}$, with the orientations encoded by the parentheses.

The transported intrinsic normal representatives \eqref{eq:transported_intrinsic_normal} used below are fixed as follows.  The normals of $F_1,F_2,F_3$ are evaluated directly at $v$, so $n_i=n_i(v)$ is just the extrinsic normal $N_i$ viewed in $T_v\Msig$ for $i=1,2,3$.  The face $F_4$ does not contain $v$, so we choose $V_2\in F_4$ as its reference point and transport $n_4(V_2)$ to $v$ along the edge
\be
E_{24}=F_1\cap F_3
\ee
with orientation $(42)$:
\be
n_4(v):=o_{42}\,n_4(V_2)\,.
\label{eq:n4_transport}
\ee
We call this distinguished transport edge $E_{24}=F_1\cap F_3$ the \textbf{special edge}.  If $n_4$ is first evaluated at another point of $F_4$, it is transported inside $F_4$ to $V_2$ and then to $v$ along the special edge by \eqref{eq:n4_transport}.  Since $F_4$ is totally geodesic, this produces the same transported intrinsic representative.

\begin{definition}[Simple-path convention]
\label{def:simple_path_convention}
Fix the ordered vertex labels, the base vertex $v=V_4$, and the special edge $E_{24}=F_1\cap F_3$.  The \emph{simple-path convention} is the following common basing of the four face normals and face-boundary holonomies.  The normal representatives are fixed as in \eqref{eq:n4_transport}: $n_1,n_2,n_3$ are evaluated at $v$, while $n_4$ is transported to $v$ along the special edge.  The four based face loops are
\be
\begin{aligned}
\ell_1&:=(43)\circ(32)\circ(24)\,,&
\ell_2&:=(41)\circ(13)\circ(34)\,,&
\ell_3&:=(42)\circ(21)\circ(14)\,,
\\
\ell_4&:=(42)\circ(23)\circ(31)\circ(12)\circ(24)\,.
\end{aligned}
\label{eq:simple_loops}
\ee
Here $\ell_i$ goes once around $\partial F_i$ and is based
at $v$.  For $i=1,2,3$ this is immediate because $F_i$ contains $v$.
For $F_4$, the loop first follows the special edge $E_{24}$ with
orientation $(24)$, runs around $\partial F_4$, and returns along
$E_{24}$ with orientation $(42)$.
\end{definition}

\begin{figure}[h!]
\begin{subfigure}[t]{0.2\linewidth}
\includegraphics{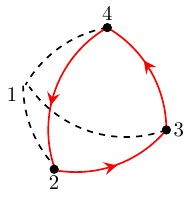}
\caption{$\ell_1$}
\end{subfigure}
\begin{subfigure}[t]{0.2\linewidth}
\includegraphics{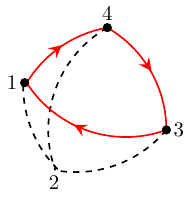}
 \caption{$\ell_2$}
\end{subfigure}
\begin{subfigure}[t]{0.2\linewidth}
\includegraphics{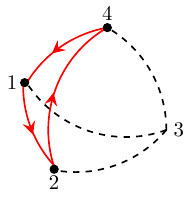}
 \caption{$\ell_3$}
\end{subfigure}
\begin{subfigure}[t]{0.2\linewidth}
\includegraphics{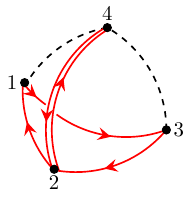}
 \caption{$\ell_4$}
\end{subfigure}
\caption{The set of simple paths ({\it in red}) $\{\ell_1,\ell_2,\ell_3,\ell_4\}$ \eqref{eq:simple_loops} defined with vertex $V_4$ as the base point and the edge $E_{24}=F_1\cap F_3$ as the special edge. They satisfy $\ell_4\circ\ell_3\circ\ell_2\circ\ell_1=1$. Simple paths on a tetrahedron with negative curvature are defined similarly. }
\label{fig:simple_paths}
\end{figure}

The displayed order in Definition \ref{def:simple_path_convention} is part of the fixed labeled {tetraheron} convention in Figure \ref{fig:dS-AdS-tetra}: once the ordered vertex labels, the base vertex $V_4$, and the special edge $E_{24}=F_1\cap F_3$ are chosen, these loops define which boundary orientation is paired with the chosen outward normal by Stokes' convention.  An odd relabeling of the tetrahedron reverses the chirality of this convention.  If the relabeling changes the base vertex or the oriented use of the special edge, for example $V_2\leftrightarrow V_4$, one must first rewrite the simple paths with the new base data rather than compare the old formulas literally.
The combinatorial content of this convention is displayed in Figure \ref{fig:simple_paths}.

Let $\rho$ be the parallel-transport representation of these based loops and set
\be
O_i:=\rho(\ell_i)\in\SO^+(1,2)\,.
\ee
Thus each $O_i$ is now a linear isometry of the same Lorentzian vector space $T_v\Msig$. {We call it the vector holonomy of the face $F_i$.} The reason for introducing these based holonomies is that the four oriented face loops satisfy one path relation, and this relation becomes the holonomy closure condition used in the reconstruction theorem.
Equivalently,
\be
O_1=o_{43}o_{32}o_{24}\,,\quad 
O_2=o_{41}o_{13}o_{34}\,,\quad 
O_3=o_{42}o_{21}o_{14}\,,\quad
O_4=o_{42}o_{23}o_{31}o_{12}o_{24}\,.
\label{eq:Oi_simple_paths}
\ee
The ordering in \eqref{eq:simple_loops} is chosen so that the four oriented face boundaries close around the tetrahedron:
\be
\ell_4\circ\ell_3\circ\ell_2\circ\ell_1=1\,.
\ee
Applying $\rho$ gives the based closure relation
\be
O_4O_3O_2O_1=\id\,,
\qquad
O_i\in\SO^+(1,2)\,.
\label{eq:SO_closure}
\ee

We next recall the standard elliptic, hyperbolic, and parabolic classification of nontrivial elements of $\SO^+(1,2)\cong\PSL(2,\mathbb R)$ \cite{Beardon:1983}, adapted to the face-normal convention above.  For later use, we call the one-dimensional coordinate along a chosen connected stabilizer subgroup the {\it stabilizer coordinate}.
For a normalized non-null representative, this is the {rotation} angle or boost rapidity.  
In contrast, a parabolic holonomy determines the scaled null generator, which is the invariant product of this coefficient and the null representative, rather than either factor separately.

\begin{lemma}
\label{lemma:SO_axis}
Let $F_i\subset \widetilde{\Pi}_i$ be an
oriented {face} in the totally geodesic support surface $\widetilde{\Pi}_i$
and let $O_i\in\SO^+(1,2)$ be its nontrivial
based boundary holonomy.  Choose a nonzero transported intrinsic normal representative $n_i\in T_v\Msig$, defined from the extrinsic normal $N_i$ by \eqref{eq:intrinsic_extrinsic_normal}--\eqref{eq:transported_intrinsic_normal}, and normalize it when it is non-null.  Then $O_i n_i=n_i$ and
\be
O_i=\exp\big(\Theta_iJ(n_i)\big)\,,
\label{eq:SO_axis_exp}
\ee
where the stabilizer coordinate belongs to
{\be
\Theta_i\in
\begin{cases}
\mathbb R/2\pi\mathbb Z\,, & \tanip{n_i}{n_i}=-1\,,\\[0.08cm]
\mathbb R\,, & \tanip{n_i}{n_i}=+1\ \hbox{or}\ 0\,.
\end{cases}
\label{eq:stabilizer_coordinate_domains}
\ee}
The three cases are respectively elliptic, hyperbolic, and parabolic.
\end{lemma}

\begin{proof}
Write $\widetilde{\Pi}_i=\Msig\cap N_i^\perp$.  For $p\in \widetilde{\Pi}_i$,
\be
\ambip{N_i}{p}=0\,.
\ee
Thus the constant ambient vector $N_i$ is tangent to $\Msig$ along the
support surface and is orthogonal to $T\widetilde{\Pi}_i$.  Since its ambient derivative
vanishes, its Levi-Civita derivative on $\Msig$ also vanishes after tangent
projection.  Hence $n_i$ is parallel along $\partial F_i$.  After
applying the simple-path basing convention, parallel transport around the
face loop fixes the transported representative:
\be
O_in_i=n_i\,.
\ee

It remains to identify the connected stabilizer of the representative $n_i$.  Since $J$ is identified with the Lorentzian cross product \eqref{eq:def-J}, an infinitesimal stabilizer of $n$ is of the form $J(u)$ with
\be
J(u)n=u\times n=0\,.
\ee
For every nonzero $n$, including a null vector, this forces $u\in\mathbb Rn$, so the connected stabilizer has Lie algebra $\mathbb RJ(n)$.  If $\tanip{n}{n}=-1$, it acts as rotations on the spacelike plane $n^\perp$, and the parameter is an angle modulo $2\pi$.  If $\tanip{n}{n}=+1$, it acts as boosts on the Lorentzian plane $n^\perp$, and the parameter is a real rapidity.  If $\tanip{n}{n}=0$, its nontrivial elements are parabolic and the parameter is real.  Therefore,
\be
O_i=\exp\big(\Theta_iJ(n_i)\big)\,.
\ee
\end{proof}

For a based face holonomy $O_i$, let $\mathcal N_i\subset T_v\Msig$ be its invariant normal line,
{\be
\mathcal N_i=\mathbb R n_i\,.
\ee}
The representative and branch ambiguities in this choice are collected in Remark \ref{remark:branches}.

Using \eqref{eq:J_cube}, the closed forms of \eqref{eq:SO_axis_exp} are
\be
O_i=\id+r_iJ(n_i)+\widetilde r_iJ(n_i)^2\,,
\qquad
(r_i,\widetilde r_i)=
\begin{cases}
\big(\sin\Theta_i,1-\cos\Theta_i\big)\,, & \tanip{n_i}{n_i}=-1\,,\\[0.12cm]
\big(\sinh\Theta_i,\cosh\Theta_i-1\big)\,, & \tanip{n_i}{n_i}=+1\,,\\[0.12cm]
{\left(\Theta_i,\frac12\Theta_i^2\right)}\,, & \tanip{n_i}{n_i}=0\,.
\end{cases}
\label{eq:SO_exp_closed}
\ee
Thus $r_i$ is the coefficient of the infinitesimal generator $J(n_i)$ in the vector holonomy.  Consequently,
\be
\frac{O_i-O_i^{-1}}{2}
=r_iJ(n_i)\,,
\label{eq:Oi_minus_inverse}
\ee
and
\be
\trvec(O_i)=
\begin{cases}
1+2\cos\Theta_i\,, & \tanip{n_i}{n_i}=-1\,,\\
1+2\cosh\Theta_i\,, & \tanip{n_i}{n_i}=+1\,,\\
3\,, & \tanip{n_i}{n_i}=0\,.
\end{cases}
\label{eq:trace_Oi}
\ee
Thus elliptic holonomies ($\trvec(O_i)<3$) correspond to spacelike faces, hyperbolic holonomies ($\trvec(O_i)>3$) to timelike faces, and parabolic holonomies ($\trvec(O_i)=3$) to null faces.

The classification above is algebraic.  When the holonomy is the actual Levi-Civita holonomy of a tetrahedron, the same stabilizer coordinate has a geometric meaning.  The next lemma explains this meaning: the angle, rapidity, or coefficient relative to a chosen null representative is the coefficient of the based area-normal flux in the stabilizer algebra.  This identifies the geometric content of the stabilizer coordinate once a tetrahedron and an oriented face branch are already given.

Let $T\subset\Msig$ be a tetrahedron with the simple-path convention fixed.  For its $i$-th face, set
\be
q_i=
\begin{cases}
v\,, & i=1,2,3\,,\\
V_2\,, & i=4\,,
\end{cases}
\ee
and let $Q_i:T_{q_i}\Msig\to T_v\Msig$ be the identity for $i=1,2,3$ and the special-edge transport for $i=4$.

$\mathcal{M}_\sigma$ is equipped with a Lorentzian co-triad field $e^I=e^I_\mu dx^\mu$.  On the selected face, we use the frame based at $q_i$ in which the chosen nonzero face normal has constant components. 
 Define the $\R^{1,2}$-valued area-normal two-form
\be
\mathcal S^I:=\frac12\varepsilon^I{}_{JK}e^J\wedge e^K\,.
\label{eq:area_normal_two_form}
\ee
Here $\varepsilon_{012}=+1$, with internal indices raised and lowered by $\eta_T$.
The flux vector of the {face} $F_i$ in the
frame based at $q_i$ is
\be
{\bf E}_{i}(q_i)^{I} =
\int_{F_i}\mathcal S^I\,.
\label{eq:face_flux_qi}
\ee
We base the flux at $v$ by
\be
\mathbf E_i:=Q_i\mathbf E_i(q_i)\in T_v\Msig\,.
\label{eq:based_face_flux}
\ee

For a non-null face $F_i$, let $\mathfrak n_i(q_i)$ be the intrinsic outward-pointing unit normal at $q_i$, let $\nu_i=\tanip{\mathfrak n_i(q_i)}{\mathfrak n_i(q_i)}=\pm1$, and let $\mathfrak n_i=Q_i\mathfrak n_i(q_i)$ be the based outward unit normal.  The oriented area form $\rd A_i^{\rm or}$ is the scalar two-form on $F_i$ defined by
\be
\rd A_i^{\rm or}:=\tanip{\mathcal S}{\mathfrak n_i(q_i)}
=
\eta_{IL}\mathcal S^I\mathfrak n_i(q_i)^L\,.
\label{eq:oriented_area_form_def}
\ee
Its orientation is the one induced by the chosen oriented boundary loop of $F_i$.

We have the following lemma.

\begin{lemma}
\label{lemma:stabilizer_parameter_area}
Let $O_i\in\SO^+(1,2)$ be the based Levi-Civita holonomy around the oriented boundary of $F_i$.  
With the based flux defined above, we have
\be
O_i=\exp\big(-\sigma J(\mathbf E_i)\big)\,.
\label{eq:face_holonomy_vector_area}
\ee
If $F_i$ is non-null, then
\be
\mathbf E_i=\mathfrak a_i \mathfrak n_i \,,\quad \mathfrak a_i:=\nu_i\int_{F_i}\rd A_i^{\rm or}\,,
\label{eq:theta_signed_holonomy_area}
\ee 
where $\rd A_i^{\rm or}$ is defined by \eqref{eq:oriented_area_form_def}.  If the representative used in \eqref{eq:SO_axis_exp} is $n_i=\epsilon_i\mathfrak n_i$, $\epsilon_i=\pm1$, then the stabilizer coordinate satisfies
\be
\Theta_i\equiv -\epsilon_i\sigma\mathfrak a_i\,,
\label{eq:theta_flux_meaning}
\ee
with ordinary equality for hyperbolic holonomies and equality modulo $2\pi$ for elliptic holonomies.
		 
If $F_i$ is null, there is no unit normal and no canonical scalar area decomposition.  
 For the chosen nonzero representative $n_i$, there is a unique real coefficient $\Theta_i$ such that
\be
\sigma\mathbf E_i=-\Theta_i n_i\,.
\ee

\end{lemma}

\begin{proof}
Let $\omega^I{}_{J}$ be the Levi-Civita spin connection in the co-triad $e^I$, with $D=\rd+\omega$ and parallel transport $\mathcal P\exp(-\int\omega)$, and let
\be
R^I{}_{J}=\rd\omega^I{}_{J}+\omega^I{}_{K}\wedge\omega^K{}_{J}
\ee
be its curvature two-form.  Since $\mathcal M_\sigma$ has constant sectional curvature $\sigma$, the Gauss equation gives
\be
R(X,Y)Z
=
\sigma\big(\tanip{Y}{Z}X-\tanip{X}{Z}Y\big)\,,\quad X,Y,Z\in T\mathcal M_{\sigma}\,.
\label{eq:constant_curvature_operator}
\ee
In co-triad and spin-connection notation, this becomes
\be
R^I{}_{J}=\sigma e^I\wedge e_J
\quad \Longrightarrow\quad
R^I:=\frac12{\varepsilon^I{}_{JK}}R^{JK}
=
\frac{\sigma}{2}\varepsilon^I{}_{JK}e^J\wedge e^K\,.
\label{eq:curvature_two_form_constant}
\ee
With the convention above, $R^{\mathcal M_\sigma}=\sigma J(\mathcal S)$.
Strictly speaking, the holonomy of a connection is a path-ordered
exponential of the {negative of} connection one-form around the boundary.
Because $F_i$ is simply connected and its chosen nonzero
normal vector $n_i$ is parallel along a totally geodesic
face, choose a trivialization in which $n_i$ is constant. Then $D n_i=0$ implies that the restricted connection takes values in the one-dimensional fixed-vector algebra $\mathfrak{stab}(n_i)$, which
is abelian. For a null face, fixing the vector rather than only its line
is essential. Hence path ordering drops out, and the ordinary Stokes
formula applies.
Using the identification $J:\mathbb R^{1,2}\to\mathfrak{so}(1,2)$ and the equivariance $Q_iJ(u)Q_i^{-1}=J(Q_iu)$, the based boundary holonomy is, therefore,
\be
O_i
=
\exp\left(
{-}Q_i\left(\int_{F_i}
R^{\mathcal M_\sigma}\right)Q_i^{-1}
\right)
=
\exp\big(-\sigma J(\mathbf E_i)\big)\,.
\label{eq:face_holonomy_curvature_integral}
\ee
This gives \eqref{eq:face_holonomy_vector_area}.

For a non-null face, the vector-valued two-form \eqref{eq:area_normal_two_form} is normal to $F_i$.  Hence, on $F_i$, it has the form $\mathcal S^I=c\,\mathfrak n_i(q_i)^I$ for some scalar two-form $c$.  Contracting the internal vector index with the unit normal gives
\be
\tanip{\mathcal S}{\mathfrak n_i(q_i)}
=
c\,\tanip{\mathfrak n_i(q_i)}{\mathfrak n_i(q_i)}
=
c\,\nu_i\,.
\ee
By \eqref{eq:oriented_area_form_def}, this contraction is $\rd A_i^{\rm or}$, so $c=\nu_i\rd A_i^{\rm or}$.  Therefore,
\be
\mathcal S^I
=
\nu_i\,\mathfrak n_i(q_i)^I\,\rd A_i^{\rm or}
\qquad \hbox{on }F_i\,.
\ee
Integrating gives
\be
\mathbf E_i(q_i)
=
\nu_i\left(\int_{F_i}\rd A_i^{\rm or}\right)\mathfrak n_i(q_i)\,.
\ee
Applying $Q_i$ and comparing with $\mathbf E_i=\mathfrak a_i\mathfrak n_i$ proves
\be
\mathfrak a_i=\nu_i\int_{F_i}\rd A_i^{\rm or}\,.
\ee
Since
\be
O_i=\exp\big(-\sigma\mathfrak a_iJ(\mathfrak n_i)\big)
=\exp\big(-(\epsilon_i\sigma\mathfrak a_i)J(n_i)\big)\,,
\ee
comparison with \eqref{eq:SO_axis_exp} gives \eqref{eq:theta_flux_meaning}, with the usual $2\pi$ ambiguity in the elliptic case.

The null case is the same flux construction without a unit normalization of the normal line.
\end{proof}
We call $\mathfrak a_i$ defined in \eqref{eq:theta_signed_holonomy_area} the {\it algebraic oriented area} of the non-null triangle $F_i$.

\begin{remark}%
\label{remark:branches}
Lemma \ref{lemma:SO_axis} uses one symbol $n_i$ for a chosen nonzero representative of the invariant normal line in every causal case.  The vector holonomy alone does not canonically determine the outward representative on that line.  For a hyperbolic holonomy,
\be
\exp\big(\Theta J(n)\big)=\exp\big((- \Theta)J(-n)\big)\,,
\ee
while for an elliptic holonomy with principal angle $0<\Theta<2\pi$,
\be
\exp\big(\Theta J(n)\big)=\exp\big((2\pi-\Theta)J(-n)\big)\,.
\ee
For a null representative there is, in addition, no unit normalization.  The replacement
{\be
n_i\longmapsto c_i n_i\,,
\qquad
\Theta_i\longmapsto c_i^{-1}\Theta_i\,,
\qquad
c_i\in\mathbb R^\times\,,
\label{eq:null_rescaling}
\ee}
leaves the parabolic holonomy unchanged.  Equivalently, the holonomy fixes the nilpotent logarithm vector
{\be
\kappa_i:=\Theta_i n_i\,,
\qquad
O_i=\exp\big(J(\kappa_i)\big)\,,
\qquad
\tanip{\kappa_i}{\kappa_i}=0\,,
\label{eq:null_logarithm_vector}
\ee}
but not its decomposition into a null representative and a stabilizer coordinate.  The special case $c_i=-1$ is precisely the sign reversal of the null representative.

Thus the outward representative is a global tetrahedral branch choice for every causal type, fixed by the simultaneous convexity condition in Section \ref{sec:orientation_convexity} {below}, not by the positivity of one stabilizer coordinate.  Once this branch is fixed, Lemma \ref{lemma:stabilizer_parameter_area} gives the geometric meaning of that coordinate.
\end{remark}

\subsection{\texorpdfstring{$\SL(2,\mathbb R)$}{SL(2,R)} spin lifts}
\label{sec:spin_lifts}

The previous section described the based face holonomies as elements of the common tangent Lorentz group $\SO^+(1,2)$.  For trace formulas it is useful to pass to its spin double cover
\be
\Spin^+(1,2)\cong\SL(2,\mathbb R)\,.
\ee
Since the holonomies act on the Lorentzian tangent space $T_v\Msig$, and this tangent space has signature $(1,2)$ for both $\dS^3$ and $\AdS^3$, the same spin cover and the same $\mathfrak{sl}(2,\mathbb R)$ formulas apply in both cases.  

After choosing an oriented time-oriented frame at $v$, we identify $T_v\Msig$ with the fixed reference Lorentzian vector space $(\mathbb R^{1,2},\eta_T)$, and then identify $\mathbb R^{1,2}$ with $\mathfrak{sl}(2,\mathbb R)$.  We choose the generators
\be
\tau_0=\frac12\mat{cc}{0&1\\-1&0}\,,
\qquad
\tau_1=\frac12\mat{cc}{0&1\\1&0}\,,
\qquad
\tau_2=\frac12\mat{cc}{1&0\\0&-1}\,,
\ee
which satisfy
\be
\trspin(\tau_\mu\tau_\nu)=\frac12(\eta_T)_{\mu\nu}\,,
\qquad
\trspin(\tau_\mu\tau_\nu\tau_\rho)=\frac14\varepsilon_{\mu\nu\rho}\,,
\label{eq:tau_traces}
\ee
where the Levi-Civita symbol satisfies the contraction \eqref{eq:epsilon_contraction}. 
For $x=x^\mu e_\mu\in\mathbb R^{1,2}$, set
\be
\mathcal T(x):=x^\mu\tau_\mu\in\mathfrak{sl}(2,\mathbb R)\,.
\label{eq:T_map}
\ee
This is a component contraction with the chosen generators. %
The causal type of $x$ determines the real Cartan type of the generator $\mathcal T(x)$:
\be
\mathcal T(x)^2=\frac14\tanip{x}{x}\id\,.
\label{eq:T_square}
\ee
Thus a timelike unit vector gives a compact Cartan generator, conjugate to $\pm\tau_0$, and its exponential is elliptic.  A spacelike unit vector gives a split Cartan generator, conjugate to $\tau_2$, and its exponential is hyperbolic.  A null vector gives a nilpotent generator, namely a nonzero traceless matrix $\mathcal T(n)$ with $\mathcal T(n)^2=0$, and hence a parabolic exponential.  
The covering map $\pi:\SL(2,\mathbb R)\to\SO^+(1,2)$ sends the spin exponential
\be
H_i=\epsilon_i \exp\big(\Theta_i\mathcal T(n_i)\big)
\label{eq:spin_exp}
\ee
to \eqref{eq:SO_axis_exp}, where $\epsilon_i=\pm1$ is the central ambiguity of $\SL(2,\mathbb R)\to\SO^+(1,2)$.  Using \eqref{eq:T_square}, its closed form for all three causal types is
\be
H_i=\epsilon_i\lb c_i\id+s_i\mathcal T(n_i)\rb\,,
\qquad
(c_i,s_i)=
\begin{cases}
\left(\cos\frac{\Theta_i}{2},\,2\sin\frac{\Theta_i}{2}\right)\,,
& \tanip{n_i}{n_i}=-1\,,\\[0.15cm]
\left(\cosh\frac{\Theta_i}{2},\,2\sinh\frac{\Theta_i}{2}\right)\,,
& \tanip{n_i}{n_i}=+1\,,\\[0.15cm]
{\left(1,\,\Theta_i\right)}\,,
& \tanip{n_i}{n_i}=0\,.
\end{cases}
\label{eq:spin_closed_form}
\ee
Thus $\epsilon_is_i$ is the coefficient of $\mathcal T(n_i)$ in the spin lift of every causal case.  For a parabolic holonomy, $s_i=\Theta_i$ and \eqref{eq:null_logarithm_vector} gives $s_in_i=\kappa_i$.  Hence the traceless part fixes $\epsilon_i\kappa_i$, while the separation into a null representative and a stabilizer coordinate retains the rescaling freedom \eqref{eq:null_rescaling}.

The spin version contains no new geometric input.  Once a lift of the based Levi-Civita holonomy \eqref{eq:face_holonomy_vector_area} is chosen, the corresponding spin holonomy is
\be
H_i=\epsilon_i\exp\big(-\sigma\mathcal T(\mathbf E_i)\big)\,.
\ee

\subsection{Dihedral invariants from based normals}
\label{sec:dihedral_invariants}

At the base vertex $v=V_4$, the extrinsic normals $N_1,N_2,N_3$ are tangent to $\Msig$ and therefore give the intrinsic representatives $n_1,n_2,n_3\in T_v\Msig$ directly.  The fourth extrinsic normal $N_4$ is first viewed as $n_4(V_2)\in T_{V_2}\Msig$ and then represented at $v$ by the special-edge transport \eqref{eq:n4_transport}.  For most adjacent face pairs, this chosen transport already computes the correct ambient inner product.  The only exceptional entry is the pair $(F_2,F_4)$, whose common edge is not compatible with the special-edge representative of $n_4$ without inserting a face holonomy.

For an already existing tetrahedron, the next proposition proves that the following tangent-space formula gives the ambient Gram matrix.  In the converse direction of the reconstruction theorem as we will see in Section \ref{sec:main_reconstruction}, the same formula will be used first as the {\it definition} of the reconstructed Gram matrix from holonomy data.  Only after Lemmas \ref{lemma:lorentzian_gram_signature}, \ref{lemma:ambient_gram_reconstruction}, and \ref{lemma:vertex_reality} are applied is it interpreted as the Gram matrix of ambient normals.

\begin{prop}
\label{prop:dihedral_equivalence}
Let $n_i\in T_v\Msig$ be the transported intrinsic normal representatives associated with a tetrahedron in $\Msig$.  With the simple-path convention of Definition \ref{def:simple_path_convention}, the ambient Gram matrix is recovered by
\be
G_{ij}=
\begin{cases}
\tanip{n_i}{n_j}\,, & (i,j)\neq(2,4),(4,2)\,,\\[0.12cm]
\tanip{n_2}{O_1n_4}
=\tanip{n_2}{O_3^{-1}n_4}\,, & (i,j)=(2,4),(4,2)\,.
\end{cases}
\label{eq:reconstructed_gram}
\ee
If $H_i\in\SL(2,\mathbb R)$ are spin lifts with $\pi(H_i)=O_i$, their action on tangent vectors is the adjoint action transported through $\mathcal T$:
\be
H_i\cdot x
:=\mathcal T^{-1}\!\left(H_i\,\mathcal T(x)\,H_i^{-1}\right)
=O_ix\,,
\qquad x\in T_v\Msig\simeq\mathbb R^{1,2}\,.
\label{eq:spin_adjoint_vector_action}
\ee
Thus the exceptional entry may equivalently be written as
\be
G_{24}=G_{42}
=\tanip{n_2}{H_1\cdot n_4}
=\tanip{n_2}{H_3^{-1}\cdot n_4}\,,
\label{eq:reconstructed_gram-H}
\ee
while all nonexceptional entries are unchanged.
\end{prop}
\begin{proof}
For $i,j\in\{1,2,3\}$, both faces meet at the base vertex $v$.  The extrinsic normals $N_i,N_j$ are already tangent vectors at $v$, namely $n_i,n_j$, so the ambient pairing restricts to the tangent pairing:
\be
\ambip{N_i}{N_j}=\tanip{n_i}{n_j}\,.
\ee

Next consider a pair involving $F_4$ and one of the two faces sharing the special edge, namely $F_1$ or $F_3$.  At the reference point $V_2\in F_4$, both relevant normals are tangent to $\Msig$.  Transporting them to $v$ along the same special edge preserves the Lorentzian inner product, so again
\be\label{eq:metric_compatibility_13}
\ambip{N_i}{N_4}=\tanip{n_i}{n_4}\,,
\qquad i=1,3\,.
\ee
The same argument also gives the diagonal entries.

The only remaining off-diagonal entry is the adjacent pair $(F_2,F_4)$.  Their common edge is not the special edge used to define $n_4$.  To compare the two normals at the endpoint $V_3$ of $F_2\cap F_4$, one transports $n_4(V_2)$ first along $(32)$ and then transports the result to $v$ along $(43)$.  Relative to the special-edge representative $n_4=o_{42}n_4(V_2)$, this is exactly
\be
o_{43}o_{32}n_4(V_2)
=o_{43}o_{32}o_{24}n_4
=O_1n_4\,.
\ee
Metric compatibility therefore gives
\be
\ambip{N_2}{N_4}
=\tanip{n_2}{O_1n_4}\,.
\ee
Doing the same comparison through the other endpoint $V_1$ gives the equivalent expression $\tanip{n_2}{O_3^{-1}n_4}$.  Finally, \eqref{eq:spin_adjoint_vector_action} is precisely the vector action of the projected holonomy $\pi(H_i)=O_i$, and the central sign of a spin lift cancels in the conjugation.  This proves the spin-lift form as well.
\end{proof}

For vector holonomies, the trace formulas are most uniform if one uses the scaled Lie-algebra representative
\be
a_i:=r_in_i\,,
\qquad
A_i:=\frac{O_i-O_i^{-1}}{2}=J(a_i)\,,
\label{eq:vector_Ai_scaled}
\ee
where $r_i$ is defined in \eqref{eq:SO_exp_closed}.
Note that this antisymmetric part is useful only when $r_i\neq0$.  In particular, for an elliptic half-turn one has $\Theta_i=\pi$ modulo $2\pi$, hence $r_i=\sin\Theta_i=0$ and
\be
O_i=O_i^{-1}\,,
\qquad
\frac{O_i-O_i^{-1}}{2}=0\,.
\ee
The invariant normal line is then still well defined, namely $\ker(O_i-\id)$, but it cannot be extracted from \eqref{eq:Oi_minus_inverse}.  The reconstruction theorem below is, therefore, formulated in terms of invariant normal lines, not only in terms of the antisymmetric part.

When this is not the case, we have
\begin{align}
\frac12\trvec(A_iA_j)&=\tanip{a_i}{a_j}\,,\label{eq:SO_connected_pair}\\
\trvec(A_iA_jA_k)&=\tanip{a_i\times a_j}{a_k}\,.
\label{eq:SO_connected_triple}
\end{align}
For the exceptional Gram entry $G_{24}$, the normal $n_4$ is replaced by $O_1n_4$ or equivalently by $O_3^{-1}n_4$.  Since $A_4=J(a_4)$ is an endomorphism of $T_v\Msig$, the same transport acts on $A_4$ by conjugation.  Equivalently, the conjugated holonomy $O_1O_4O_1^{-1}$ has invariant normal line $\mathbb R(O_1n_4)$.  The corresponding Lie-algebra generators are therefore
\be
A_{4|1}:=\frac{O_1O_4O_1^{-1}-(O_1O_4O_1^{-1})^{-1}}{2}
=O_1A_4O_1^{-1}
=J(O_1a_4)\,,
\label{eq:vector_conj_A41}
\ee
and
\be
A_{4|3^{-1}}:=O_3^{-1}A_4O_3
=J(O_3^{-1}a_4)\,.
\label{eq:vector_conj_A4_3inv}
\ee
The conjugation appears because $J$ is equivariant under the orientation-preserving Lorentz action: $O\,J(a)\,O^{-1}=J(Oa)$.  Equations \eqref{eq:SO_connected_pair} and \eqref{eq:SO_connected_triple} then apply without further change, using $A_{4|1}$ or $A_{4|3^{-1}}$ in place of $A_4$ for the exceptional insertion.  Whenever the relevant factors $r_i$ are nonzero, the Gram entries of the chosen representatives $n_i$ are recovered by dividing by the appropriate products of $r_i$.  For a null face, $r_i=\Theta_i$ and the trace formulas before this division compute invariants of the logarithm-fixed vector $\kappa_i=\Theta_i n_i$.

For spin holonomies, similarly write the traceless part
\be
B_i:=\impart(H_i):=H_i-\frac12\trspin(H_i)\id=\mathcal T(b_i)\,,
\qquad
b_i:=
\epsilon_i s_in_i\,,
\label{eq:spin_Bi_scaled}
\ee
Here $s_i$ is the coefficient introduced in \eqref{eq:spin_closed_form}.  

For any integer $m\geq2$, define the connected spin trace of a word by taking the ordinary trace of the corresponding traceless parts:
\be
\big\langle H_{i_1}\cdots H_{i_m}\big\rangle_{C,\SL}
:=\frac12\trspin\big(\impart(H_{i_1})\cdots\impart(H_{i_m})\big)
=\frac12\trspin\big(B_{i_1}\cdots B_{i_m}\big)\,.
\label{eq:SL_connected_general}
\ee
In particular, for $m=2$ and $m=3$,
\begin{align}
\langle H_iH_j\rangle_{C,\SL}
&=\frac12\trspin\big(\mathcal T(b_i)\mathcal T(b_j)\big)
=\frac{1}{4}\tanip{b_i}{b_j}\,,\label{eq:SL_connected_pair}\\
\langle H_iH_jH_k\rangle_{C,\SL}
&=\frac12\trspin\big(\mathcal T(b_i)\mathcal T(b_j)\mathcal T(b_k)\big)
=\frac{1}{8}\tanip{b_i\times b_j}{b_k}\,.
\label{eq:SL_connected_triple}
\end{align}
In the last line, we have used \eqref{eq:tau_traces} and the cross-product convention \eqref{eq:cross_product}. Explicitly, $\trspin(\mathcal T(u)\mathcal T(v)\mathcal T(w))=\frac14\tanip{u\times v}{w}$.

For non-null normalized representatives, $b_i=\epsilon_i s_in_i$ and
\be
\langle H_i^2\rangle_{C,\SL}
=\frac{s_i^2}{4}\tanip{n_i}{n_i}\,.
\ee
Recall $\nu_i:=\tanip{n_i}{n_i}\in\{\pm1\}$ and let 
\be
\rho_i:=2\sqrt{\nu_i\langle H_i^2\rangle_{C,\SL}}=|s_i|\,.
\ee
Then
\be
n_i=\pm_i\,\frac{b_i}{\rho_i}\,,
\label{eq:branch_sign_extraction}
\ee
where the signs $\pm_i$ are the branch choices in extracting the normalized representative from the traceless part.
Equivalently, if a parametrization $H_i=\epsilon_i\bigl[c_i\id+s_i\mathcal T(n_i)\bigr]$ has already been chosen, then $\pm_i=\sgnop(\epsilon_i s_i)$.  Thus the signs are chosen once for each face and then used consistently in all Gram entries and triple products.  They are not independent signs attached to individual formulas.  With these branches, the normalized inner products are equivalently
\begin{align}
\tanip{n_i}{n_j}
&=
\frac{4\,\pm_i\pm_j\,\langle H_iH_j\rangle_{C,\SL}}
{\rho_i\rho_j}\,,\label{eq:SL_normalized_pair}\\
\tanip{n_i\times n_j}{n_k}
&=
\frac{8\,\pm_i\pm_j\pm_k\,\langle H_iH_jH_k\rangle_{C,\SL}}
{\rho_i\rho_j\rho_k}\,.
\label{eq:SL_normalized_triple}
\end{align}
These are the normalized non-null versions of \eqref{eq:SL_connected_pair}--\eqref{eq:SL_connected_triple}.  The global outward criterion for fixing the signs is given in Section \ref{sec:orientation_convexity}.  The normalized formulas are not used for null faces, because both unit normalization and the displayed denominators degenerate in the parabolic case.
For the exceptional insertion, define the conjugated spin lifts and their traceless parts by
\be
H_{4|1}:=H_1H_4H_1^{-1}\,,
\qquad
B_{4|1}:=\impart(H_{4|1})=H_1B_4H_1^{-1}
=\mathcal T(H_1\cdot b_4)\,,
\label{eq:spin_conj_B41}
\ee
and similarly
\be
H_{4|3^{-1}}:=H_3^{-1}H_4H_3\,,
\qquad
B_{4|3^{-1}}:=\impart(H_{4|3^{-1}})
=H_3^{-1}B_4H_3
=\mathcal T\big((H_3^{-1})\cdot b_4\big)\,.
\label{eq:spin_conj_B4_3inv}
\ee
The exceptional spin-trace expressions are obtained from \eqref{eq:SL_connected_pair} and \eqref{eq:SL_connected_triple} by replacing $B_4$ with $B_{4|1}$, or equivalently with $B_{4|3^{-1}}$.  For non-null normalized representatives one divides by products of $s_i$.  For null representatives one divides by the chosen nonzero parabolic coordinates $\Theta_i$, together with the spin-lift signs where appropriate.  Equivalently, without making these choices, the formulas directly compute invariants of the logarithm-fixed vectors $\kappa_i$.

\section{Orientation, Chirality, and Convexity}
\label{sec:orientation_convexity}

This section supplies the orientation and global convexity information not contained in the Gram matrix alone.  The vertex triple products relate the principal-minor conditions to the ambient model and select the outward representatives of the four normal lines.  The support-sign and inverse-Gram criteria then determine whether the selected half-spaces form a finite strictly convex tetrahedron.  Configurations that fail the strict domain condition are treated as obstructions rather than as tetrahedra.  These results provide the geometric input for the reconstruction theorem of Section \ref{sec:main_reconstruction}.

The Gram matrix reconstructs the ambient hyperplanes but not the outward branch.  The missing information is orientation-sensitive.  At the vertex opposite $F_i$, let $\{j,k,l\}=\{1,2,3,4\}\setminus\{i\}$ in the cyclic order induced by the chosen boundary orientation, and define
\be
\chi_i(V_i):=\tanip{n_j(V_i)\times n_k(V_i)}{n_l(V_i)}\,.
\label{eq:chi_at_vertex}
\ee
Parallel transport preserves the Lorentzian scalar triple product.  With the base point $v=V_4$ and the same simple-path convention, we write simply $\chi_i=\chi_i(v)$, with
\be
\begin{aligned}
\chi_1&=\tanip{n_3\times n_2}{O_3^{-1}n_4}\,,\\
\chi_2&=\tanip{n_1\times n_3}{n_4}\,,\\
\chi_3&=\tanip{n_2\times n_1}{O_1n_4}\,,\\
\chi_4&=\tanip{n_1\times n_2}{n_3}\,.
\end{aligned}
\label{eq:transported_chi}
\ee
These quantities also provide the operational rule for the branch signs $\pm_i$ in \eqref{eq:branch_sign_extraction}.  Start with any preliminary choices of the four signs.  If the preliminary branch of face $F_p$ is changed by a sign $\delta_p\in\{\pm1\}$, then $\chi_q$ changes exactly when $p\neq q$. 

The transformation law only tells us how the candidate orientation data react to branch flips.  It does not yet explain why a common sign should encode the outward branch.  The next lemma records the signature information already forced by the reconstructed Lorentzian vertex data.  Then Proposition \ref{prop:det_chi} relates the three-dimensional triple product at a vertex to an ambient invariant that is independent of which vertex is used.

\begin{lemma}%
\label{lemma:lorentzian_gram_signature}
Let $G$ be the reconstructed Gram matrix \eqref{eq:reconstructed_gram}, and let $G_{\hat i}$ be the principal submatrix obtained by deleting the $i$-th row and column.  With the simple-path convention,
\be
\det G_{\hat i}=-\chi_i^2\,.
\label{eq:principal_minor_chi}
\ee
Consequently,
\be
\det G_{\hat i}\neq0
\Longleftrightarrow
\chi_i\neq0\,,
\ee
and when this holds,
\be
\Inertia(G_{\hat i})=(1,2)\,,
\qquad i=1,\ldots,4\,.
\ee
If in addition $G$ is nondegenerate and every $G_{\hat i}$ is nondegenerate, then
\be
\Inertia(G)\in\{(1,3),(2,2)\}\,.
\ee
Equivalently,
\be
\det G<0\Longleftrightarrow \Inertia(G)=(1,3)\,,
\qquad
\det G>0\Longleftrightarrow \Inertia(G)=(2,2)\,.
\label{eq:det_sign_inertia}
\ee
\end{lemma}

\begin{proof}
For each $i$, take the three transported intrinsic normal representatives entering the definition of $\chi_i$ in \eqref{eq:transported_chi}.  By the algebraic definition of the reconstructed Gram matrix in \eqref{eq:reconstructed_gram}, including the same exceptional transport insertion for the pair $(F_2,F_4)$, their Lorentzian Gram matrix is exactly $G_{\hat i}$.  If $M_i$ is the $3\times3$ matrix whose columns are these three tangent vectors in an oriented pseudo-orthonormal frame, then
\be
G_{\hat i}=M_i^\top\eta_TM_i\,.
\ee
Taking determinants gives
\be
\det G_{\hat i}
=\det(\eta_T)\,\det(M_i)^2
=-\chi_i^2\,,
\ee
because $\det(\eta_T)=-1$ and $\det(M_i)$ is the Lorentzian scalar triple product defining $\chi_i$, up to the fixed sign already built into the ordering convention.  Thus $\det G_{\hat i}\neq0$ is equivalent to $\chi_i\neq0$.  In that case $M_i$ is invertible, so $G_{\hat i}$ is congruent to $\eta_T$ and has inertia $(1,2)$.

For the final claim, choose any $i$ and let the eigenvalues of $G$ be
\be
\lambda_1\leq\lambda_2\leq\lambda_3\leq\lambda_4\,,
\ee
while those of $G_{\hat i}$ are
\be
\mu_1\leq\mu_2\leq\mu_3\,.
\ee
The Cauchy interlacing theorem \cite{HornJohnson:2013} gives
\be
\lambda_1\leq\mu_1\leq\lambda_2\leq\mu_2
\leq\lambda_3\leq\mu_3\leq\lambda_4\,.
\ee
Since $\Inertia(G_{\hat i})=(1,2)$, we have $\mu_1<0<\mu_2\leq\mu_3$.  The inequalities above imply $\lambda_1<0$, so $G$ has at least one negative eigenvalue, and $\lambda_3>0$, $\lambda_4>0$, so $G$ has at least two positive eigenvalues.  Because $G$ is nondegenerate, the remaining eigenvalue $\lambda_2$ is either negative or positive.  Therefore, the only possibilities are $\Inertia(G)=(2,2)$ and $\Inertia(G)=(1,3)$.  The determinant sign distinguishes them.
\end{proof}

\begin{remark}
This is the Lorentzian analogue of the principal-minor simplification in the $\SO(3)$ curved Minkowski theorem in \cite{Haggard:2015ima}.  The difference is that the tangent metric is $\eta_T$ rather than the Euclidean metric, so the nonzero principal minors are negative instead of positive.
\end{remark}

\begin{prop}
\label{prop:det_chi}
For a tetrahedral configuration in $\Msig$,
\be
\det(N_1,N_2,N_3,N_4)
=\sigma_{\rm chir}\,\ambip{V_i}{N_i}\,\chi_i\,,
\label{eq:det_chi_unified}
\ee
where $\sigma_{\rm chir}\in\{\pm1\}$ is a fixed global sign determined once and for all by the convention relating the ambient orientation of $\Vsig$ to the oriented tangent frames of $\Msig$.
\end{prop}

\begin{proof}
At vertex $V_i$, the three incident normals $N_j,N_k,N_l$ lie in $T_{V_i}\Msig$.  Here $\{j,k,l\}=\{1,2,3,4\}\setminus\{i\}$ with the same cyclic order used in the definition of $\chi_i$.  Decompose the remaining normal into a tangent part and a component along $V_i$:
\be
N_i=N_i^\top+\alpha_i V_i\,,
\qquad
N_i^\top\in T_{V_i}\Msig\,,
\ee
where $N_i^\top$ is orthogonal to $V_i$, not to $N_i$ in general.  Since $\ambip{V_i}{V_i}=\sigma$,
\be
0=\ambip{N_i^\top}{V_i}
=\ambip{N_i-\alpha_iV_i}{V_i}
=\ambip{N_i}{V_i}-\alpha_i\sigma\,,
\qquad
\alpha_i=\sigma\,\ambip{V_i}{N_i}\,.
\ee
Thus
\be
N_i=N_i^\top+\sigma\,\ambip{V_i}{N_i}\,V_i\,.
\ee
The tangent vectors $N_j,N_k,N_l$ span $T_{V_i}\Msig$.  Hence $N_i^\top$ is a linear combination of them, so the determinant obtained by replacing $N_i$ with $N_i^\top$ vanishes.  Therefore, only the component along $V_i$ remains:
\be
\det(N_1,N_2,N_3,N_4)
=\epsilon_i\,\sigma\,\ambip{V_i}{N_i}\,
\det(V_i,N_j,N_k,N_l)\,,
\ee
where $\epsilon_i$ is the sign from moving the $i$-th normal to the displayed position.  Now choose an oriented pseudo-orthonormal frame $(e_0,e_1,e_2)$ of $T_{V_i}\Msig$ compatible with the global chirality convention.  In this frame, the tangent determinant of $N_j,N_k,N_l$ is the Lorentzian scalar triple product $\chi_i$.  The sign comparing $(V_i,e_0,e_1,e_2)$ with the ambient orientation, together with $\epsilon_i$ and the factor $\sigma$, is fixed by the global chirality convention and is denoted by $\sigma_{\rm chir}$.  This gives \eqref{eq:det_chi_unified}.
\end{proof}

\begin{remark}
The sign $\sigma_{\rm chir}$ is global because it records only the chosen orientation convention, not the face label.  The support number $\ambip{V_i}{N_i}$ is geometrical and is not absorbed into this convention.  Thus the signs of the $\chi_i$ are vertex-wise shadows of one ambient determinant, up to the support signs.  If the half-spaces are chosen by the outward convention \eqref{eq:halfspace_convention}, then a strictly convex tetrahedron satisfies $\ambip{V_i}{N_i}<0$ for every opposite vertex by \eqref{eq:outward_support}.  This common negative sign is the input for the outward half-space criterion below.
\end{remark}

\begin{prop}
\label{prop:convex_branch}
Assume that a choice of branches for the four transported intrinsic normal representatives gives a tetrahedral configuration in $\Msig$ in the sense of Definition \ref{def:tetrahedral_configuration}, with nondegenerate ambient Gram matrix.  The chosen representatives are compatible with the outward support branch if and only if
\be
\chi_i\neq0\,,\qquad
\sgnop(\chi_1)=\sgnop(\chi_2)=\sgnop(\chi_3)=\sgnop(\chi_4)\,,
\label{eq:common_chi_sign}
\ee
with the common sign fixed to be positive after choosing the global chirality convention.
\end{prop}

\begin{proof}
By Proposition \ref{prop:det_chi}, the same nonzero ambient determinant is related to each $\chi_i$ by
\be
\ambip{V_i}{N_i}
=\frac{\det(N_1,N_2,N_3,N_4)}
{\sigma_{\rm chir}\chi_i}\,.
\ee
Thus the four $\chi_i$ have one nonzero sign if and only if the four support numbers $\ambip{V_i}{N_i}$ have one nonzero sign.  For $j\neq i$, the vertex $V_j$ lies on the face $F_i$, hence $\ambip{V_j}{N_i}=0$.  The only vertex not on $F_i$ is $V_i$, so the side of $F_i$ entering the tetrahedral region is determined by the sign of $\ambip{V_i}{N_i}$.  With the outward convention \eqref{eq:halfspace_convention}, this is the outward support branch exactly when
\be
\ambip{V_i}{N_i}<0\,,
\qquad i=1,\ldots,4\,.
\ee
The remaining simultaneous sign is fixed by the global chirality convention, so we may state the criterion as positivity of all $\chi_i$ after that convention has been chosen.
\end{proof}

Proposition \ref{prop:convex_branch} identifies the outward side of each face by the signs of the support numbers, but this is still a face-by-face statement.  We also need to know that the four selected sides assemble into the tetrahedral region in Definition \ref{def:generalized_tetrahedron}, and we need an explicit rule for choosing representatives of the invariant normal lines.  Lemma \ref{lemma:simplicial_halfspace_realization} proves the first point, including the necessary global domain-containment condition, while Lemma \ref{lemma:branch_selection} proves the second.

\begin{lemma}
\label{lemma:simplicial_halfspace_realization}
Let $(N_1,\ldots,N_4)$ be a tetrahedral configuration in $\Msig$ in the sense of Definition \ref{def:tetrahedral_configuration} with selected candidate vertices $V_i$.  Suppose the normal signs have been chosen so that
\be
h_i:=\ambip{V_i}{N_i}<0\,.
\label{eq:support_numbers_negative}
\ee
{Assume in addition that the signed inverse Gram form is strictly copositive, namely
\be
\mu^\top(\sigma G^{-1})\mu>0
\qquad
\hbox{for every nonzero } \mu\in\mathbb R_{\geq0}^4\,.
\label{eq:strict_inverse_gram_copositivity}
\ee
}
Then {} $T=C_N\cap\Msig$
is the finite strictly convex tetrahedron selected by the four oriented half-spaces.  Its face $F_i$ is obtained by setting $\lambda_i=0$, hence has vertices $V_j,V_k,V_l$ with $\{i,j,k,l\}=\{1,2,3,4\}$, and the opposite vertex $V_i$ lies in the interior of $\Pi_i^-$. 
\end{lemma}

\begin{proof}
Let $W_i$ be the dual basis to the normals, \ie $\ambip{W_i}{N_j}=\delta_{ij}$. 
Since $V_i$ lies on the three faces $F_j$ with $j\neq i$, we have $\ambip{V_i}{N_j}=0$ for $j\neq i$.  Together with \eqref{eq:support_numbers_negative}, this gives
\be
V_i=h_iW_i\,.
\ee
Thus every $X\in\mathbb V_\sigma$ has the expansion
\be
X=\sum_{i=1}^4\ambip{X}{N_i}\,W_i
=
\sum_{i=1}^4
\frac{\ambip{X}{N_i}}{h_i}\,V_i\,.
\label{eq:dual_vertex_expansion}
\ee
Since $h_i<0$, the inequalities $\ambip{X}{N_i}\leq0$ are equivalent to
\be
\lambda_i:=\frac{\ambip{X}{N_i}}{h_i}\geq0\,.
\ee
Then the ambient cone $C_N$ \eqref{eq:ambient_cone} is exactly the closed simplicial cone generated by the four vertex rays:
\be
C_N
=
\left\{
\sum_{i=1}^4\lambda_iV_i\ \middle|\ \lambda_i\geq0
\right\}\,.
\label{eq:simplicial_cone}
\ee
Intersecting with $\ambip{X}{X}=\sigma$ gives the half-space intersection in \eqref{eq:halfspace_convention}. { One must still verify that every $X\in C_N\setminus\{0\}$ can be uniquely rescaled by a positive factor to a point of $\mathcal M_\sigma$. 
Write
\be
X=\sum_i\lambda_iV_i=-\sum_i\mu_iW_i\,,
\qquad
\mu_i:=-h_i\lambda_i\geq0\,.
\ee
Since $X\neq 0$, we have $\mu\neq 0$. The dual basis $W_i$ has Gram matrix $G^{-1}$, and therefore,
\be
\sigma\ambip{X}{X}
=\mu^\top(\sigma G^{-1})\mu>0
\ee
for every nonzero $X\in C_N$, by \eqref{eq:strict_inverse_gram_copositivity}. 
 Hence every nonzero ray of $C_N$ meets $\Msig$ in exactly one positive radial normalization,
\be
X\longmapsto
\frac{X}{\sqrt{\sigma\ambip{X}{X}}}\,.
\ee
{ Every nonzero ray of $C_N$ has a unique representative $X=\sum_i\lambda_iV_i$ with $\lambda_i\geq0$ and $\sum_i\lambda_i=1$.  Thus the standard compact coefficient simplex parametrizes these rays.} The positive radial normalization above is a continuous bijection from this simplex onto $\text{}T$, and hence a homeomorphism, since $\text{}T\subset\mathbb V_\sigma$ is Hausdorff. Moreover, if $X=\sum_j\lambda_jV_j$, then
\be
 \ambip{X}{N_i}=\lambda_i h_i\,.
\ee
Therefore $X\in {\text{}\Pi_i}$ if and only if $\lambda_i=0$, so $\text{}T\cap \widetilde{\Pi}_i$ is the intersection of $\mathcal M_\sigma$ with the subcone generated by the remaining three vertices.
Finally, $ \ambip{V_i}{N_i}=h_i<0$ ({\it cf.} \eqref{eq:outward_support}), so the opposite vertex $V_i$ lies in the interior of $\Pi_i^-$. Thus $T$ is the asserted finite strictly convex tetrahedron.}
\end{proof}

\begin{remark}%
The strict domain condition depends only on the four oriented normal rays, not on their positive normalizations.  Indeed, if
\be
N_i\longmapsto r_iN_i\,,
\qquad r_i>0\,,
\qquad R:=\operatorname{diag}(r_1,\ldots,r_4)\,,
\ee
then $G\longmapsto RGR$ and
\be
\sigma G^{-1}\longmapsto R^{-1}(\sigma G^{-1})R^{-1}\,.
\ee
Since $R^{-1}$ maps $\mathbb R_{\geq0}^4$ bijectively to itself, copositivity and strict copositivity are unchanged.  This observation is particularly useful for null normals, whose nonzero representatives have no preferred unit normalization.
\end{remark}

The strict copositivity condition
\eqref{eq:strict_inverse_gram_copositivity} is the exact
domain-containment condition, but it is quantified over all nonzero
vectors in $\mathbb R_{\geq0}^4$.  The following lemma replaces it by a finite list of
inequalities computed from the Gram matrix $G$.

\begin{lemma}%
\label{lemma:finite_domain_admissibility_criterion}
Let $G$ be a non-degenerate Gram matrix \eqref{eq:reconstructed_gram} of the chosen outward normal
representatives, with $\det G_{\hat i}\neq0$ for every $i$. Set
\be
\sigma:=-\sgnop(\det G)\,,
\qquad
A:=\sigma G^{-1}\,.
\ee
Then $A_{ii}>0$ for every $i$.  Define the normalized dual-vertex form
\be
B:=D_A^{-1/2}AD_A^{-1/2}\,,
\qquad
D_A:=\diag(A_{11},A_{22},A_{33},A_{44})\,,
\label{eq:normalized_domain_form}
\ee
so that $B_{ii}=1$.  For
$\widehat\ell:=\{i,j,k\}=\{1,2,3,4\}\setminus\{\ell\}$, define
\be
L_\ell:=1+B_{ij}+B_{ik}+B_{jk}\,.
\label{eq:face_domain_scalar}
\ee
Then $A$ is strictly copositive if and only if all three groups of
conditions hold:
\begin{align}
B_{ij}&>-1 &&(1\leq i<j\leq4)\,,
\label{eq:finite_domain_edge_tests}\\
G_{\ell\ell}&<0\quad\hbox{or}\quad L_\ell>0
&&(\ell=1,\ldots,4)\,,
\label{eq:finite_domain_face_tests}\\
\det G&>0\quad\hbox{or}\quad G\nleq0\ \hbox{entrywise}\,.
\label{eq:finite_domain_full_test}
\end{align}
In \eqref{eq:finite_domain_full_test}, $G\nleq0$ means that at least one
entry of $G$ is positive.  In particular, after the coordinate faces have
passed their strict tests, there is no additional full-interior test in
anti-de Sitter signature.
\end{lemma}

\begin{proof}
By the cofactor formula and Lemma
\ref{lemma:lorentzian_gram_signature},
\be
(G^{-1})_{ii}
=
\frac{\det G_{\hat i}}{\det G}\,,
\qquad
\det G_{\hat i}=-\chi_i^2\,.
\ee
Therefore,
\be
A_{ii}
=
-\sgnop(\det G)\frac{\det G_{\hat i}}{\det G}
=
-\frac{\det G_{\hat i}}{|\det G|}
=
\frac{\chi_i^2}{|\det G|}
>0\,,
\ee
where the last inequality follows from
$\det G_{\hat i}\neq0$, equivalently $\chi_i\neq0$.

Let $P:=D_A^{-1/2}$, so that $B=P^{\top}AP$.  Since $P$ is a
diagonal matrix with positive entries, $x\mapsto Px$ is a bijection of
$\mathbb R_{\geq0}^4$, and
\be
x^{\top}Bx=(Px)^{\top}A(Px)\,.
\ee
Hence $A$ is strictly copositive if and only if $B$ is strictly
copositive.  By homogeneity, it suffices to test $B$ on the standard
coefficient simplex
\be
\Delta^3:=
\left\{
x\in\mathbb R_{\geq0}^4
\ \middle|\
\sum_{i=1}^4x_i=1
\right\}\,.
\label{eq:standard_coefficient_simplex}
\ee
Its four vertices are the standard basis vectors $e_1,\ldots,e_4$, and
$e_i^{\top}Be_i=B_{ii}=1$.  On the edge joining $e_i$ and $e_j$, write
\be
x=se_i+te_j\,,
\qquad
s,t\geq0\,,
\qquad
s+t=1\,.
\ee
Then
\be
x^{\top}Bx
=s^2+2B_{ij}st+t^2
=(s-t)^2+2(1+B_{ij})st\,.
\ee
Therefore this edge is strictly positive if and only if
$B_{ij}>-1$, which proves
\eqref{eq:finite_domain_edge_tests}.

For $\widehat\ell=\{i,j,k\}$, let $R_\ell:=B[\widehat\ell,\widehat\ell]$
be the principal submatrix representing the restriction to the triangular
face $x_\ell=0$.  Write
\be
R_\ell=
\begin{pmatrix}
1&a&b\\ a&1&c\\ b&c&1
\end{pmatrix}\,,
\qquad
a=B_{ij}\,,\quad b=B_{ik}\,,\quad c=B_{jk}\,.
\ee
Assume its three edge tests hold, so that $a,b,c>-1$.  The strict
$3\times3$ copositivity criterion \cite{Hadeler:1983} states that
$R_\ell$ is strictly copositive if and only if $L_\ell>0$ or $\det R_\ell>0$, which is equivalent to 
\be
L_\ell+S_\ell>0\,,
\qquad
S_\ell:=\sqrt{2(1+a)(1+b)(1+c)}\,,
\label{eq:three_by_three_strict_copositivity}
\ee
since
\be
L_\ell^2-S_\ell^2=-\det R_\ell\,.
\label{eq:face_scalar_determinant_identity}
\ee
Jacobi's complementary-minor identity and
\eqref{eq:normalized_domain_form} give the exact relation
\be
\det R_\ell
=-\frac{G_{\ell\ell}}
{|\det G|\prod_{p\neq\ell}A_{pp}}\,.
\ee
If $G_{\ell\ell}<0$, then $\det R_\ell>0$, and
\eqref{eq:face_scalar_determinant_identity} gives
$S_\ell>|L_\ell|$.  Hence $L_\ell+S_\ell>0$, so $R_\ell$ is strictly
copositive.  If $G_{\ell\ell}\geq0$, then $|L_\ell|\geq S_\ell$, and
\eqref{eq:three_by_three_strict_copositivity} holds exactly when
$L_\ell>0$.  This proves
\eqref{eq:finite_domain_face_tests}.

After all four face tests hold, every proper principal submatrix of $A$
is strictly copositive.  The Cottle--Habetler--Lemke criterion
\cite[Theorem~3.1]{CottleHabetlerLemke:1970}, together with the nonsingularity of $A$,
then gives
\be
A\ \hbox{is strictly copositive}
\quad\Longleftrightarrow\quad
\det A>0
\quad\hbox{or}\quad
A^{-1}\nleq0\ \hbox{entrywise}\,.
\label{eq:boundary_strict_inverse_test}
\ee
Here $\det A=1/\det G$.  If $\det G>0$, the first alternative holds.
If $\det G<0$, then $\sigma=+1$ and $A^{-1}=G$, so the second alternative
is exactly $G\nleq0$ entrywise.  This proves
\eqref{eq:finite_domain_full_test}.
\end{proof}

The three groups of conditions in Lemma \ref{lemma:finite_domain_admissibility_criterion} correspond to the three nontrivial strata of
the coefficient simplex.  Condition \eqref{eq:finite_domain_edge_tests}
tests its six edges, condition \eqref{eq:finite_domain_face_tests} tests its
four triangular faces after their edges have passed, and condition
\eqref{eq:finite_domain_full_test} excludes a remaining obstruction in the
full interior.
The strict copositivity checklist is summarized below.  Here ``automatic'' means that no additional
inequality needs to be checked.

\begin{center}
\begingroup
\small
\setlength{\tabcolsep}{3pt}
\renewcommand{\arraystretch}{1.25}
\begin{tabular}{@{}p{0.20\linewidth}p{0.18\linewidth}
p{0.29\linewidth}p{0.23\linewidth}@{}}
\hline
\raggedright\textbf{Causal sector}
&
\raggedright\textbf{Edges \eqref{eq:finite_domain_edge_tests}}
&
\raggedright\textbf{Faces \eqref{eq:finite_domain_face_tests}}
&
\raggedright\textbf{Interior \eqref{eq:finite_domain_full_test}}
\tabularnewline
\hline

\raggedright AdS
&
\raggedright Check all six inequalities
&
\raggedright Check $L_\ell>0$ only when $G_{\ell\ell}\geq0$
&
\raggedright Automatic since $\det G>0$
\tabularnewline

\raggedright dS with some $G_{\ell\ell}>0$
&
\raggedright Check all six inequalities
&
\raggedright Check $L_\ell>0$ only when $G_{\ell\ell}\geq0$
&
\raggedright Automatic since $G$ has a positive diagonal entry
\tabularnewline

\raggedright All-elliptic dS
&
\raggedright Automatic
&
\raggedright Automatic
&
\raggedright Require $G_{ij}>0$ for some $i<j$
\tabularnewline

\raggedright All-null dS
&
\raggedright Automatic
&
\raggedright Require $L_\ell>0$ for every $\ell$
&
\raggedright Require $G_{ij}>0$ for some $i<j$
\tabularnewline

\raggedright Mixed elliptic/null dS
&
\raggedright Automatic
&
\raggedright Require $L_\ell>0$ for the null labels
&
\raggedright Require $G_{ij}>0$ for some $i<j$
\tabularnewline

\hline
\end{tabular}
\endgroup
\end{center}

{Appendix~\ref{subsec:coordinate_support_obstructions}
explains how failures of {the }domain condition are located at
candidate vertices, along edges, inside faces, or in the
full interior.}

\begin{lemma}
\label{lemma:branch_selection}
Assume the four transported intrinsic normal representatives are nonzero, of arbitrary causal type, and that the preliminary vertex triple products $\chi_i$ are all nonzero.  Under sign changes $n_i\mapsto\delta_i n_i$, $\delta_i=\pm1$, the transformed triple products are
\be
\chi_i^\delta
=\left(\prod_{p\neq i}\delta_p\right)\chi_i\,.
\label{eq:chi_branch_transform}
\ee
After fixing global chirality, there is a unique choice of signs for which all $\chi_i^\delta$ are positive.  Explicitly, if $s_i=\sgnop(\chi_i)$ and $S=\prod_{p=1}^4s_p$, then
\be
\delta_i=S\,s_i\,.
\label{eq:unique_positive_branch}
\ee
\end{lemma}

\begin{proof}
The transformation law follows because $\chi_i$ contains exactly the three normals incident at $V_i$, namely all normals except $n_i$.  Let $D=\prod_{p=1}^4\delta_p$.  The condition $\sgnop(\chi_i^\delta)=+1$ is equivalent to
\be
\frac{D}{\delta_i}\,s_i=1\,,
\qquad\hbox{hence}\qquad
\delta_i=D\,s_i\,.
\ee
Taking the product over $i$ gives
\be
D=\prod_{i=1}^4\delta_i
=\prod_{i=1}^4(Ds_i)
=D^4\prod_{i=1}^4s_i
=S\,,
\ee
because $D=\pm1$.  Thus $\delta_i=S\,s_i$, and this is the only solution with positive common sign.
\end{proof}

{Positive rescalings of the representatives do not affect the signs of the triple products.  Therefore, after choosing any nonzero scale on each null normal line, Lemma \ref{lemma:branch_selection} applies without change.

\section{Main Reconstruction Theorem}
\label{sec:main_reconstruction}

We state the reconstruction first for
$O_i\in\SO^+(1,2)$ and then for chosen spin lifts
$H_i\in\SL(2,\mathbb R)$.  The former determine the geometry; the latter add
the central lift data.  Both use the same reconstruction criterion.

The triple products now play two roles.  By Lemma \ref{lemma:lorentzian_gram_signature}, their nonvanishing is equivalent to nondegeneracy of the vertex principal submatrices $G_{\hat i}$.  Their common positive sign is not an additional geometric inequality.  Lemma \ref{lemma:branch_selection} shows, for arbitrary causal mixtures, that it is the canonical outward branch choice after fixing global chirality. 
There is nevertheless one separate global inequality.  After the outward representatives have been selected, the signed inverse Gram form must be strictly copositive as in \eqref{eq:strict_inverse_gram_copositivity}.  This says that the whole simplicial cone, not only its four vertex rays, stays inside the projective domain of the reconstructed quadric.
 Once all representatives are fixed, the remaining stabilizer coordinates are fixed below in Lemma \ref{lemma:global_stabilizer_parameter_uniqueness} by closure and the positive endpoint triple products.  For a unit timelike normal, the elliptic stabilizer coordinate lives on the circle $\mathbb R/2\pi\mathbb Z$.  For a unit spacelike normal, or for a fixed nonzero null representative, the hyperbolic or parabolic stabilizer coordinate lives on $\mathbb R$.

By Lemma \ref{lemma:lorentzian_gram_signature}, nondegeneracy of
$G$ and all $G_{\hat i}$ leaves exactly the de Sitter inertia $(1,3)$ or the
anti-de Sitter inertia $(2,2)$.  Hence the model is selected by
\be
\sigma_G:=-\sgnop(\det G)\,.
\ee

\begin{theorem}[Holonomy Minkowski theorem for tetrahedra in $\dS^3$ and $\AdS^3$]
\label{theorem:unified_minkowski}
Let $O_1,O_2,O_3,O_4\in\SO^+(1,2)$ be nontrivial holonomies,
\be
O_i\neq\id\,,
\qquad
i=1,\ldots,4\,,
\label{eq:non-trivial-O}
\ee
and satisfy the based closure relation
\be
O_4O_3O_2O_1=\id\,.
\ee
Assume that:
\begin{enumerate}
\item the simple-path convention, including the choice of special edge, has been fixed.
\item the reconstructed Gram matrix $G$ of \eqref{eq:reconstructed_gram}, computed from the preliminary transported representatives, is nondegenerate, and every vertex principal submatrix $G_{\hat i}$ is nondegenerate.
\item after applying Lemma \ref{lemma:branch_selection}, $G$ denotes the Gram matrix of the resulting outward representatives and
\be
\lambda^\top\bigl[-\sgnop(\det G)G^{-1}\bigr]\lambda>0
\qquad
\hbox{for every nonzero }\lambda\in\mathbb R_{\geq0}^4\,.
\label{eq:theorem_strict_copositivity}
\ee
\end{enumerate}
Set $\sigma_G:=-\sgnop(\det G)$.  Then the holonomy data determine a unique finite strictly convex tetrahedron $T_G\subset\mathcal M_{\sigma_G}$ up to the global ambient group $\mathcal O_{\sigma_G}$.  
The extrinsic face normals of $T_G$ have the Gram matrix $G$, the outward branch is the one selected by \eqref{eq:common_chi_sign}, and the prescribed holonomies are exactly the based Levi-Civita face holonomies of $T_G$:
\be
O_i=\widehat O_i(T_G)\,,
\qquad i=1,\ldots,4\,.
\ee
\end{theorem}

\begin{remark}
The theorem is stated in the direction needed for reconstruction from prescribed holonomies.  Conversely, a finite strictly convex tetrahedron whose based Levi-Civita face holonomies are nontrivial gives such input data.  The sign rule $\sigma=-\sgnop(\det G)$ then recovers whether the ambient model is $\dS^3$ or $\AdS^3$.
\end{remark}

Before proving the theorem, we record the branch uniqueness statement for the stabilizer coordinates in the following lemma.  The elliptic coordinate is treated directly as a point of the circle $\mathbb R/2\pi\mathbb Z$, while the hyperbolic and parabolic coordinates are real.  Thus fixed normal representatives, fixed reconstructed Gram matrix, and the positive outward branch determine the stabilizer coordinates uniquely.

\begin{lemma}
\label{lemma:global_stabilizer_parameter_uniqueness}
Fix nonzero based normal representatives $n_1,\ldots,n_4\in T_v\mathcal M_{\sigma}$, normalized by $\tanip{n_i}{n_i}=\pm1$ in the non-null cases.  Let
\be
I_i:=
\begin{cases}
\mathbb R/2\pi\mathbb Z\,, & \tanip{n_i}{n_i}=-1\,,\\[0.08cm]
\mathbb R\,, & \tanip{n_i}{n_i}=+1\ \hbox{or}\ 0\,,
\end{cases}
\ee
and define
\be
\Omega_i(s_i):=\exp\big(J(s_in_i)\big)\,.
\ee
Let $s,t\in I_1\times I_2\times I_3\times I_4$ satisfy the two closure equations
\be
\Omega_4(s_4)\Omega_3(s_3)\Omega_2(s_2)\Omega_1(s_1)
=
\Omega_4(t_4)\Omega_3(t_3)\Omega_2(t_2)\Omega_1(t_1)
=\id
\ee
and the exceptional scalar equality
\be
\tanip{n_2}{\Omega_1(s_1)n_4}
=
\tanip{n_2}{\Omega_1(t_1)n_4}\,.
\ee
If the endpoint triple products
\be
\tanip{n_2\times n_1}{\Omega_1(s_1)n_4}>0\,,
\qquad
\tanip{n_2\times n_1}{\Omega_1(t_1)n_4}>0\,,
\label{eq:triple-1}
\ee
and
\be
\tanip{n_3\times n_2}{\Omega_3(s_3)^{-1}n_4}>0\,,
\qquad
\tanip{n_3\times n_2}{\Omega_3(t_3)^{-1}n_4}>0
\label{eq:triple-2}
\ee
hold, then $s_i=t_i$ for all $i$.  
\end{lemma}

\begin{proof}
We first use a one-parameter observation.  For one fixed nonzero axis $n$, write
\be
\Omega_n(u):=\exp\big(J(un)\big)\,.
\ee
The subscript $n$ only records that this is the one-parameter subgroup stabilizing the line $\mathbb R n$.  Let
\be
h(u):=\tanip{a}{\Omega_n(u)b}
\ee
on this coordinate space.  The closed exponential formulas give the following form of $h$ as
\be
A+B\cos u+C\sin u\,,\qquad
A+B\cosh u+C\sinh u\,,\qquad
A+Bu+\frac12Cu^2\,,
\ee
for elliptic, hyperbolic, and parabolic stabilizers respectively.  In the elliptic case, choose lifts of $x,y\in\mathbb R/2\pi\mathbb Z$ with $y-x\in(-2\pi,2\pi)$.  If $x\neq y$ on the circle, then $y-x\neq0$ and $\sin((y-x)/2)\neq0$.  Let $m=(x+y)/2$ and $d=(y-x)/2$.  Then
\be
h(y)-h(x)=2\sin d\,h'(m)\,.
\ee
In the hyperbolic and parabolic cases, the corresponding identities are
\be
h(y)-h(x)=2\sinh d\,h'(m)\,,
\qquad
h(y)-h(x)=2d\,h'(m)\,.
\ee
Thus $h(x)=h(y)$ with $x\neq y$ forces $h'(m)=0$ in all three cases.  For the same three elementary forms, $h'(m)=0$ also gives
\be
h'(x)+h'(y)=0\,.
\ee
Thus the endpoint derivatives have opposite signs, unless both vanish.  Hence
\be
h(x)=h(y)\,,\qquad h'(x)>0\,,\qquad h'(y)>0
\ee
imply $x=y$. 

Apply this to
\be
f_1(u):=\tanip{n_2}{\Omega_1(u)n_4}\,.
\ee
Since $\Omega_1(u)=\exp(uJ(n_1))$, we have
\be
\frac{d}{du}\Omega_1(u)=J(n_1)\Omega_1(u)\,.
\ee
Therefore, using $J(n_1)w=n_1\times w$ and the cyclicity of the Lorentzian scalar triple product,
\be
f_1'(u)
=\tanip{n_2}{J(n_1)\Omega_1(u)n_4}  
=\tanip{n_2}{n_1\times\Omega_1(u)n_4} 
=\tanip{n_2\times n_1}{\Omega_1(u)n_4}\,.
\ee
{The} exceptional scalar equality and the two positive endpoint signs give $s_1=t_1$.  Cancelling the equal first factors in the closure equations gives
\be
Q:=
\Omega_4(s_4)\Omega_3(s_3)\Omega_2(s_2)
=
\Omega_4(t_4)\Omega_3(t_3)\Omega_2(t_2)\,.
\ee
Applying this equality to $n_2$ and using that $\Omega_2$ fixes $n_2$ gives
\be
\Omega_4(s_4)\Omega_3(s_3)n_2
=
\Omega_4(t_4)\Omega_3(t_3)n_2\,.
\ee
Taking the inner product with $n_4$, and using that both $\Omega_4(s_4)$ and $\Omega_4(t_4)$ fix $n_4$, gives
\be
\tanip{n_4}{\Omega_3(s_3)n_2}
=
\tanip{n_4}{\Omega_3(t_3)n_2}\,.
\ee
Now apply the one-parameter observation to
\be
f_3(u):=\tanip{n_2}{\Omega_3(u)^{-1}n_4}
=\tanip{n_4}{\Omega_3(u)n_2}\,.
\ee
Its derivative is
\be
f_3'(u)=\tanip{n_3\times n_2}{\Omega_3(u)^{-1}n_4}\,,
\ee
so the equality just derived and the two positive endpoint signs give $s_3=t_3$.

Using $s_3=t_3$, apply the equality of $Q$ to $n_2$ once more.  Since $\Omega_2$ fixes $n_2$, we get
\be
\Omega_4(s_4)a=\Omega_4(t_4)a\,,
\qquad
a:=\Omega_3(s_3)n_2\,.
\ee
The positivity of $f_3'(s_3)$ implies $a\notin\mathbb R n_4$.  Therefore, $\Omega_4(t_4-s_4)\equiv \Omega_4(s_4)^{-1}\Omega_4(t_4)$ cannot be a nontrivial stabilizer element, because every nontrivial element of this one-parameter subgroup fixes only the axis $\mathbb R n_4$.  Hence $\Omega_4(s_4)=\Omega_4(t_4)$, and the definition of $I_4$ gives $s_4=t_4$ as stabilizer coordinates.  The closure equation then gives $\Omega_2(s_2)=\Omega_2(t_2)$, and the definition of $I_2$ gives $s_2=t_2$.
\end{proof}

\begin{proof}[Proof of Theorem \ref{theorem:unified_minkowski}]
Each nontrivial $O_i\in\SO^+(1,2)$ has an invariant normal line $\mathcal N_i=\ker(O_i-\id)$.  
Choose any preliminary representative $n_i$ of this line, normalized when it is non-null, and write the holonomy as in Lemma \ref{lemma:SO_axis}.  These choices give the preliminary matrix $G$ in \eqref{eq:reconstructed_gram}. 
Changing any representative by a nonzero real factor transforms $G$ by diagonal congruence.  Hence its inertia, determinant sign, nondegeneracy, and the nondegeneracy of its principal submatrices are independent of these preliminary choices.

At this stage, $G$ is only the matrix computed from the transported tangent normal representatives by \eqref{eq:reconstructed_gram}, and no ambient tetrahedron has yet been constructed.  Lemma \ref{lemma:lorentzian_gram_signature} is therefore the first step: it uses only the Lorentzian tangent model and shows that the nondegenerate principal submatrices force
\be
\Inertia(G)=\Inertia(\eta_{\sigma_G})\,,
\qquad
\Inertia(G_{\hat i})=(1,2)\,.
\ee
Then Lemma \ref{lemma:ambient_gram_reconstruction} applies to this abstract symmetric matrix and reconstructs four linearly independent ambient normals $N_i\in\mathbb V_{\sigma_G}$ uniquely up to $\mathcal O_{\sigma_G}$.  With those ambient normals in hand, Lemma \ref{lemma:vertex_reality} applies to the principal submatrices and shows that the four triple intersections are real vertices of $\mathcal M_{\sigma_G}$.

It remains only to choose the correct half-space branch.  Lemma \ref{lemma:branch_selection} chooses the unique outward signs for all normal representatives, including the null ones, up to global chirality. Replace $G$ by the diagonally congruent Gram matrix of these outward representatives and continue to denote it by $G$. Note that its determinant sign and inertia are unchanged.
The compensating changes of the stabilizer coordinates leave every prescribed holonomy unchanged.  Proposition \ref{prop:convex_branch} then gives support numbers
\be
\ambip{V_i}{N_i}<0\,,
\qquad i=1,\ldots,4\,,
\ee
for the outward convention.  Lemma \ref{lemma:simplicial_halfspace_realization} applies and shows that the selected half-spaces cut out the closed simplicial cone generated by the four vertex rays. Its quadric section is therefore the strictly convex tetrahedron with the prescribed faces and vertices. 
	Indeed, Assumption 3 is precisely the strict copositivity hypothesis \eqref{eq:strict_inverse_gram_copositivity} of that lemma with $\sigma=\sigma_G$.  Thus the whole cone has the correct quadric sign, and the reconstructed tetrahedron is finite.
	
Let $\widehat O_i(T_G)$ be its based Levi-Civita face holonomies.
{The reconstructed tetrahedron determines a based normal line fixed by
each face holonomy. Using the sub Gram matrix $G_{\hat4}$ at the base vertex $V_4$, choose
the base frame so that the three reconstructed outward normals through
$V_4$ are identified with $n_1,n_2,n_3$, respecting the selected
orientation.

Let $u$ be the fourth outward normal transported to the base vertex
along the special edge. By the metric-compatibility argument in the
proof of Proposition \ref{prop:dihedral_equivalence},
\be
\tanip{n_1}{u}=G_{14},\qquad
\tanip{n_3}{u}=G_{34},\qquad
\tanip{u}{u}=G_{44}.
\ee
The prescribed $n_4$ satisfies the same relations by
\eqref{eq:reconstructed_gram}. Hence
\be
u=n_4+c\,(n_1\times n_3)
\ee
for some $c\in\mathbb R$. Since
\be
\det G_{\hat2}=-\chi_2^2\neq0,
\ee
equality of the norms leaves only the two orientation-related choices,
and the selected sign of $\chi_2$ excludes the one opposite to
$n_4$. Thus $u=n_4$. We may therefore write
\be
O_i=\Omega_i(s_i),\qquad
\widehat O_i(T_G)=\Omega_i(t_i),\qquad
\Omega_i(u)=\exp\bigl(J(un_i)\bigr).
\ee}
The hyperbolic rapidity and {the coefficient relative to a fixed null representative} are real coordinates, while the elliptic angle is a coordinate on the $2\pi$-periodic circle.

Both ordered four-tuples of holonomies satisfy closure.  They also have the same exceptional Gram scalar,
\be
\tanip{n_2}{O_1n_4}
=G_{24}
=\tanip{n_2}{\widehat O_1(T_G)n_4}\,,
\ee
where the first equality is the definition of the reconstructed Gram entry and the second is Proposition \ref{prop:dihedral_equivalence} applied to $T_G$.  The endpoint triple products \eqref{eq:triple-1}--\eqref{eq:triple-2} in Lemma \ref{lemma:global_stabilizer_parameter_uniqueness} are positive for the prescribed holonomies by the selected branch and are positive for $\widehat O_i(T_G)$ by outward convexity.  Lemma \ref{lemma:global_stabilizer_parameter_uniqueness} therefore gives $s_i=t_i$ for every $i$, hence
\be
O_i=\widehat O_i(T_G)\,,
\qquad i=1,\ldots,4\,.
\ee

 The uniqueness is the uniqueness of the Gram factorization, modulo the global ambient isometry group, together with the unique outward branch selected by the triple-product signs.

\end{proof}

{Appendix~\ref{subsec:global_branch_examples} illustrates
the reconstruction criterion with two numerical tests starting
from closing $\SO^+(1,2)$ holonomies: an anti-de Sitter
reconstruction test and a domain-obstruction test, both with
mixed causal face types.}

\begin{remark}
The equality $O_i=\widehat O_i(T_G)$ is the point where the stabilizer-coordinate rigidity enters.  Once representatives of the normal lines are fixed, Lemma \ref{lemma:global_stabilizer_parameter_uniqueness} uses closure, the exceptional entry $G_{24}$, and the positive endpoint triple products to fix the remaining angle, boost,{ }or parabolic coefficient.  For a null line, only the invariant product of this coefficient and the chosen representative is intrinsic.
\end{remark}

\begin{remark}%
Null faces are allowed in Theorem \ref{theorem:unified_minkowski}.  The nontriviality assumption does not exclude them: a nontrivial parabolic holonomy has a well-defined null fixed line and a nilpotent logarithm.  In that case, $G_{ii}=0$ for the corresponding face, but Assumption 2 can still hold because a nondegenerate Lorentzian three-plane may contain null vectors.  By Lemma \ref{lemma:lorentzian_gram_signature}, the condition $\det G_{\hat i}\neq0$ is equivalent to $\chi_i\neq0$ and excludes the degenerate case where the three incident normals span a degenerate subspace and {excludes} the case when the opposite vertex is not a real point of the selected model $\mathcal M_{\sigma_G}$.
 For every causal type, the common sign of the $\chi_i$ is a branch-selection rule and not an independent hypothesis.

\end{remark}

Before writing the spin version, recall that $\ker\pi=\{\pm\id\}$.  Thus a projected closure relation in $\SO^+(1,2)$ lifts to
\be
H_4H_3H_2H_1=\epsilon\id\,,
\qquad
\epsilon=\pm1\,.
\ee
The same central freedom is available face by face: replacing $H_i$ by $-H_i$ leaves $\pi(H_i)$ unchanged but flips $B_i=\impart(H_i)$. 
This changes the signed trace representative, while the geometric normal branch is selected independently by Lemma \ref{lemma:branch_selection}.

\begin{corollary}[$\SL(2,\mathbb R)$ form]
\label{cor:SL_reconstruction}
Let $H_1,H_2,H_3,H_4\in\SL(2,\mathbb R)$ be noncentral,
\be
H_i\neq\pm\id\,,
\qquad
i=1,\ldots,4\,,
\label{eq:non-trivial-H}
\ee
and satisfy the spin closure relation
\be
H_4H_3H_2H_1=\epsilon\id\,,
\qquad
\epsilon=\pm1\,.
\ee
Assume that:
\begin{enumerate}
\item the simple-path convention, including the choice of special edge, has been fixed;
\item the Gram matrix $G$ reconstructed from any chosen representatives of the normal lines extracted from the traceless parts is nondegenerate, and every vertex principal submatrix $G_{\hat i}$ is nondegenerate;
\item after selecting the outward representatives as in Lemma \ref{lemma:branch_selection} and denoting their Gram matrix again by $G$, the signed inverse Gram form $-\sgnop(\det G)G^{-1}$ is strictly copositive.
\end{enumerate}
Then, with $\sigma_G:=-\sgnop(\det G)$, the spin data determine a unique finite strictly convex tetrahedron in $\mathcal M_{\sigma_G}$ up to $\mathcal O_{\sigma_G}$ in the sense of Theorem \ref{theorem:unified_minkowski}.  The projected holonomies $\pi(H_i)$ are exactly the based Levi-Civita face holonomies of this tetrahedron.   
\end{corollary}

\begin{proof}
The traceless parts $B_i=\impart(H_i)$ determine nonzero vectors $b_i$ on the invariant normal lines. 
{For a parabolic input, \eqref{eq:spin_closed_form} gives $b_i=\epsilon_i\kappa_i=\epsilon_i\Theta_i n_i$.  Thus the connected spin traces determine the logarithm-fixed product, while any chosen nonzero null representative $n_i$ is obtained by a nonzero rescaling and the coefficient $\Theta_i$ changes inversely. } The connected spin-trace formulas \eqref{eq:SL_connected_pair}--\eqref{eq:SL_connected_triple}, with the same exceptional transport insertions as in \eqref{eq:reconstructed_gram-H}, reconstruct a diagonally congruent Gram matrix for any such choices.

Set $O_i:=\pi(H_i)$.  The spin closure implies $O_4O_3O_2O_1=\id$, so the projected holonomies satisfy the vector closure condition.  The agreement with the vector formulas holds because the adjoint action of $H_i$ on $\mathfrak{sl}(2,\mathbb R)$ is the vector action of $O_i$.

The difference from the vector statement is how the normal signs are made available.  In the spin presentation, multiplying one input by the central element $-\id$ preserves the projected holonomy and preserves closure after projection, while it sends $B_i$ and the extracted representative $b_i$ to their negatives.  This applies to parabolic faces as well, because the nilpotent part seen by the connected spin traces changes sign.  For a spacelike face, this sign freedom is therefore not the elliptic periodic ambiguity of the projected holonomy.  It is the central lift freedom $H_i\leftrightarrow -H_i$.  The principal-submatrix assumption makes the triple products nonzero by Lemma \ref{lemma:lorentzian_gram_signature}. 
{Lemma \ref{lemma:branch_selection} then selects the outward representatives from the invariant normal lines, independently of the central lift choices. } Theorem \ref{theorem:unified_minkowski} completes the reconstruction, including the equality between the projected holonomies and the based Levi-Civita face holonomies.

It remains only to interpret the lift.  Choose one coherent system of spin lifts of the six reconstructed edge transports, with the lift of a reversed edge equal to the inverse lift, and denote the induced face products by $\widehat H_i(T_G)$.  Cancellation in the lifted simple-path words gives
\be
\widehat H_4(T_G)\widehat H_3(T_G)
\widehat H_2(T_G)\widehat H_1(T_G)=\id\,.
\ee
Since $\pi(H_i)=\widehat O_i(T_G)=\pi(\widehat H_i(T_G))$, there are signs $\delta_i=\pm1$ such that $H_i=\delta_i\widehat H_i(T_G)$, and comparison of the two products gives $\prod_i\delta_i=\epsilon$.  Changing the sign of one lifted edge transport flips the two incident face products, and such pair flips generate all even face-sign patterns.  Hence $\epsilon=+1$ permits simultaneous agreement with all prescribed $H_i$.  For $\epsilon=-1$ the odd total sign makes such agreement impossible for one coherent lifted system, although every $H_i$ separately is a lift of the correct projected face holonomy.  Changing any one prescribed sign reduces to the even case.
\end{proof}

\begin{remark}
The theorem should be read as a generic reconstruction theorem.  Degenerate cases where vertices go to infinity or hyperplanes become dependent require separate limiting statements.  A holonomy has a non-isolated logarithm when the holonomy no longer determines a unique normal line.  In the vector representation this happens, for example, at $O_i=\id$, since $\exp(2\pi mJ(n))=\id$ for every unit timelike $n$.  In the spin representation the analogous central cases are $H_i=\pm\id$.  These cases are excluded from the generic theorem because the face normal cannot be extracted from the holonomy alone.  By contrast, a rotation by $\pi$ around a timelike normal has a well-defined fixed line, even though the antisymmetric part $(O_i-O_i^{-1})/2$ vanishes.  Such a case should be handled by the eigenspace of $O_i$, not by the antisymmetric-part formula \eqref{eq:Oi_minus_inverse}.
\end{remark}

\section{Curvature-Zero Limit}
\label{sec:flat_limit}

The unit-radius convention used so far hides the relation with the ordinary vector closure condition in flat space.  We denote the curvature radius by $R$ and write
\be
\mathcal M_{\sigma,R}
:=
\{X\in\Vsig\mid \ambip{X}{X}=\sigma R^2\}\,.
\ee
The sectional curvature is $\sigma/R^2$, so the flat limit is $R\to\infty$ in a bounded geodesic coordinate patch.  Equivalently, if $K:=\sigma/R^2$ is the signed curvature, then the flat geometry is the common boundary $K=0$ between the positive-curvature de Sitter branch $K>0$ and the negative-curvature anti-de Sitter branch $K<0$.

Choose a unit radial vector $e\in\Vsig$ with $\ambip{e}{e}=\sigma$ and base points $p_R=R e\in \mathcal M_{\sigma,R}$.  Then
\be
T_{p_R}\mathcal M_{\sigma,R}=e^\perp
\ee
is canonically identified, for all $R$, with the same tangent Minkowski space.  We denote its pairing by
\be
\tanip{x}{y}:=\ambip{x}{y}\,,
\qquad x,y\in e^\perp\,.
\ee
A limiting flat face with normal $n_i\in e^\perp$ and support number $h_i$ is written as
\be
\Pi_i^{\rm flat}
=
\{x\in e^\perp\mid \tanip{x}{n_i}=h_i\}\,.
\ee

\begin{figure}[h!]
\centering
\includegraphics{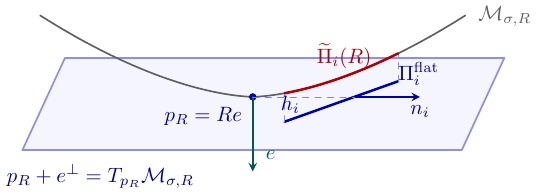}
\caption{A local schematic of the flat limit, suppressing one tangent direction.  The curved model $\mathcal M_{\sigma,R}$ passes through $p_R=Re$ and has affine tangent space $p_R+e^\perp$.    After translating $p_R+e^\perp$ to $e^\perp$, $h_i$ is the affine support coefficient of the limiting face; for null $n_i$ it is not a metric distance.  The curved {supporting surface} ${\widetilde{\Pi}_i(R)}$ is the nearby slice of $\mathcal M_{\sigma,R}$ whose geodesic-coordinate image converges to $\Pi_i^{\rm flat}$ as $R\to\infty$.}
\label{fig:flat_limit_local_schematic}
\end{figure}
The dimension-reduced picture is illustrated in Fig.{~}\ref{fig:flat_limit_local_schematic} with one direction suppressed. 
For a family of curved {supporting surfaces} ${\widetilde{\Pi}_i(R)}=N_i(R)^\perp\cap\mathcal M_{\sigma,R}$ converging to this flat face, choose the extrinsic normals $N_i(R)$ so that their causal normalization, or null scale, converges to that of $n_i$.  The finite-curvature Gram matrix is
\be
G_{ij}(R):=\ambip{N_i(R)}{N_j(R)}\,.
\ee

\begin{lemma}
\label{lemma:flat_limit_first_order_normal_gram}
With the conventions above, suppose that ${\widetilde{\Pi}_i(R)}$ converges to
$\Pi_i^{\rm flat}$ in a bounded coordinate patch and that each chosen normal
representative admits a first-order expansion in $R^{-1}$,
\be
N_i(R)=n_i+\frac{w_i}{R}+O(R^{-2})
\ee
for some $w_i\in\Vsig$.  Then there is a vector $v_i\in e^\perp$ such that
\be
N_i(R)
=n_i+\frac{v_i}{R}-\frac{\sigma h_i}{R}e+O(R^{-2})\,.
\label{eq:flat_limit_first_order_normal_expansion}
\ee
Consequently,
\be
G_{ij}(R)=\tanip{n_i}{n_j}+O(R^{-1})\,.
\label{eq:flat_limit_first_order_gram_rate}
\ee
\end{lemma}

\begin{proof}
Write $w_i=v_i+a_i e$, where $v_i\in e^\perp$ and $a_i\in\mathbb R$.
Every point of $\mathcal M_{\sigma,R}$ near $p_R=R e$ can be written uniquely as
\be
X_R(x)=\alpha_R(x)e+x\,,
\qquad x\in e^\perp\,,
\ee
on the branch with $\alpha_R(0)=R$.  The equation $\ambip{X_R(x)}{X_R(x)}=\sigma R^2$ gives
\be
\sigma\alpha_R(x)^2+\tanip{x}{x}
=
\sigma R^2\,,
\ee
therefore
\be
\alpha_R(x)
=
R\sqrt{1-\frac{\sigma\tanip{x}{x}}{R^2}}
=
R-\frac{\sigma\tanip{x}{x}}{2R}
+O(R^{-3})
\label{eq:flat_limit_determinant_signed_asymptotic}
\ee
for bounded $x$.
Substitution in the curved face equation gives
\be
\ambip{X_R(x)}{N_i(R)}
=\tanip{x}{n_i}+\sigma a_i+O(R^{-1})\,.
\ee
Its zero set defining the curved face converges to $\tanip{x}{n_i}=h_i$, so
$\sigma a_i=-h_i$ and hence $a_i=-\sigma h_i$.  This proves
\eqref{eq:flat_limit_first_order_normal_expansion}.

Taking ambient inner products and using $e\perp n_i,v_i$ gives
\be
G_{ij}(R)
=\tanip{n_i}{n_j}
+\frac{\tanip{v_i}{n_j}+\tanip{n_i}{v_j}}{R}
+O(R^{-2})\,,
\ee
which proves \eqref{eq:flat_limit_first_order_gram_rate}.
\end{proof}

For every nondegenerate finite-radius configuration,
$G(R)=\mathsf N(R)^\top\eta_\sigma\mathsf N(R)$, where the columns of
$\mathsf N(R)$ are the four ambient normals. Hence
\be
\det G(R)
=
-\sigma\det\mathsf N(R)^2\,,
\ee
and therefore $\sgnop(\det G(R))=-\sigma$. In the flat limit,
$G(R)$ converges to the Gram matrix of four vectors in the
three-dimensional space $e^\perp$, so the limiting Gram matrix is singular.

\begin{prop}%
\label{prop:automatic_flat_holonomy_expansion}
Let $T_R\subset\mathcal M_{\sigma,R}$ be finite domain-admissible
tetrahedra whose ordinary triangular face $F_i(R)\subset \widetilde{\Pi}(R)$, based edge
paths, and based frames converge in $C^1$ on one bounded normal-coordinate
chart to the corresponding data of a bounded affine tetrahedron.  Let
$\mathbf E_i(R)$ be their based area fluxes.  Then
\be
\mathbf E_i(R)\longrightarrow\mathbf E_i^0\,,
\qquad
O_i(R)=\exp\!\left(-\frac{\sigma}{R^2}J(\mathbf E_i(R))\right)
=\id-\frac{\sigma}{R^2}J(\mathbf E_i^0)+o(R^{-2})\,.
\label{eq:flat_limit_holonomy_expansion}
\ee
If the curved closure holds for every $R$, then
\be
\sum_{i=1}^4\mathbf E_i^0=0\,.
\ee
For an abstract closing group family, the same conclusion holds whenever
the displayed asymptotic expansion is assumed.
\end{prop}

\begin{proof}
After restoring the curvature radius,
\eqref{eq:curvature_two_form_constant} reads
$R^I=(\sigma/R^2)\mathcal S^I$.  In the normal-adapted gauge of Lemma
\ref{lemma:stabilizer_parameter_area}, Abelian Stokes gives the displayed
exponential exactly.  The assumed $C^1$ convergence of the patches, paths,
frames implies convergence of the vector-valued
area forms, their integrals, and the basing transports; hence
$\mathbf E_i(R)\to\mathbf E_i^0$.  Expanding the exponential and
$O_4O_3O_2O_1=\id$ gives $J(\sum_i\mathbf E_i^0)=0$.  Since $J$ is an
isomorphism, the stated closure follows.
\end{proof}

{
 Let $T_R\subset\mathcal M_{\sigma,R}$ be a family of finite
domain-admissible non-degenerate tetrahedra.  
Suppose that, in normal coordinates centered at $p_R$, the triangular
face and chosen basing paths lie in a fixed bounded coordinate
domain and converge in $C^1$, while the chosen base frames converge.
Then $\mathbf E_i(R)\to\mathbf E_i^0$, and Lemma
\ref{lemma:stabilizer_parameter_area} gives
\be
O_i(R)
=
\exp\!\left(-\frac{\sigma}{R^2}J(\mathbf E_i(R))\right)
=
\id-\frac{\sigma}{R^2}J(\mathbf E_i^0)+o(R^{-2})\,.
\ee
Expanding the closure relation $O_4O_3O_2O_1=\id$ therefore yields
\be
\sum_{i=1}^4\mathbf E_i^0=0\,.
\label{eq:flat_Lorentzian_Minkowski_closure}
\ee
For an abstract closing group family, the same conclusion holds if
the displayed holonomy expansion is assumed. Assume $\chi_i(R)\neq0$ for every $i$ {in the limit $R\rightarrow\infty$}, and that each $\mathbf E_i^0$ is a
nonzero multiple of the corresponding normal{. Then} every three
$\mathbf E_i^0$ are linearly independent. 
In the fixed reference frame
$\eta_T=\mathrm{diag}(-1,1,1)$, define the Euclidean area vectors $\mathbf E_i^{\mathrm E}:=\eta_T\mathbf E_i^0$. 
Since $\eta_T$ is invertible, they are again linearly independent in every
triple, and \eqref{eq:flat_Lorentzian_Minkowski_closure} gives
\be
\sum_{i=1}^4\mathbf E_i^{\mathrm E}=0\,.
\ee
The classical Minkowski theorem therefore reconstructs a unique bounded
convex Euclidean tetrahedron, up to translation, with outward area vectors
$\mathbf E_i^{\mathrm E}$.  Regard the same affine tetrahedron as a subset of
$\mathbb R^{1,2}$.  Since the Euclidean and Lorentzian volume forms agree in
the chosen frame and $\eta_T^2=\id$, its Lorentzian face vectors are
\be
\eta_T\mathbf E_i^{\mathrm E}=\mathbf E_i^0\,.
\ee
Thus it is the required convex Lorentzian tetrahedron with the prescribed
area scale.  The argument remains valid when some $\mathbf E_i^0$ are
null,
and the
prescribed vectors fix the scale. %

\section{Character Varieties and Flat Connections}
\label{sec:character_varieties_flat_connections}

The closure relation also has a standard moduli-space interpretation in the language of relative character varieties and flat connections \cite{Goldman:1984,AtiyahBott:1983}.  Let $S_{0,4}$ be an oriented four-holed sphere, with boundary loops $\gamma_i$ ordered so that
\be
\gamma_4\circ\gamma_3\circ\gamma_2\circ\gamma_1=1\,.
\ee
For a Lie group $G$ and prescribed boundary conjugacy classes $\mathcal C_i\subset G$, the relative character variety is
\be
\mathcal M_{\vec{\mathcal C}}(S_{0,4},G)
=
\left\{
(g_1,g_2,g_3,g_4)\in\mathcal C_1\times\cdots\times\mathcal C_4
\,\middle|\,
g_4g_3g_2g_1=\id
\right\}/G\,,
\label{eq:relative_character_variety}
\ee
where $G$ acts by simultaneous conjugation.  Thus the based holonomies of Theorem \ref{theorem:unified_minkowski}, modulo the common-frame freedom, define points in the relative character variety for $G=\SO^+(1,2)\cong\PSL(2,\mathbb R)$.  The conjugacy classes encode the face type and its holonomy parameter: elliptic for spacelike faces, hyperbolic for timelike faces, and parabolic for null faces.
Nonemptiness of \eqref{eq:relative_character_variety} supplies
closing holonomies only.  A finite tetrahedron additionally requires a
nondegenerate reconstructed Gram matrix, the outward branch, and strict
copositivity of its signed inverse (see Theorem
\ref{theorem:character_geometric_chambers}).

For spin holonomies, there is one small but important refinement.  A projected $\PSL(2,\mathbb R)$ representation admits lifts satisfying $H_4H_3H_2H_1=\epsilon\id$, with $\epsilon=\pm1$.  If we change the lifts by central signs $H_i\mapsto\delta_iH_i$, then
\be
(\delta_4H_4)(\delta_3H_3)(\delta_2H_2)(\delta_1H_1)
=
\left(\prod_{i=1}^4\delta_i\right)\epsilon\,\id\,.
\label{eq:spin_lift_sign_change}
\ee
Hence one obtains the standard $\SL(2,\mathbb R)$ character-variety relation precisely when $\prod_i\delta_i=\epsilon$.  For example, if $\epsilon=-1$, replacing any one lift by its negative changes the product to $\id$.  The sign in Corollary \ref{cor:SL_reconstruction} therefore records the chosen spin lift, not an additional obstruction in the projected geometry.

This is the Chern--Simons phase-space viewpoint \cite{ChernSimons:1974,Witten:1989}.
 On the smooth locus represented by irreducible flat connections, the moduli space of flat $G$-connections on $S_{0,4}$ with fixed boundary conjugacy classes carries the Goldman, equivalently Atiyah--Bott, symplectic structure \cite{Goldman:1984,AtiyahBott:1983}.  In the compact real form $G=\SO(3)$ or $\SU(2)$, this is the phase space used in the curved Minkowski theorem of Haggard--Han--Riello \cite{Haggard:2015ima}.  In the Lorentzian real form $G=\PSL(2,\mathbb R)$ or $\SL(2,\mathbb R)$, it is the natural boundary phase space for $\SL(2,\mathbb R)$ Chern--Simons theory.  Our reconstruction theorem can therefore be viewed as a local holonomy-to-tetrahedron reconstruction problem on this character-variety phase space.
The following subsections express this geometric selection
directly in trace coordinates and describe its boundary.

\subsection{Geometric chambers of the relative character surface}
\label{subsec:character_geometric_chambers}

The reconstruction criterion can be expressed directly as a semialgebraic chamber decomposition of the relative character surface.  The geometric reconstruction depends only on the projected $\PSL(2,\mathbb R)\cong \SO^+(1,2)$ holonomies, but ordinary traces are naturally functions of chosen $\SL(2,\mathbb R)$ lifts.  We therefore formulate the computation in trace coordinates on the real Fricke surface and then restrict it to the trace-coordinate image of the relative $\SL(2,\mathbb R)$ character variety. The resulting geometric chambers descend to the corresponding liftable relative $\PSL(2,\mathbb R)$ character variety.  Choose spin lifts satisfying
\be
H_4H_3H_2H_1=\id\,,
\label{eq:H-closure}
\ee
and write
\be
t_i:=\trspin(H_i)\,.
\label{eq:t_i-def}
\ee
We call $\{H_1,H_2,H_3,H_4\}$ satisfying \eqref{eq:H-closure} a closing $\SL(2,\mathbb R)$ quadruple.
Fixing the four traces determines the usual affine trace surface, while fixing four conjugacy classes inside $\SL(2,\mathbb R)$ additionally specifies connected components or sheets inside its real locus.  This distinction is relevant because the trace alone does not always distinguish the oriented elliptic classes or the two nontrivial parabolic classes {at each trace $+2$ or $-2$} in $\SL(2,\mathbb R)$.

Introduce the trace coordinates adapted to the simple-path convention,
\be
x:=\trspin(H_2H_1)\,,
\qquad
y:=\trspin(H_3H_2)\,,
\qquad
z:=\trspin(H_1H_3)\,.
\label{eq:adapted_fricke_coordinates}
\ee
The closure relation and cyclicity of the trace give
\be
x=\trspin(H_4H_3)\,,
\qquad
y=\trspin(H_1H_4)\,,
\ee
and
\be
z
=\trspin\!\left(H_2H_1H_4H_1^{-1}\right)\,.
\label{eq:exceptional_fricke_coordinate}
\ee
Indeed, closure gives $H_2H_1H_4H_1^{-1}=(H_1H_3)^{-1}$, and an $\SL(2,\mathbb R)$ matrix and its inverse have the same trace.  Thus the third trace coordinate is precisely the one required by the exceptional pair $(F_2,F_4)$.

The traces $\mathbf t=(t_1,t_2,t_3,t_4)$ defined in \eqref{eq:t_i-def} and $x,y,z$ defined in \eqref{eq:adapted_fricke_coordinates} obey the Fricke relation for the four-holed sphere written in the coordinates \cite{vogt1889invariants,fricke1897vorlesungen,Nekrasov:2014yra}:
\be
\begin{split}
\mathcal F_{\mathbf t}(x,y,z)
:={}&x^2+y^2+z^2+xyz
-(t_1t_2+t_3t_4)x\\
&-(t_1t_4+t_2t_3)y
-(t_1t_3+t_2t_4)z
+\sum_{i=1}^4t_i^2+t_1t_2t_3t_4-4
=0\,.
\end{split}
\label{eq:adapted_fricke_cubic}
\ee

For fixed real boundary traces $\mathbf t$, we define the {\it real Fricke surface} to be the real affine cubic
{\be
\mathfrak F_{\mathbf t}(\mathbb R)
:=
\left\{
(x,y,z)\in\mathbb R^3
\,\middle|\,
\mathcal F_{\mathbf t}(x,y,z)=0
\right\}\,.
\label{eq:real_fricke_surface}
\ee}
Here ``real'' refers only to the trace coordinates $t_i,x,y,z$.  It does not impose an $\SL(2,\mathbb R)$ real-form condition on a representation.  For example, when $t_i=2\cos\theta_i$ with $|t_i|<2$, the phase eigenvalues $e^{\pm\mathrm i\theta_i}$ occur both for elliptic $\SL(2,\mathbb R)$ elements and for $\SU(2)$ elements.  The trace-coordinate image of the relative $\SL(2,\mathbb R)$ character variety with conjugacy classes prescribed inside $\SL(2,\mathbb R)$ is therefore, in general, only a semialgebraic subset, often a union of components or sheets, inside $\mathfrak F_{\mathbf t}(\mathbb R)$. The full real affine cubic can also contain compact-real-form characters with the same real traces, and fixing the traces alone need not distinguish the two elliptic conjugacy classes with opposite rotation directions, or the two nontrivial parabolic conjugacy classes at trace $+2$ or $-2$.
The topology of the nonsingular real Fricke surfaces associated with the
four-holed sphere was classified by Benedetto--Goldman
\cite{BenedettoGoldman1999}, while Maloni--Palesi--Tan studied their
mapping-class-group dynamics and domains of discontinuity
\cite{MaloniPalesiTan2015}.

The signed trace coordinates themselves do not descend to $\PSL(2,\mathbb R)$. However, consider any other closing lifts that are of the form
\be
H_i\longmapsto\delta_iH_i\,,
\qquad
\delta_i\in\{\pm1\}\,,
\qquad
\prod_{i=1}^4\delta_i=1\,.
\label{eq:closing_lift_action}
\ee
They act on the trace coordinates by
\be
t_i\longmapsto\delta_it_i\,,
\qquad
x\longmapsto\delta_1\delta_2x\,,
\qquad
y\longmapsto\delta_2\delta_3y\,,
\qquad
z\longmapsto\delta_1\delta_3z\,.
\label{eq:fricke_lift_action}
\ee
Every term of \eqref{eq:adapted_fricke_cubic} is invariant under \eqref{eq:fricke_lift_action}.  Thus, closing-lift changes \eqref{eq:closing_lift_action} identify sign-related $\SL(2,\mathbb R)$ Fricke surfaces lying over the same liftable $\PSL(2,\mathbb R)$ character locus.

Define the symmetric trace-polynomial matrix
\be
\mathcal G_{\mathbf t}(x,y,z)
:=
\begin{pmatrix}
t_1^2-4 & 2x-t_1t_2 & 2z-t_1t_3 & 2y-t_1t_4\\
2x-t_1t_2 & t_2^2-4 & 2y-t_2t_3 & 2z-t_2t_4\\
2z-t_1t_3 & 2y-t_2t_3 & t_3^2-4 & 2x-t_3t_4\\
2y-t_1t_4 & 2z-t_2t_4 & 2x-t_3t_4 & t_4^2-4
\end{pmatrix}\,.
\label{eq:character_connected_gram}
\ee

\begin{lemma}
\label{lemma:adapted_fricke_gram_map}

Suppose that $\{H_1,\ldots,H_4\}$ is a closing $\SL(2,\R)$ quadruple with noncentral entries, and let
$(x,y,z)$ be their adapted trace triple. Let $G$ be the reconstructed Gram matrix
\eqref{eq:reconstructed_gram}, formed using normalized non-null
representatives and arbitrary nonzero null representatives $n_i$. Then there is an invertible real diagonal matrix $D$ such that
\be
\mathcal G_{\mathbf t}=DGD\,.
\label{eq:character_gram_congruence}
\ee
  Consequently, $\mathcal G_{\mathbf t}$ and $G$ have the same inertia and determinant sign, while corresponding principal minors have the same signs and the same zero loci.
\end{lemma}

\begin{proof}
Let $B_i=\impart(H_i)=\mathcal T(b_i)$ be the traceless representatives defined in \eqref{eq:spin_Bi_scaled}.
The Cayley--Hamilton identity gives
\be
\trspin(H_i^2)=t_i^2-2\,,
\ee
while the connected-trace definition gives
\be
\tanip{b_i}{b_j}
=4\langle H_iH_j\rangle_{C,\SL}
=2\trspin(H_iH_j)-t_it_j\,.
\label{eq:character_connected_pair}
\ee
Using \eqref{eq:adapted_fricke_coordinates} and \eqref{eq:exceptional_fricke_coordinate} in \eqref{eq:character_connected_pair} gives \eqref{eq:character_connected_gram}.  For every face, write
\be
b_i=d_in_i\,,
\qquad
d_i=\epsilon_i s_i\,,
\label{eq:character_gram_scale_factors}
\ee
where $\epsilon_i$ and $s_i$ are defined in \eqref{eq:spin_closed_form}. All $d_i$ are nonzero because a noncentral holonomy has $s_i\neq0$ in every causal case.  Since the exceptional representative also scales linearly,
\be
H_1\cdot b_4=d_4(H_1\cdot n_4)=d_4O_1n_4\,,
\ee
the same factors apply to every entry, including the exceptional pair $(F_2,F_4)$.  Therefore,
\be
\mathcal G_{\mathbf t}=DGD\,,
\qquad
D:=\operatorname{diag}(d_1,d_2,d_3,d_4)\,.
\ee
Sylvester's law of inertia gives equality of inertia.  Moreover,
\be
\det\mathcal G_{\mathbf t}=(\det D)^2\det G\,,
\qquad
\det(\mathcal G_{\mathbf t})_{\hat i}
=\bigl(\det D_{\hat i}\bigr)^2\det G_{\hat i}\,,
\ee
which proves the remaining claims.
\end{proof}

Under the closing-lift change \eqref{eq:closing_lift_action}, the traceless representatives transform as $b_i\mapsto\delta_ib_i$.  Hence
\be
\mathcal G_{\mathbf t}
\longmapsto
D_\delta\mathcal G_{\mathbf t}D_\delta\,,
\qquad
D_\delta:=\operatorname{diag}(\delta_1,\delta_2,\delta_3,\delta_4)\,.
\label{eq:character_gram_lift_action}
\ee
It follows that the inertia, determinant, and all principal minors of $\mathcal G_{\mathbf t}$ are independent of the closing lifts.  These quantities therefore define functions on the corresponding liftable relative $\PSL(2,\mathbb R)$ character locus, even though the individual signed trace coordinates do not.

Define functions on the real Fricke surface $\mathfrak F_{\mathbf t}(\mathbb R)$ by
\be
\Delta_{\mathbf t}(x,y,z)
:=\det\mathcal G_{\mathbf t}(x,y,z)\,,
\qquad
m_i(x,y,z)
:=\det\bigl(\mathcal G_{\mathbf t}(x,y,z)\bigr)_{\hat i}\,.
\label{eq:character_discriminants}
\ee
Also define
\be
P_i:=-\partial\mathcal F_{\mathbf t}/\partial t_i\,.
\label{eq:def-Pi}
\ee
\begin{lemma}
\label{lemma:character_principal_minor_squares}
For arbitrary real trace coordinates, the principal determinants of \eqref{eq:character_connected_gram} satisfy
{\be
m_i
=16\mathcal F_{\mathbf t}(x,y,z)-4P_i^2\,,
\qquad i=1,\ldots,4\,.
\label{eq:character_principal_minor_square_identity}
\ee}
Consequently, on the entire real Fricke surface $\mathfrak F_{\mathbf t}(\mathbb R)$ \eqref{eq:real_fricke_surface},
\be
m_i=-4P_i^2\leq0\,.
\label{eq:character_principal_triples}
\ee
\end{lemma}

\begin{proof}
For a symmetric $3\times3$ matrix,
{\be
\det
\begin{pmatrix}
A&u&v\\
u&B&w\\
v&w&C
\end{pmatrix}
=ABC+2uvw-Aw^2-Bv^2-Cu^2\,.
\ee}
Applying this formula to each principal submatrix of \eqref{eq:character_connected_gram}, collecting the terms proportional to the Fricke polynomial \eqref{eq:adapted_fricke_cubic}, and completing the remaining square gives \eqref{eq:character_principal_minor_square_identity}.  Restriction to $\mathcal F_{\mathbf t}=0$ gives \eqref{eq:character_principal_triples}.
\end{proof}

\begin{lemma}%
\label{lemma:character_trace_triples}
Suppose that a point of $\mathfrak F_{\mathbf t}(\mathbb R)$ is represented by a closing $\SL(2,\mathbb R)$ quadruple $H_i$, and let $B_i=\mathcal T(b_i)$ be its traceless representatives.  Define the ordered scaled Lorentzian triple products
{\be
\begin{aligned}
\widetilde\chi_1
&:=\tanip{b_3\times b_2}{(H_3^{-1})\cdot b_4}\,,
&
\widetilde\chi_2
&:=\tanip{b_1\times b_3}{b_4}\,,\\
\widetilde\chi_3
&:=\tanip{b_2\times b_1}{H_1\cdot b_4}\,,
&
\widetilde\chi_4
&:=\tanip{b_1\times b_2}{b_3}\,.
\end{aligned}
\label{eq:character_scaled_triples}
\ee}
With precisely these orderings,
{\be
\widetilde\chi_i=2P_i\,,
\qquad i=1,\ldots,4\,.
\label{eq:character_P_triple_relation}
\ee}
Moreover, if $b_i=d_in_i$ as in \eqref{eq:character_gram_scale_factors}, then
{\be
2P_i=\widetilde\chi_i
=\left(\prod_{p\neq i}d_p\right)\chi_i\,.
\label{eq:character_P_normalized_triple_relation}
\ee}
Here the $\chi_i$ are the  triple products
defined in \eqref{eq:transported_chi}.
\end{lemma}

\begin{proof}
For any $A,B,C\in\SL(2,\mathbb R)$, the standard $2\times2$ trace identity is
{\be
\begin{split}
\trspin(ABC)+\trspin(ACB)
={}&\trspin(AB)\trspin(C)+\trspin(AC)\trspin(B)\\
&+\trspin(BC)\trspin(A)
-\trspin(A)\trspin(B)\trspin(C)\,.
\end{split}
\label{eq:three_matrix_trace_reversal}
\ee}
For each ordered triple in \eqref{eq:character_scaled_triples}, closure makes the trace of its reversed ordering equal to the trace $t_i$ of the omitted holonomy. Expanding the traceless parts $B_i=H_i-\frac12t_i\id$ in \eqref{eq:three_matrix_trace_reversal} therefore gives
{\be
\begin{aligned}
P_1&=2\trspin(B_3B_2B_{4|3^{-1}})\,,
&
P_2&=2\trspin(B_1B_3B_4)\,,\\
P_3&=2\trspin(B_2B_1B_{4|1})\,,
&
P_4&=2\trspin(B_1B_2B_3)\,.
\end{aligned}
\label{eq:character_P_connected_traces}
\ee}
The identity $\trspin(\mathcal T(u)\mathcal T(v)\mathcal T(w)) =\frac14\tanip{u\times v}{w}$ from \eqref{eq:SL_connected_triple}, together with \eqref{eq:spin_conj_B41} and \eqref{eq:spin_conj_B4_3inv}, proves \eqref{eq:character_P_triple_relation}.  Finally, substituting $b_p=d_pn_p$ into \eqref{eq:character_scaled_triples} and comparing with \eqref{eq:transported_chi} proves \eqref{eq:character_P_normalized_triple_relation}.
\end{proof}

Lemma \ref{lemma:character_trace_triples} applies only on the locus represented by actual $\SL(2,\mathbb R)$ holonomies, where the Lorentzian vectors $b_i$ and triple products are defined.  On the full real Fricke surface, including any compact-real-form locus, the polynomials $P_i$ remain globally defined but need not have a Lorentzian triple-product interpretation.  Equation \eqref{eq:character_principal_triples} nevertheless shows directly that $m_i\leq0$ on the entire real Fricke surface $\mathfrak F_{\mathbf t}(\mathbb R)$.  On the $\SL(2,\mathbb R)$ locus, it also recovers Lemma \ref{lemma:lorentzian_gram_signature}. Hence $m_i\neq0$, equivalently $m_i<0$, is exactly the nondegeneracy condition for the vertex-normal triple associated with the candidate vertex opposite $F_i$.

\begin{lemma}
\label{lemma:character_absolute_irreducibility}
If a closing $\SL(2,\mathbb R)$ quadruple has $m_i\neq0$ for at least one
$i$, then its complexification is irreducible.
\end{lemma}
\begin{proof}
If the complexified representation were reducible, all four matrices could
be made simultaneously upper triangular.  Their traceless parts would then
lie in the two-dimensional space of traceless upper triangular matrices, so
every alternating triple product in
\eqref{eq:character_scaled_triples} would vanish.  Lemma
\ref{lemma:character_trace_triples} would give $P_i=0$ for every $i$, and
\eqref{eq:character_principal_triples} would give $m_i=0$ for every $i$, a
contradiction.
\end{proof}

On the curved finite locus $\Delta_{\mathbf t}\neq0$ and $P_i\neq0$ for
every $i$, let
\be
D_P:=\diag(P_1,P_2,P_3,P_4)\,,
\qquad
A_{\mathbf t}:=-D_P\operatorname{adj}(\mathcal G_{\mathbf t})D_P\,.
\label{eq:character_polynomial_domain_form}
\ee
Their entries are polynomials in the trace coordinates.
Define

\be
\mu_{\mathbf t}(x,y,z)
:=
\min_{\substack{\lambda_i\geq0\\ \sum_i\lambda_i=1}}
\lambda^\top A_{\mathbf t}(x,y,z)\lambda\,.
\label{eq:character_copositivity_simplex_test}
\ee
By homogeneity, strict copositivity of $A_{\mathbf t}$ is equivalent to $\mu_{\mathbf t}>0$.
The following proposition translates this to restrictions on $A_{\mathbf t}$ directly. 

\begin{prop}
\label{prop:character_polynomial_domain_criterion}

{Let $(x,y,z)$ be represented by a closing
$\SL(2,\mathbb R)$ quadruple, and suppose that
\be
\Delta_{\mathbf t}\neq0\,,
\qquad
P_i\neq0
\quad\hbox{for every }i\,.
\ee}
Define the six edge margins
\be
E_{ij}:=(A_{\mathbf t})_{ij}+4P_i^2P_j^2
\qquad(i<j)
\label{eq:character_edge_domain_margins}
\ee
and, for $\widehat\ell=\{i,j,k\}$, the four face margins
\be
\Lambda_\ell
:=4P_i^2P_j^2P_k^2
+P_k^2(A_{\mathbf t})_{ij}
+P_j^2(A_{\mathbf t})_{ik}
+P_i^2(A_{\mathbf t})_{jk}\,.
\label{eq:character_face_domain_margins}
\ee
Then the finite domain condition \eqref{eq:theorem_strict_copositivity} in Assumption~3 of Theorem
\ref{theorem:unified_minkowski} holds if and only if $A_{\mathbf t}$ is
strictly copositive, and these equivalent conditions hold if and only if
\begin{align}
E_{ij}&>0 &&(i<j)\,,
\label{eq:character_edge_domain_tests}\\
t_\ell^2&<4\quad\hbox{or}\quad\Lambda_\ell>0
&&(\ell=1,\ldots,4)\,,
\label{eq:character_face_domain_tests}\\
\Delta_{\mathbf t}&>0\quad\hbox{or}\quad
P_rP_s(\mathcal G_{\mathbf t})_{rs}>0
\quad\hbox{for some }r,s\,.
\label{eq:character_full_domain_test}
\end{align}

\end{prop}

This is the character-polynomial version of Lemma \ref{lemma:finite_domain_admissibility_criterion}, and \eqref{eq:character_edge_domain_tests} -- \eqref{eq:character_full_domain_test} correspond to \eqref{eq:finite_domain_edge_tests} -- \eqref{eq:finite_domain_full_test} respectively. 

\begin{proof}
Choose the $n_i$ in Lemma \ref{lemma:adapted_fricke_gram_map} to be the
outward representatives selected by Lemma \ref{lemma:branch_selection}, and
let $G$ denote their Gram matrix.  Thus every $\chi_i$ is positive.
{
Recall from \eqref{eq:character_gram_scale_factors} that the vector
$b_i$ determined by the traceless part of $H_i$ satisfies
\be
b_i=d_in_i\,,
\qquad
d_i=\epsilon_is_i\neq0\,,
\ee
where $n_i$ is the chosen outward normal representative.
} Let
\be
S_P:=\diag(\sgnop P_1,\ldots,\sgnop P_4)\,,
\qquad
D_{|d|}:=\diag(|d_1|,\ldots,|d_4|)\,.
\ee
By \eqref{eq:character_P_normalized_triple_relation},
$\sgnop(P_id_i)=\sgnop(\prod_p d_p)$, independently of $i$.  Thus the
vectors $\sgnop(P_i)b_i$ are, up to one common overall sign, $|d_i|n_i$.
By Lemma \ref{lemma:adapted_fricke_gram_map}, their Gram matrix is consequently
\be
\widehat{\mathcal G}_{\mathbf t}
{=}S_P\mathcal G_{\mathbf t}S_P
=D_{|d|}GD_{|d|}\,,
\label{eq:character_sign_selected_gram_congruence}
\ee
Thus $\widehat{\mathcal G}_{\mathbf t}$ is positively diagonally congruent
to the Gram matrix of outward representatives $n_i$.  Moreover,

{\be\begin{split}
A_{\mathbf t}&=
-D_P\operatorname{adj}(\mathcal G_{\mathbf t})D_P=
-\Delta_{\mathbf t}
D_P\mathcal G_{\mathbf t}^{-1}D_P=
-\Delta_{\mathbf t}
D_{|P|}S_P
\mathcal G_{\mathbf t}^{-1}
S_PD_{|P|}\\
&=
-\Delta_{\mathbf t}
D_{|P|}
\widehat{\mathcal G}_{\mathbf t}^{-1}
D_{|P|}=|\Delta_{\mathbf t}|D_{|P|}
\left[-\sgnop(\Delta_{\mathbf t})
\widehat{\mathcal G}_{\mathbf t}^{-1}\right]D_{|P|}\,,
\end{split}
\label{eq:character_polynomial_outward_congruence}
\ee
where $D_{|P|}:=\diag(|P_i|)$. We have used $\operatorname{adj}(\mathcal G_{\mathbf t}) =\Delta_{\mathbf t}\mathcal G_{\mathbf t}^{-1}$ in the second equality, $D_P=D_{|P|}S_P=S_PD_{|P|}$ in the third equality and $-\Delta_{\mathbf t}=|\Delta_{\mathbf t}|\bigl[-\operatorname{sgn}(\Delta_{\mathbf t})\bigr]$ in the last equality. }
  Positive diagonal
congruence preserves the nonnegative orthant, so this proves that strict
copositivity of $A_{\mathbf t}$ is exactly the finite domain condition \eqref{eq:theorem_strict_copositivity}.
It also leaves the diagonal normalization \eqref{eq:normalized_domain_form}
unchanged.

We now translate the three conditions of Lemma
\ref{lemma:finite_domain_admissibility_criterion}. First, the diagonal cofactors and
\eqref{eq:character_principal_triples} give
\be
(A_{\mathbf t})_{ii}=-P_i^2m_i=4P_i^4\,.
\label{eq:character_polynomial_diagonal}
\ee
Consequently, diagonal normalization of $A_{\mathbf t}$ gives ({\it  cf.} \eqref{eq:normalized_domain_form})
\be
B_{ij}=\frac{(A_{\mathbf t})_{ij}}{4P_i^2P_j^2}\,.
\ee
Thus the edge test $B_{ij}>-1$ is $E_{ij}>0$.  For
$\widehat\ell=\{i,j,k\}$, clearing the positive denominator in
$L_\ell=1+B_{ij}+B_{ik}+B_{jk}$ gives
\be
\Lambda_\ell=4P_i^2P_j^2P_k^2L_\ell\,.
\ee
Apply Lemma \ref{lemma:finite_domain_admissibility_criterion} to
$\widehat{\mathcal G}_{\mathbf t}$.  Its diagonal entries are unchanged by
the sign selection, so
\be
(\widehat{\mathcal G}_{\mathbf t})_{\ell\ell}
=(\mathcal G_{\mathbf t})_{\ell\ell}=t_\ell^2-4\,.
\ee
Then the face condition
\be
(\widehat{\mathcal G}_{\mathbf t})_{\ell\ell}<0
\quad\hbox{or}\quad
L_\ell>0
\ee
is equivalent to
\be
t_\ell^2<4
\quad\hbox{or}\quad
\Lambda_\ell>0\,.
\ee

Finally,
\be
\det\widehat{\mathcal G}_{\mathbf t}
=
\Delta_{\mathbf t}\,,
\qquad
(\widehat{\mathcal G}_{\mathbf t})_{rs}
=
\sgnop(P_rP_s)(\mathcal G_{\mathbf t})_{rs}\,.
\ee
Hence the interior condition of Lemma
\ref{lemma:finite_domain_admissibility_criterion} is equivalent to
\be
\Delta_{\mathbf t}>0
\quad\hbox{or}\quad
P_rP_s(\mathcal G_{\mathbf t})_{rs}>0
\quad\hbox{for some }r,s\,.
\ee
These are precisely
\eqref{eq:character_edge_domain_tests}--\eqref{eq:character_full_domain_test}. 
\end{proof}

The three groups of conditions \eqref{eq:character_edge_domain_tests}--\eqref{eq:character_full_domain_test} in Proposition
\ref{prop:character_polynomial_domain_criterion} test the edges, triangular
faces, and full interior of the coefficient simplex, respectively.  Their
causal simplifications are summarized below. This is the same as the table after Lemma \ref{lemma:finite_domain_admissibility_criterion}, but now in the character-polynomial language. Again, ``automatic'' here means that
no additional inequality of that type is required.

\begin{center}
\begingroup
\small
\setlength{\tabcolsep}{3pt}
\renewcommand{\arraystretch}{1.25}
\begin{tabular}{@{}p{0.22\linewidth}p{0.18\linewidth}
p{0.29\linewidth}p{0.23\linewidth}@{}}
\hline
\raggedright\textbf{Character sector}
&\raggedright\textbf{Edges \eqref{eq:character_edge_domain_tests}}
&\raggedright\textbf{Faces \eqref{eq:character_face_domain_tests}}
&\raggedright\textbf{Interior \eqref{eq:character_full_domain_test}}
\tabularnewline
\hline

\raggedright AdS: $\Delta_{\mathbf t}>0$
&\raggedright Check all $E_{ij}>0$
&\raggedright Check $\Lambda_\ell>0$ only for $|t_\ell|\geq2$
&\raggedright Automatic
\tabularnewline

\raggedright dS with some $|t_\ell|>2$
&\raggedright Check all $E_{ij}>0$
&\raggedright Check $\Lambda_\ell>0$ only for $|t_\ell|\geq2$
&\raggedright Automatic
\tabularnewline

\raggedright All-elliptic dS: $|t_\ell|<2$ for every $\ell$
&\raggedright Automatic
&\raggedright Automatic
&\raggedright Check $P_iP_j(\mathcal G_{\mathbf t})_{ij}>0$ for some $i<j$
\tabularnewline

\raggedright All-parabolic dS: $|t_\ell|=2$ for every $\ell$
&\raggedright Automatic
&\raggedright Check all $\Lambda_\ell>0$
&\raggedright Check $P_iP_j(\mathcal G_{\mathbf t})_{ij}>0$ for some $i<j$
\tabularnewline

\raggedright Mixed elliptic/parabolic dS
&\raggedright Automatic
&\raggedright Check $\Lambda_\ell>0$ for the parabolic labels
&\raggedright Check $P_iP_j(\mathcal G_{\mathbf t})_{ij}>0$ for some $i<j$
\tabularnewline

\hline
\end{tabular}
\endgroup
\end{center}
Note that the table applies only when $\Delta_{\mathbf t}\neq0$ and every $P_i\neq0$.
It is not used on null-candidate strata or on the rank-three wall, which we will consider in Section \ref{subsec:character_boundary_strata}.

{Fix four noncentral conjugacy classes $\mathcal C_i\subset\SL(2,\mathbb R)$, and let $\mathcal Y\subset\mathfrak F_{\mathbf t}(\mathbb R)$ be the trace-coordinate image of a connected component of the corresponding relative $\SL(2,\mathbb R)$ character variety.}
\begin{theorem}[Geometric reconstruction chambers of the relative character variety]
\label{theorem:character_geometric_chambers}
 The finite nondegenerate de Sitter and anti-de Sitter reconstruction loci in $\mathcal Y$ are respectively
\be
\mathcal X_{\dS}
:=
\left\{
\,(x,y,z)\in\mathcal Y
\,\middle|\,
m_i\neq0\ \text{for all }i\,,\quad
\Delta_{\mathbf t}<0\,,\quad A_{\mathbf t}\ \text{strictly copositive}
\right\}\,,
\label{eq:character_dS_chamber}
\ee
and
\be
\mathcal X_{\AdS}
:=
\left\{
\,(x,y,z)\in\mathcal Y
\,\middle|\,
m_i\neq0\ \text{for all }i\,,\quad
\Delta_{\mathbf t}>0\,,\quad A_{\mathbf t}\ \text{strictly copositive}
\right\}\,.
\label{eq:character_AdS_chamber}
\ee
The nondegenerate flat-normal locus is
{\be
\mathcal X_{\rm flat}
:=
\left\{
\,(x,y,z)\in\mathcal Y
\,\middle|\,
m_i\neq0\ \text{for all }i\,,\quad
\Delta_{\mathbf t}=0
\right\}\,.
\label{eq:character_flat_normal_locus}
\ee}
Each of these loci may be empty.
Every point of $\mathcal X_{\dS}$ or
$\mathcal X_{\AdS}$ reconstructs a finite strictly convex nondegenerate
tetrahedron in $\dS^3$ or $\AdS^3$, respectively. Every point of
$\mathcal X_{\rm flat}$ instead determines {a finite strictly convex tetrahedron uniquely up to
ambient isometry.} All three conditions descend to
the corresponding liftable relative $\PSL(2,\mathbb R)$ character variety.

\end{theorem}

\begin{proof}
Let $(x,y,z)$ belong to one of the three displayed loci.
Choose a closing spin representative $\{H_i\}$
in the prescribed noncentral conjugacy classes. Its projections
$O_i\in\PSL(2,\mathbb R)\cong\SO^+(1,2)$ are nontrivial and satisfy closure.
By Lemma \ref{lemma:adapted_fricke_gram_map},
{\be
\mathcal G_{\mathbf t}=DGD\, {,}
\ee}
for an invertible real diagonal matrix $D$.  Thus $\mathcal G_{\mathbf t}$ and the reconstructed Gram matrix $G$ have the same inertia and determinant sign, and corresponding principal minors have the same signs and zero loci.

If $\Delta_{\mathbf t}\neq0$, the conditions $m_i\neq0$ imply that every $G_{\hat i}$ is nondegenerate.  By Lemma \ref{lemma:character_principal_minor_squares}, they are equivalently the inequalities $m_i<0$.  Lemma \ref{lemma:lorentzian_gram_signature} gives $\Inertia(G_{\hat i})=(1,2)$ and $\Inertia(G)\in\{(1,3),(2,2)\}$.  Lemma \ref{lemma:vertex_reality} therefore gives four finite vertices, while the sign of $\Delta_{\mathbf t}$ selects the de Sitter or anti-de Sitter ambient model.
  Lemma \ref{lemma:branch_selection} selects the outward representatives.  By \eqref{eq:character_polynomial_domain_form} and the positive-congruence observation following it, strict copositivity of $A_{\mathbf t}$ is equivalent to strict copositivity of their signed inverse Gram form.  Theorem \ref{theorem:unified_minkowski} therefore reconstructs the corresponding finite strictly convex tetrahedron.

If $\Delta_{\mathbf t}=0$, the nonzero principal minors imply
$\operatorname{rank}\mathcal G_{\mathbf t}=3$, with nonzero inertia $(1,2)$.
A Lorentzian Gram factorization therefore gives four normal representatives
$u_i\in\mathbb R^{1,2}$, every three of which are independent.
Using the orientation of $\mathbb R^{1,2}$ fixed by $\varepsilon_{012}=+1$, define
the affine cofactor coefficients
\be
a_i:=(-1)^{i+1}
\det(u_1,\ldots,\widehat{u_i},\ldots,u_4)\,.
\label{eq:character_flat_affine_cofactors}
\ee
Set directly
\be
\mathbf E_i:=a_iu_i\,.
\ee
The standard cofactor relation and the Lorentzian Gram identity give
\be
\sum_{i=1}^4\mathbf E_i=0\,,
\qquad
\operatorname{adj}(\mathcal G_{\mathbf t})=-aa^{\top}\,,
\qquad
a_i^2=-m_i=4P_i^2\,,
\label{eq:character_flat_coarea_vectors}
\ee
{where the last identity comes from Lemma \ref{lemma:character_principal_minor_squares}. }
Indeed, replacing $u_i$ by $\delta_i u_i$ multiplies every
$\mathbf E_i$ by the same sign $\prod_j\delta_j$, and no separate outward
selection is needed in the flat construction.
In particular, the four $\mathbf E_i$ are nonzero and every three are
linearly independent.  The Euclidean Minkowski argument of Section
\ref{sec:flat_limit}, applied via
$\mathbf E_i^{\rm E}:=\eta_T\mathbf E_i$, therefore
reconstructs a unique bounded affine Lorentzian tetrahedron up to translation,
with area vectors $\mathbf E_i$.
{Replacing $u_i$ by $Lu_i$, $L\in O(1,2)$, sends
$\mathbf E_i$ to $\det(L)L\mathbf E_i$ and changes the
tetrahedron only by a Lorentzian isometry.}

Conversely, the normal Gram matrix of a nondegenerate
strictly convex curved tetrahedron, or the area-normal Gram matrix of a
bounded nondegenerate affine Lorentzian tetrahedron, has $m_i\neq0$ for every
$i$, equivalently $m_i<0$ by Lemma
\ref{lemma:character_principal_minor_squares}. Its determinant is negative,
positive, or zero in the de Sitter, anti-de Sitter, or affine flat case,
respectively. In the two curved cases, every nonzero ray of the simplicial
cone of a finite tetrahedron has the quadric sign $\sigma_G$. The calculation
in the proof of Lemma \ref{lemma:simplicial_halfspace_realization} therefore
gives strict copositivity of the outward signed inverse Gram form, hence of
$A_{\mathbf t}$. This proves the three characterizations of
$\mathcal X_{\dS}$, $\mathcal X_{\AdS}$, and $\mathcal X_{\rm flat}$.

Finally,
\eqref{eq:character_gram_lift_action} preserves $\Delta_{\mathbf t}$ and every
$m_i$, while $P_i\mapsto\delta_iP_i$ makes
\eqref{eq:character_polynomial_domain_form} invariant, which proves the descent
statement.
\end{proof}

\subsection{Null candidate vertex lines and the determinant wall}
\label{subsec:character_boundary_strata}

{Theorem
\ref{theorem:character_geometric_chambers} treats the loci on which every
$m_i$ is nonzero.  }
It remains to treat null candidate vertex lines when
$\Delta_{\mathbf t}\neq0$, namely when at least one $m_i=0$, and the
determinant wall $\Delta_{\mathbf t}=0$.
We first treat $\Delta_{\mathbf t}\neq0$ and at least one $m_i$ vanishes.

\begin{lemma}%
\label{lemma:character_candidate_vertex_lines}
Let $(x,y,z)\in\mathcal Y$ satisfy $\Delta_{\mathbf t}\neq0$ and
$m_k=0$ for at least one $k$, and set
\be
\sigma_{\mathcal G}:=-\sgnop(\Delta_{\mathbf t})\,.
\ee
Then $\mathcal G_{\mathbf t}$ has inertia $(1,3)$ when
$\Delta_{\mathbf t}<0$ and $(2,2)$ when
$\Delta_{\mathbf t}>0$, and hence admits an ambient
realization by linearly independent vectors
$\widetilde N_i\in\mathbb V_{\sigma_{\mathcal G}}$ satisfying
{\be
\ambip{\widetilde N_i}{\widetilde N_j}
=(\mathcal G_{\mathbf t})_{ij}\,.
\label{eq:character_ambient_gram_realization}
\ee}
If $\widetilde W_i$
is the dual basis, defined by
$\ambip{\widetilde W_i}{\widetilde N_j}=\delta_{ij}$, then
the candidate vertex line
opposite the $i$-th face is $\mathbb R\widetilde W_i$, and
{\be
\ambip{\widetilde W_i}{\widetilde W_i}
=\frac{m_i}{\Delta_{\mathbf t}}\,.
\label{eq:character_candidate_vertex_norm}
\ee}
Thus
$m_i=0$ exactly when the candidate vertex line opposite $F_i$ is null, while $m_i<0$ gives a finite one.
\end{lemma}
\begin{proof}
Each principal submatrix $(\mathcal G_{\mathbf t})_{\hat i}$ is a Gram
matrix in the Lorentzian tangent model, and Lemma
\ref{lemma:character_principal_minor_squares} gives $m_i\leq0$.  The Cauchy
interlacing argument used in the proof of Lemma
\ref{lemma:lorentzian_gram_signature}, now allowing degenerate principal
submatrices,  gives at most two
negative and at most three positive eigenvalues for $\mathcal G_{\mathbf t}$.
If $\Delta_{\mathbf t}<0$, the number of negative
eigenvalues is odd and hence equals one. If $\Delta_{\mathbf t}>0$, it is
even; it cannot vanish because the positive-definite case would have every
$3\times3$ principal minor positive, whereas $m_k=0$. It therefore equals
two. 
Lemma \ref{lemma:ambient_gram_reconstruction}
therefore  gives
\eqref{eq:character_ambient_gram_realization}.
The congruence $\mathcal G_{\mathbf t}=DGD$ of Lemma
\ref{lemma:adapted_fricke_gram_map} identifies these representatives with
$\widetilde N_i=d_iN_i$ after choosing compatible ambient realizations. Thus
they define the same face hyperplanes as the geometric ambient normals
$N_i$.

The dual-basis relation $\ambip{\widetilde W_i}{\widetilde N_j}=\delta_{ij}$ makes $\mathbb R\widetilde W_i$ the common
intersection of the three ambient face hyperplanes opposite
the $i$-th face.
Moreover,
{\be
\ambip{\widetilde W_i}{\widetilde W_j}
=\bigl(\mathcal G_{\mathbf t}^{-1}\bigr)_{ij}\,,
\qquad
\bigl(\mathcal G_{\mathbf t}^{-1}\bigr)_{ii}
=\frac{\det(\mathcal G_{\mathbf t})_{\hat i}}
{\det\mathcal G_{\mathbf t}}
=\frac{m_i}{\Delta_{\mathbf t}}\,,
\ee}
which proves \eqref{eq:character_candidate_vertex_norm}.
Since $m_i\leq0$ by Lemma
\ref{lemma:character_principal_minor_squares}, a nonzero ratio has sign
$\sigma_{\mathcal G}$ and $\widetilde W_i$ can be normalized into
$\mathcal M_{\sigma_{\mathcal G}}$, while a zero ratio makes $\widetilde W_i$ null.
\end{proof}

 This lemma determines the candidate
vertex line types, but
the
domain condition remains to be checked.
 Choose signs
$\varepsilon_i\in\{\pm1\}$ selecting one ray on each
candidate vertex line, and set
\be
S_\varepsilon:=\operatorname{diag}(\varepsilon_1,\ldots,\varepsilon_4)\,,
\qquad
V_i:=\varepsilon_i\widetilde W_i\,,
\qquad
A_\varepsilon
:=\sigma_{\mathcal G}S_\varepsilon {\mathcal G_{\mathbf t}^{-1}} S_\varepsilon\,.
\label{eq:character_signed_vertex_gram}
\ee
With
$\widehat N_i:=-\varepsilon_i\widetilde N_i$, the dual-basis calculation in
the proof of Lemma \ref{lemma:simplicial_halfspace_realization} gives
\be
C_\varepsilon
:=\left\{\sum_i\lambda_iV_i\ \middle|\ \lambda_i\geq0\right\}
=\left\{X\ \middle|\
\ambip{X}{\widehat N_i}\leq0\ \text{for every }i\right\}\,,
\qquad
\sigma_{\mathcal G}\ambip{X(\lambda)}{X(\lambda)}
=\lambda^\top A_\varepsilon\lambda\,.
\ee
Hence the candidate cone lies in the closure of the
selected domain exactly when $A_\varepsilon$ is copositive:
\be
\lambda^\top A_\varepsilon\lambda\geq0
\qquad
\text{for every }\lambda\in\mathbb R_{\geq0}^4\,.
\label{eq:character_copositive_condition}
\ee
To avoid additional null directions, equality must
occur only on the coordinate rays corresponding to the  candidate lines.
This non-strict boundary condition is not implied by the $m_i$ alone. Appendix
\ref{subsec:character_copositivity_counterexample} gives a closing character
point for which no sign choice satisfies it. These boundary configurations,
and exact  holonomy matching, remain outside the reconstruction theorem.

It remains to treat the determinant wall
\be
\Delta_{\mathbf t}=0
\label{eq:character_model_wall}
\ee
separating the two finite-curvature
model signs. We call its rank-three part the flat-normal wall, because it
admits a Lorentzian Gram factorization by four vectors in
$\mathbb R^{1,2}$. The relatively open stratum on which every $m_i\neq0$ is
precisely $\mathcal X_{\rm flat}$ of Theorem
\ref{theorem:character_geometric_chambers}. The following lemma demonstrates
the geometries of the different parts of the determinant wall.

\begin{lemma}%
\label{lemma:character_determinant_wall}

Let $(x,y,z)\in\mathcal Y$ satisfy $\Delta_{\mathbf t}=0$.
If $\operatorname{rank}\mathcal G_{\mathbf t}=3$, let
$0\neq a=(a_1,\ldots,a_4)^\top\in\ker\mathcal G_{\mathbf t}$. Then there
is a constant $c<0$ such that
\be
m_i=ca_i^2
\qquad\text{for every }i\,.
\label{eq:character_wall_kernel_minors}
\ee
For any Lorentzian Gram factorization
$\mathcal G_{\mathbf t}=(\tanip{u_i}{u_j})$, all linear relations among
the $u_i$ are multiples of
\be
\sum_{i=1}^4a_iu_i=0\,.
\ee
Consequently, $(x,y,z)\in\mathcal X_{\rm flat}$ exactly when every $a_i$
is nonzero. If some $a_i=0$, no closure relation has four nonzero
coefficients, so the Gram data determine no bounded nondegenerate affine
tetrahedron.

On the whole determinant wall,
\be
m_i=0\ \text{for every }i
\quad\Longleftrightarrow\quad
\operatorname{rank}\mathcal G_{\mathbf t}\leq2\,.
\label{eq:character_wall_lower_rank}
\ee

\end{lemma}

\begin{proof}
In the rank-three case, $\operatorname{adj}(\mathcal G_{\mathbf t})$ is
a nonzero symmetric rank-one matrix with image
$\ker\mathcal G_{\mathbf t}=\mathbb Ra$. Hence
$\operatorname{adj}(\mathcal G_{\mathbf t})=caa^\top$ for some
$c\neq0$. Its diagonal entries are the $m_i$, giving
\eqref{eq:character_wall_kernel_minors}. Lemma
\ref{lemma:character_principal_minor_squares} gives $m_i\leq0$, and at
least one $m_i$ is nonzero, so $c<0$.
In particular,
$m_i=0$ if and only if $a_i=0$.

Lemmas \ref{lemma:adapted_fricke_gram_map} and
\ref{lemma:lorentzian_gram_signature} give the rank-three Lorentzian Gram
factorization
$\mathcal G_{\mathbf t}=U^\top\eta_TU$, with columns $u_i$. Since
$\operatorname{rank}U=\operatorname{rank}\mathcal G_{\mathbf t}=3$, the
inclusion $\ker U\subseteq\ker\mathcal G_{\mathbf t}$ is an equality. Hence
$Ua=\sum_i a_iu_i=0$. By the definition \eqref{eq:character_flat_normal_locus} of $\mathcal X_{\rm flat}$ and
\eqref{eq:character_wall_kernel_minors}, the point belongs to
$\mathcal X_{\rm flat}$ precisely when every $a_i\neq0$. In that case,
changing the signs of the representatives makes all coefficients in the
normal relation positive.

If some $a_i=0$, no branch choice can turn this
relation into a strictly positive four-normal closure, as required for a
bounded nondegenerate flat tetrahedron. The resulting flat normal
configuration is therefore degenerate or unbounded. If the kernel is spanned by $e_k$, then
$\mathcal G_{\mathbf t}e_k=0$, so the $k$-th row and column vanish. In a
rank-three factorization the vectors $u_i$ span $\mathbb R^{1,2}$, and the
vanishing column gives $u_k=0$. Thus the flat Gram data determine no normal
for the $k$-th face.

Finally, the principal-rank theorem for real symmetric matrices states that
the rank is the maximum order of a nonsingular principal submatrix \cite{HornJohnson:2013}. Thus the
vanishing of $\Delta_{\mathbf t}$ and all four order-three principal minors is
equivalent to rank at most two.

\end{proof}

The term ``flat-normal wall'' must not be confused with a genuine curvature-zero limit of the holonomies.  Such a limit requires a family for which
\be
H_i(R)\longrightarrow\pm\id\,,
\qquad
t_i(R)\longrightarrow\pm2\,,
\qquad
R\longrightarrow\infty\,,
\label{eq:character_genuine_flat_limit}
\ee
and whose first nontrivial terms give
finite flat area-normal vectors
as in Section \ref{sec:flat_limit}.  Except for the special boundary traces $t_i=\pm2$ for every $i$, such a family varies the boundary conjugacy classes and therefore moves between relative character surfaces.  Noncentral parabolic conjugacy classes may remain fixed as their representatives approach $\pm\id$. Thus, on a fixed generic character surface, $\Delta_{\mathbf t}=0$ is a degenerate wall of the finite-curvature ambient reconstruction even when its rank-three part admits a flat-normal interpretation.

Theorem \ref{theorem:character_geometric_chambers} and Lemmas \ref{lemma:character_candidate_vertex_lines} and \ref{lemma:character_determinant_wall} therefore give a Gram-data hierarchy covering every point of the selected noncentral character locus $\mathcal Y$:
\begin{center}
\renewcommand{\arraystretch}{1.25}
\begin{tabular}{c|c|p{8.2cm}}
$\Delta_{\mathbf t}$ & Principal minors $m_i$ & Geometry determined by the Gram data\\
\hline
$<0$ & $m_i\neq0$ for every $i$
& Finite nondegenerate $\dS^3$ tetrahedron if $A_{\mathbf t}$ is strictly copositive;  otherwise outside the tetrahedron reconstruction locus\\
$>0$ & $m_i\neq0$ for every $i$
& Finite nondegenerate $\AdS^3$ tetrahedron if $A_{\mathbf t}$ is strictly copositive; otherwise outside the tetrahedron reconstruction locus\\
$\neq0$ & $m_i=0$ for at least one $i$
& Face-plane configuration with null candidate vertices opposite those
$F_i$ for which $m_i=0$\\
$=0$ & $m_i\neq0$ for every $i$
& Bounded nondegenerate affine Lorentzian tetrahedron  \\
$=0$ & Some, but not all, $m_i$ vanish
& Degenerate or unbounded rank-three flat normal configuration\\
$=0$ & $m_i=0$ for every $i$
& Rank-at-most-two normal collapse\\
\end{tabular}
\end{center}
 In the $\Delta_{\mathbf t}\neq0$ boundary row, \eqref{eq:character_copositive_condition} tests whether the candidate cone stays in the closure of the selected domain, but these null-boundary points remain outside the tetrahedron reconstruction theorem.  In the first two rows, strict copositivity is what upgrades four finite candidate vertices to the finite tetrahedron of the reconstruction theorem.  The table gives a complete interpretation at the level of the reconstructed Gram data.

For fixed boundary traces, these strata can be determined explicitly.  Solving \eqref{eq:adapted_fricke_cubic} for $z$ gives
\be
z_{\pm}(x,y)
=
\frac{t_1t_3+t_2t_4-xy
\pm\sqrt{\mathcal D_{\mathbf t}(x,y)}}{2}\,,
\label{eq:fricke_two_sheets}
\ee
where
\be
\begin{split}
\mathcal D_{\mathbf t}(x,y)
:={}&
\bigl(xy-t_1t_3-t_2t_4\bigr)^2\\
&-4\left[
x^2+y^2
-(t_1t_2+t_3t_4)x
-(t_1t_4+t_2t_3)y
+\sum_{i=1}^4t_i^2+t_1t_2t_3t_4-4
\right]\,.
\end{split}
\label{eq:fricke_projection_discriminant}
\ee
The real Fricke surface $\mathfrak F_{\mathbf t}(\mathbb R)$ projects to the region $\mathcal D_{\mathbf t}(x,y)\geq0$.  Restricting this projection to $\mathcal Y$ and substituting $z_\pm(x,y)$ into $\Delta_{\mathbf t}$ and the four $m_i$ gives a direct sign diagram of the de Sitter, anti-de Sitter, flat-normal, null candidate-vertex, and more degenerate strata for the chosen real $\SL(2,\mathbb R)$ boundary data.
Applying Proposition \ref{prop:character_polynomial_domain_criterion} then selects the actual finite curved chambers.

Figure~\ref{fig:fricke_chambers_t3456}
illustrates this construction for the hyperbolic boundary traces
$\mathbf t=(3,4,5,6)$.  Its two panels are the sheets $z=z_-(x,y)$ and
$z=z_+(x,y)$ of the real Fricke surface.  The diagram makes visible that
Fricke reality and the sign of $\Delta_{\mathbf t}$ do not by themselves
imply tetrahedral reconstruction: on the curved finite locus, the additional
inequality $\mu_{\mathbf t}>0$ selects the domain-admissible subchambers,
while $P_i=0$ marks the null-candidate strata.

\begin{figure}[h!]
\centering
\includegraphics[width=\textwidth]{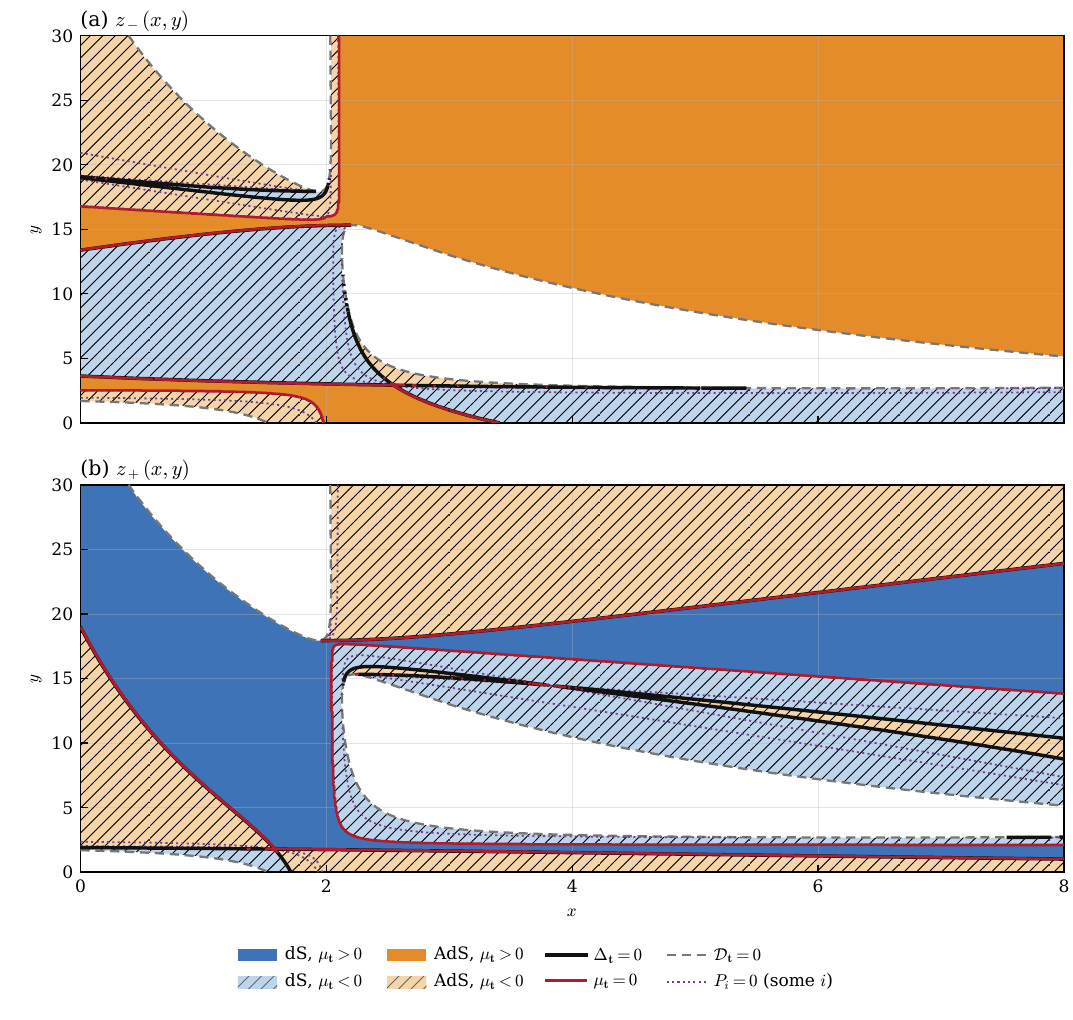}
\caption{Projected chamber diagram for the fixed
hyperbolic boundary traces $\mathbf t=(3,4,5,6)$.  The upper and lower panels
show the branches $z=z_-(x,y)$ and $z=z_+(x,y)$, respectively.  The gray
dashed curves bound the projected real $\SL(2,\mathbb R)$ character locus
$\mathcal D_{\mathbf t}\geq0$.  Within it, the solid black curves
$\Delta_{\mathbf t}=0$ separate the de Sitter and anti-de Sitter sides, and
the solid red curves $\mu_{\mathbf t}=0$ separate domain-admissible
($\mu_{\mathbf t}>0$) from domain-obstructed ($\mu_{\mathbf t}<0$) data.
Dark blue and light hatched blue mark these two regions on the de Sitter side;
dark orange and light hatched orange give the corresponding anti-de Sitter
regions.  The thin purple dotted curves are $P_i=0$, equivalently $m_i=0$,
and mark the null-candidate strata discussed in
Section~\ref{subsec:character_boundary_strata}.}
\label{fig:fricke_chambers_t3456}
\end{figure}
\section{Causal Sectors, Dual Tetrahedra, and Real Forms}
\label{sec:dual_tetra}

We now return to the reconstructed tetrahedron and organize the holonomy data by the causal type of its associated faces.  Since the face type is encoded by the sign of the extrinsic normal $N_i$, the normal line itself has a natural projective-dual interpretation as a vertex of a polar dual tetrahedron.  This dual viewpoint relates the holonomy reconstruction to three familiar geometric sectors: compact real-form curved tetrahedra in the all-spacelike case, hyperideal projective polar duals in the all-timelike AdS case, and ideal projective polar duals with framed parabolic data in the all-null AdS case. %

\subsection{Polar dual tetrahedra}
\label{subsec:polar_dual_ideal_hyperideal}

The definition of a tetrahedron is stronger than the linear independence of four ambient hyperplanes.  The extrinsic normals must be arranged so that the three hyperplanes opposite each label meet $\Msig$ in a real vertex $V_i$, equivalently in Lemma \ref{lemma:vertex_reality}, {and }the vertex line must admit a representative satisfying $\ambip{V_i}{V_i}=\sigma$.  For a general intersection of four independent hyperplanes, this need not hold: as explained in Remark \ref{remark:wrong_causal_vertex_line}, a candidate vertex line may have the wrong causal sign or be null.  In this section we define the polar dual of a tetrahedral configuration and show that, under this duality, the allowed causal types of the extrinsic face normals naturally produce ordinary, ideal, and hyperideal projective tetrahedra.

For a nonzero vector $X\in\Vsig$, write
\be
[X]:=\mathbb R^\times X\in\mathbb P(\Vsig)\,,
\qquad
\mathbb P(\Vsig):=(\Vsig\setminus\{0\})/\mathbb R^\times\,.
\ee
The projective model of $\Msig$ is the set of lines $[X]$ with $\sgnop\ambip{X}{X}=\sigma$, and its ideal boundary is the projectivized light cone
\be
\partial\mathbb P\Msig
:=\{[X]\in\mathbb P(\Vsig)\mid \ambip{X}{X}=0\}\,.
\label{eq:projective_boundary_lightcone}
\ee
The dual quadric is
\be
\Msig^\ast:=\{X\in\Vsig\mid \ambip{X}{X}=-\sigma\}\,.
\label{eq:dual_quadric_general}
\ee
Explicitly,
\be
\AdS^{3*}
=
\left\{
X\in\mathbb R^{2,2}\mid
-(X^0)^2-(X^1)^2+(X^2)^2+(X^3)^2=+1
\right\}\,,
\label{eq:AdS_dual_quadric}
\ee
while
\be
\dS^{3*}
=
\left\{
X\in\mathbb R^{1,3}\mid
-(X^0)^2+(X^1)^2+(X^2)^2+(X^3)^2=-1
\right\}\cong \mathbb{H}^3_+\sqcup \mathbb{H}^3_- \,.
\label{eq:dS_dual_quadric}
\ee
Thus $\AdS^{3*}$ is the space of spacelike unit normals to timelike totally geodesic planes in $\AdS^3$, and $\dS^{3*}$ is the two-sheeted hyperbolic three-space.  After choosing one sheet, $\dS^{3*}$ is the usual hyperbolic three-space in the Beltrami--Klein projective model.

\begin{definition}[Ordinary, ideal, and hyperideal projective vertices]
\label{def:ordinary_ideal_hyperideal_vertices}
Let $L=[X]\in\mathbb P(\Vsig)$ be a projective vertex line.  Relative to $\Msig$, we call $L$
\be
\begin{array}{c|c}
\hbox{ordinary} & \sgnop\ambip{X}{X}=\sigma\,,\\
\hbox{ideal} & \ambip{X}{X}=0\,,\\
\hbox{hyperideal} & \sgnop\ambip{X}{X}=-\sigma\,.
\end{array}
\label{eq:ordinary_ideal_hyperideal_vertex_types}
\ee
An {\bf ideal tetrahedron} is a projective tetrahedron whose four vertex lines are ideal.
A {\bf hyperideal projective tetrahedron} relative to $\Msig$ is a
projective tetrahedron whose four vertex lines are hyperideal. 

For $\sigma=-1$, a hyperideal projective tetrahedron is called a
{\bf standard hyperideal AdS tetrahedron} if its vertex lines admit
representatives $X_i\in\AdS^{3*}$ satisfying
\be
\ambip{X_i}{X_j}<-1
\qquad
(i\neq j)\,.
\label{eq:coherent_hyperideal_segment_condition}
\ee
For such a standard hyperideal AdS tetrahedron, the corresponding {\bf truncated hyperideal tetrahedron} is the region in $\AdS^3$ cut out by its four face planes and the truncation planes $X_i^\perp\cap\AdS^3$.

For $\sigma=+1$, after choosing one sheet of
$\dS^{3*}\cong\mathbb H^3_+\sqcup\mathbb H^3_-$, the four hyperideal
vertex lines determine an ordinary hyperbolic tetrahedron in the dual
space. 
\end{definition}

\begin{definition}[Projective dual tetrahedron]
\label{def:polar_dual_tetrahedron}
Let $\mathcal C=(N_1,\ldots,N_4)$ be a tetrahedral configuration in
$\Msig$.  For each $i$, let
\be
L_i
:=
\bigcap_{r\neq i}\Pi_r
=
\left(
\operatorname{span}\{N_r\mid r\neq i\}
\right)^\perp\,.
\ee
The projective dual of $\mathcal C$ is the ordered projective tetrahedron {configuration}
$T^\ast_{\mathbb P}$ with vertex lines
\be
L_i^\ast:=[N_i]\in\mathbb P(\Vsig)\,,
\qquad
i=1,\ldots,4\,.
\label{eq:polar_dual_vertices}
\ee
Its face {plane} opposite $L_i^\ast$ is
\be
F_i^\ast
:=
\mathbb P\!\left(
\operatorname{span}\{N_r\mid r\neq i\}
\right)
=
\mathbb P(L_i^\perp)\,.
\label{eq:polar_dual_faces}
\ee
If $\mathcal C$ underlies a tetrahedron in the sense of Definition
\ref{def:generalized_tetrahedron}, with selected vertex
$V_i\in L_i\cap\Msig$, then
\be
F_i^\ast=\mathbb P(V_i^\perp)\,.
\ee
\end{definition}

\begin{prop}
\label{prop:dual_vertex_causal_type}

Let $\mathcal C=(N_1,\ldots,N_4)$ be a tetrahedral configuration in
$\Msig$, with projective dual $T^\ast_{\mathbb P}$.
 The causal type of the support surface ${\widetilde{\Pi}_i}=N_i^\perp\cap\Msig$ determines the type of the corresponding dual vertex as follows:
\be
\begin{array}{c|c|c|c}
\hbox{support surface }{\widetilde{\Pi}_i} & \ambip{N_i}{N_i} & \sigma=-1\ (\AdS^3) & \sigma=+1\ (\dS^3)\\
\hline
\hbox{spacelike} & -1 & \hbox{ordinary} & \hbox{hyperideal}\\
\hbox{timelike} & +1 & \hbox{hyperideal} & \hbox{ordinary}\\
\hbox{null} & 0 & \hbox{ideal} & \hbox{ideal}
\end{array}
\label{eq:dual_face_vertex_type_table}
\ee
If all four support surfaces are null, then $T^\ast_{\mathbb P}$ is an
ideal tetrahedron.  If they are all timelike in $\AdS^3$, or all spacelike
in $\dS^3$, then $T^\ast_{\mathbb P}$ has four hyperideal vertex lines.
In the AdS case, it is a standard hyperideal AdS tetrahedron if and only if
these lines admit coherent representatives satisfying
\eqref{eq:coherent_hyperideal_segment_condition}.

If all four support surfaces are spacelike in $\AdS^3$, or all timelike in
$\dS^3$, then every dual vertex is ordinary.  After normalizing these
vertices in $\Msig$, the dual determines a tetrahedral configuration whose
face-normal lines are the original candidate vertex lines $L_i$.
\end{prop}

\begin{proof}
The dual vertex corresponding to ${\widetilde{\Pi}_i}$ is $[N_i]$.  A spacelike support surface has a timelike ambient normal, a timelike support surface has a spacelike ambient normal, and a null support surface has a null ambient normal.  Comparing the sign of $\ambip{N_i}{N_i}$ with the three alternatives in \eqref{eq:ordinary_ideal_hyperideal_vertex_types} gives the table and its ideal and hyperideal consequences.   The standard hyperideal AdS conclusion is
equivalent to the additional coherent segment condition
\eqref{eq:coherent_hyperideal_segment_condition}.  

If all the dual vertices are ordinary, they can be normalized to lie on $\Msig$. 
For each $i$, choose a nonzero
$W_i\in L_i$.  Since the $N_i$ form a basis, $W_i$ may be rescaled so that
\be
\ambip{W_i}{N_j}=\delta_{ij}\,.
\ee
Hence the $W_i$ form the dual basis and are linearly independent.
By
\eqref{eq:polar_dual_faces}, the face opposite $[N_i]$ has ambient normal
line $L_i=[W_i]$.
  Hence the ordinary dual is a tetrahedral configuration in $\Msig$.
\end{proof}

\subsection{All-spacelike sector and the \texorpdfstring{$\SU(2)$}{SU(2)} real form}
\label{subsec:all_spacelike_SU2}

We now isolate the sector in which every face is spacelike.  By the convention of Section \ref{sec:tetrahedra}, this means that every transported intrinsic normal is timelike:
\be
\tanip{n_i}{n_i}=-1\,,
\qquad
i=1,\ldots,4\,.
\ee
Thus all based face holonomies are elliptic,
\be
O_i=\exp\big(\Theta_iJ(n_i)\big)\,,
\label{eq:all_spacelike_elliptic}
\ee
and, choosing the principal spin lift, their spin representatives have the form
\be
H_i
=\cos\frac{\Theta_i}{2}\,\id
 +2\sin\frac{\Theta_i}{2}\,\mathcal T(n_i)\,.
\label{eq:all_spacelike_spin}
\ee
The parameter $\Theta_i$ is the holonomy angle of the Riemannian face.  For $\dS^3$ the spacelike faces have positive intrinsic curvature, while for $\AdS^3$ they have negative intrinsic curvature, so the sign convention for $\Theta_i$ changes exactly as in the spherical and hyperbolic cases of \cite{Haggard:2015ima}.

The comparison with Haggard--Han--Riello
\cite{Haggard:2015ima} is a change of real form, not a map from the actual
$\SL(2,\mathbb R)$ data to $\SU(2)$. This is because the actual Lorentzian holonomies
still live in $\SL(2,\mathbb R)$, and although each face holonomy is elliptic,
elliptic elements do not form a subgroup.  The formal replacement is %
\be
\big(\mathbb R^{1,2},\tanip{\cdot}{\cdot},\SO^+(1,2),\SL(2,\mathbb R)\big)
\longleftrightarrow
\big(\mathbb R^3,\cdot,\SO(3),\SU(2)\big)\,.
\label{eq:real_form_dictionary}
\ee
Under the replacement \eqref{eq:real_form_dictionary}, the algebraic reconstruction mechanism of Theorem \ref{theorem:unified_minkowski} becomes the compact curved Minkowski theorem for spherical or hyperbolic tetrahedra encoded by closing $\SU(2)$ holonomies in \cite{Haggard:2015ima}.  More concretely, the Lorentzian timelike unit normal $n_i$ is replaced by a Euclidean unit normal $\mathbf n_i$, and the corresponding compact face holonomy is
\be
R_i=\exp\big(\Theta_i\,\mathsf J_E(\mathbf n_i)\big)\,,
\qquad
R_i\in\SO(3)\,,
\label{eq:SO3_face_holonomy}
\ee
where $\mathsf J_E(\mathbf u)\mathbf v=\mathbf u\times_E\mathbf v$.  The simple-path convention is unchanged because it only depends on the labeled {tetrahedron}, so closure becomes
\be
R_4R_3R_2R_1=\id
\label{eq:SO3_closure}
\ee
in the Euclidean real form.  The Gram and orientation tests are the same formulas with the Euclidean inner product and cross product:
\be
\Gamma_{ij}=
\begin{cases}
\mathbf n_i\cdot\mathbf n_j\,, & (i,j)\neq(2,4),(4,2)\,,\\[0.12cm]
\mathbf n_2\cdot R_1\mathbf n_4
=\mathbf n_2\cdot R_3^{-1}\mathbf n_4\,, & (i,j)=(2,4),(4,2)\,,
\end{cases}
\label{eq:SU2_Gram}
\ee
and
\be
(\mathbf n_j\times_E\mathbf n_k)\cdot\mathbf n_l>0\,,
\qquad
\{i,j,k,l\}=\{1,2,3,4\}\,,
\label{eq:SU2_triple_signs}
\ee
with the same transported insertion for the special pair $(2,4)$.  Passing from $\SO(3)$ to its spin cover gives the $\SU(2)$ holonomies of \cite{Haggard:2015ima},
\be
U_i=\exp\big(\Theta_i\,\mathbf n_i\cdot\vec\tau_{\SU(2)}\big)\,,
\qquad
U_4U_3U_2U_1=\pm\id\,.
\label{eq:SU2_spin_closure}
\ee
The sign is the central lift of the identity in $\SO(3)$.  Thus the compact real-form statement has the same closure, Gram-matrix, branch-selection, and holonomy-matching structure as the Lorentzian theorem, but it is a different real form rather than a literal image of the original $\SL(2,\mathbb R)$ holonomies.  If one wants an honest $\SU(2)$ flat connection on $S_{0,4}$, the central signs of the lifts are chosen so that $U_4U_3U_2U_1=\id$, exactly as in the central-sign rule \eqref{eq:spin_lift_sign_change}.

\subsection{All-timelike AdS sector and its polar dual}
\label{subsec:all_timelike}

We next consider the sector in which the reconstructed model is $\AdS^3$ and all four faces of our tetrahedron are timelike.  Equivalently, in the notation of Theorem \ref{theorem:unified_minkowski}, this is the branch with $\det G>0$ and hence $\sigma=-1$.  Then the transported intrinsic normals are spacelike,
\be
\tanip{n_i}{n_i}=+1\,,
\qquad i=1,\ldots,4\,,
\ee
and the non-central spin holonomies are hyperbolic.  After choosing spin lifts, this means $|\trspin(H_i)|>2$.  By Proposition \ref{prop:dual_vertex_causal_type}, the polar dual has vertices in $\AdS^{3*}$.
It lies in the standard truncated-hyperideal locus only when
its four normal lines admit coherent lifts satisfying
\eqref{eq:coherent_hyperideal_segment_condition}.  This is the point of
contact with the hyperideal AdS tetrahedra and four-holed-sphere character
varieties studied in \cite{Liu:2025tzv}.  In this subsection we write
\be
(X,Y)_-:=X^\top\mathrm{diag}(-1,-1,1,1)Y\,{,}
\ee
for the ambient pairing of $\mathbb R^{2,2}$. 

The standard hyperideal dual sublocus of this all-timelike sector is
particularly relevant to recent developments connecting Virasoro TQFT
with three-dimensional quantum gravity at negative cosmological constant
\cite{CollierEberhardtZhang2023,CollierEberhardtZhang2024,
Hartman2025,Hartman2025Triangulating}. In triangulation-based
formulations, Virasoro fusion kernels and the closely related
modular-double $b$-$6j$ symbols provide tetrahedral amplitudes. Their
semiclassical asymptotics have recently been related directly to the
volumes and co-volumes of truncated hyperideal tetrahedra in Lorentzian
$\AdS^3$ \cite{Liu:2025tzv}. Hence, on the coherent segment-incidence
sublocus identified below, the polar dual of an all-timelike tetrahedron
belongs to the same geometric class that appears in these asymptotic
formulae. Note that this is a geometric point of contact rather than an identification of
the four face holonomies with the six $b$-$6j$ labels, or of the volume
of the primal tetrahedron with the volume or co-volume of its
hyperideal dual.

\begin{prop}%
\label{prop:edge_adapted_hyperideal_polar_locus}
Let $T\subset\AdS^3$ be an all-timelike finite tetrahedron and let $G$ be
the Gram matrix of unit spacelike representatives of its four face-normal
lines.  Its projective dual is a standard hyperideal AdS tetrahedron in the
edge-length convention of \cite{Liu:2025tzv} if and only if there is a sign
matrix
\be
S_{\varepsilon}=\diag(\varepsilon_1,\ldots,\varepsilon_4)
\ee
such that
\be
(S_{\varepsilon}GS_{\varepsilon})_{ij}<-1
\qquad(i\neq j)\,.
\label{eq:edge_adapted_hyperideal_test}
\ee
The signs are unique up to simultaneous reversal, and the six dual edge
lengths are
\be
\ell_{ij}=\operatorname{arcosh}
\bigl(-(S_{\varepsilon}GS_{\varepsilon})_{ij}\bigr)\,.
\label{eq:polar_edge_lengths_from_normal_gram}
\ee
Equivalently, \eqref{eq:edge_adapted_hyperideal_test} holds precisely when
$|G_{ij}|>1$ for every pair and
\be
\sgnop(G_{ij}G_{jk}G_{ki})=-1
\ee
for every triple of distinct labels.
\end{prop}

\begin{proof}
Replacing $N_i$ by $X_i=\varepsilon_iN_i$ does not change the
projective vertex $[N_i]$, and gives
\be
(X_i,X_j)_-
=
\varepsilon_i\varepsilon_jG_{ij}
=
(S_{\varepsilon}GS_{\varepsilon})_{ij}\,.
\ee
By Definition \ref{def:ordinary_ideal_hyperideal_vertices}, the polar dual is
a standard hyperideal AdS tetrahedron precisely when these representatives
can be chosen so that
\be
\varepsilon_i\varepsilon_jG_{ij}<-1
\qquad
(i\neq j)\,.
\ee

Put $s_{ij}:=\sgnop(G_{ij})$. Since $|G_{ij}|>1$, the preceding
inequalities are equivalent to
\be
\varepsilon_i\varepsilon_js_{ij}=-1
\qquad
(i\neq j)\,.
\label{eq:coherent_vertex_sign_equations}
\ee
Multiplying the three equations around a triangle gives
\be
s_{ij}s_{jk}s_{ki}=-1\,.
\ee
Conversely, suppose that this triangle condition holds. Set
$\varepsilon_1=1$ and
\be
\varepsilon_j=-s_{1j}
\qquad
(j=2,3,4)\,.
\ee
The equations involving vertex $1$ then hold by construction. For
$j,k\neq1$, the triangle condition for $(1,j,k)$ gives
\be
\varepsilon_j\varepsilon_ks_{jk}
=
s_{1j}s_{1k}s_{jk}
=-1\,.
\ee
Thus
\be
(S_{\varepsilon}GS_{\varepsilon})_{ij}
=
-|G_{ij}|<-1
\qquad
(i\neq j)\,,
\ee
which proves the equivalence.

If $\varepsilon_i'$ is another solution, then
$\delta_i:=\varepsilon_i'\varepsilon_i$ satisfies
$\delta_i\delta_j=1$ for every $i\neq j$. Hence all $\delta_i$ are
equal, so the vertex signs are unique up to simultaneous reversal.

Finally, for $c_{ij}:=(X_i,X_j)_-<-1$, the standard hyperboloid distance formula is
\be
c_{ij}=-\cosh\ell_{ij}\,,
\ee
which proves \eqref{eq:polar_edge_lengths_from_normal_gram}. 
\end{proof}

The converse requires a
separate global check: starting from a standard hyperideal edge length Gram
matrix, its ordinary dual candidate must satisfy the domain 
condition of Section \ref{sec:orientation_convexity}. The following
corollary shows that this condition is automatic. In fact, the
appropriate signed inverse Gram form is entrywise positive, which is
stronger than strict copositivity.

\begin{prop}
\label{cor:standard_hyperideal_polar_copositivity}
Let $G_\ell$ be the edge length Gram matrix of a non-degenerate hyperideal AdS tetrahedron.
Then {it admits a finite domain-contained AdS realization.}
\end{prop}

\begin{proof}
$G_\ell$ satisfies \cite{Liu:2025tzv}
\be
(G_\ell)_{ii}=1\,,
\qquad
(G_\ell)_{ij}=-\cosh\ell_{ij}<-1\,,
\qquad
\det G_\ell>0\,.
\ee
Then there are signs
$\varepsilon_i^{\rm pol}\in\{\pm1\}$, unique up to simultaneous
reversal, with exactly two positive and two negative signs, such that,
for
\be
S_{\rm pol}
:=
\diag\bigl(
\varepsilon_1^{\rm pol}\,,
\ldots\,,
\varepsilon_4^{\rm pol}
\bigr)\,,
\ee
the matrix
\be
-S_{\rm pol}G_\ell^{-1}S_{\rm pol}
\ee
is entrywise positive.
Its equal-sign pairs form one of the
three opposite-pair patterns, as determined by $G_\ell$.
Let $D=\diag(\sqrt{-\lb G_\ell^{-1}\rb_{ii}})$.  The normalized inverse
$U=D^{-1}G_\ell^{-1}D^{-1}$ has diagonal entries $-1$.  The cofactor formulas and the
three standard AdS dihedral-angle patterns give $|U_{ij}|>1$, with two
opposite pairings negative and the other four positive
\cite{Liu:2025tzv}.  Choose $S_{\rm pol}$ constant on the two negative
pairs and opposite across the remaining four pairs.  Then every entry of
$-S_{\rm pol}US_{\rm pol}$ is positive, and
\be
-S_{\rm pol}G_\ell^{-1}S_{\rm pol}
=D\bigl[-S_{\rm pol}US_{\rm pol}\bigr]D
\ee
is entrywise positive as well.  It is therefore strictly copositive.  If
$N_i$ have Gram matrix $S_{\rm pol}G_\ell S_{\rm pol}$ and $W_i$ are the
dual basis, the normalized choices $V_i=-W_i/\sqrt{-\lb G_\ell^{-1}\rb_{ii}}$ have negative
opposite supports.  Lemma \ref{lemma:simplicial_halfspace_realization}
then gives the finite dual tetrahedron.
\end{proof}

The inclusion supplied by {Proposition}
\ref{cor:standard_hyperideal_polar_copositivity} is strict. 
That is, a finite all-timelike AdS tetrahedron {is} not necessarily dual to a standard hyperideal AdS tetrahedron. For example, with $I_2$ denoting the $2\times2$ identity matrix, take
\be
G_*=
\begin{pmatrix}
I_2&-\sqrt2 I_2\\
-\sqrt2 I_2&I_2
\end{pmatrix},
\qquad
-G_*^{-1}=
\begin{pmatrix}
I_2&\sqrt2 I_2\\
\sqrt2 I_2&I_2
\end{pmatrix}.
\ee
Here $G_*$ has inertia $(2,2)$ and every vertex principal minor equals $-1$. Its signed inverse has nonnegative entries and positive diagonal, hence is strictly copositive. Lemma \ref{lemma:simplicial_halfspace_realization} therefore gives a finite all-timelike AdS tetrahedron, whereas $(G_*)_{12}=0$ excludes \eqref{eq:edge_adapted_hyperideal_test}.
}

Following \cite{Liu:2025tzv}, identify $\SL(2,\mathbb R)$ with $\AdS^{3*}$ by the linear map
\be
\Phi(M)=
\left(
\frac{a-d}{2}\,,
\frac{b+c}{2}\,,
\frac{a+d}{2}\,,
\frac{b-c}{2}
\right)\,,
\qquad
M=
\begin{pmatrix}
a&b\\ c&d
\end{pmatrix}
\in\SL(2,\mathbb R)\,.
\label{eq:Phi_SL2R_AdS_dual}
\ee
With our ambient pairing, this normalization gives
\be
(\Phi(A),\Phi(B))_-
=
\frac12\trspin(AB^{-1})\,.
\label{eq:Phi_trace_pairing}
\ee
The group-manifold identification $\Phi$ must not be confused
with the Lie-algebra identification $\mathcal T:\mathbb R^{1,2}\to
\mathfrak{sl}(2,\mathbb R)$ used earlier to extract invariant normal axes.
In the hyperbolic-surface convention of \cite{Liu:2025tzv}, let $A_1,A_2,A_3$ be boundary-curve lifts of an oriented hyperbolic four-holed sphere, chosen so that $\trspin(A_i)<-2$ for $i=1,2,3$, and set
\be
A_4=(A_3A_2A_1)^{-1}\,,
\qquad
A_5=(A_2A_1)^{-1}\,,
\qquad
A_6=(A_3A_2)^{-1}\,.
\label{eq:six_hyperbolic_curves}
\ee

The Fuchsian, maximal-relative-Euler-class hypothesis implies
that the induced lifts $A_4,A_5,A_6$ also satisfy
$\trspin(A_k)<-2$ \cite{Goldman1988}.  The four points
\be
v_1=\Phi(\id)\,,
\qquad
v_2=\Phi(A_2A_1)\,,
\qquad
v_3=\Phi(A_1)\,,
\qquad
v_4=\Phi(A_3A_2A_1)
\label{eq:hyperideal_vertices_from_SL2R}
\ee
then satisfy
\be
\begin{array}{lll}
(v_1,v_2)_-=\frac12\trspin(A_5)\,,&
(v_1,v_3)_-=\frac12\trspin(A_1)\,,&
(v_1,v_4)_-=\frac12\trspin(A_4)\,,\\[0.12cm]
(v_2,v_3)_-=\frac12\trspin(A_2)\,,&
(v_2,v_4)_-=\frac12\trspin(A_3)\,,&
(v_3,v_4)_-=\frac12\trspin(A_6)\,.
\end{array}
\label{eq:hyperideal_pairing_traces}
\ee
Hence $(v_i,v_j)_-<-1$ for all $i\neq j$.  The projective line segment joining $v_i$ and $v_j$ therefore meets the projective model of $\AdS^3$, and the four points determine a truncated hyperideal AdS tetrahedron \cite{Liu:2025tzv}.

The construction in
\eqref{eq:hyperideal_vertices_from_SL2R} is a standard Fuchsian
parametrization of a hyperideal tetrahedron, whereas the present
reconstruction first produces the polar lines $[N_i]$ from the invariant
axes of four different based face holonomies $H_i$.  

\subsection{All-null sector and its polar dual}
\label{subsec:all_null}

Finally, we consider the sector in which every face holonomy is nontrivial {and }parabolic.  In spin {representation}, this means
{\be
H_i=\epsilon_i\big(\id+\mathcal T(\kappa_i)\big)\,,
\qquad
\epsilon_i=\pm1\,,
\qquad
\tanip{\kappa_i}{\kappa_i}=0\,,
\label{eq:all_null_spin_data}
\ee}
as in \eqref{eq:spin_closed_form}, and equivalently, $|\trspin(H_i)|=2$.  For each face one may choose any nonzero null representative $n_i$ on the line $\mathbb R\kappa_i$ and write $\kappa_i=\Theta_i n_i$.

If all four faces are null, then the four dual vertices $[N_i]$ lie on the projectivized light cone \eqref{eq:projective_boundary_lightcone}.  Thus the polar dual tetrahedron of an all-null-face tetrahedron is an ideal tetrahedron. This relation between null tetrahedra and ideal tetrahedra is similar to that in \cite{MeusburgerScarinci2022}. 

There is a parallel statement on the character-variety side.  A nontrivial parabolic element of $\SL(2,\mathbb R)$ has a unique invariant line in $\mathbb R^2$, but for cluster and Chern--Simons coordinates one usually works with framed local systems: together with a representation $\rho:\pi_1(S_{0,4})\to\SL(2,\mathbb R)$, one chooses an invariant flag
\be
\xi_i\in\mathbb{R}\mathrm{P}^1\,,
\qquad
\rho(\gamma_i)\xi_i=\xi_i
\ee
at each hole {of $S_{0,4}$}.  The framed parabolic character space is therefore schematically
\be
\mathcal X_{\mathrm{fr}}(S_{0,4},\SL(2,\mathbb R))
=
\left\{
(\rho,\xi_1,\ldots,\xi_4)
\mid
\rho(\gamma_i)\ \hbox{is parabolic and}\ \rho(\gamma_i)\xi_i=\xi_i
\right\}/\SL(2,\mathbb R)\,.
\label{eq:framed_parabolic_character_space}
\ee

Here $\xi_i$ encodes the projective null-normal line, while the
nilpotent logarithm vector $\kappa_i$ is its trace representative
\cite{DimofteGaiottoVanderVeen:2015}.  Choosing $n_i$ and the compensating
coordinate $\Theta_i$ fixes a geometric representative without changing the
determinant or principal-minor tests.  Thus the all-null sector is naturally
real framed parabolic $\SL(2,\mathbb R)$ data.
After complexification to the usual $\PSL(2,\mathbb C)$
state-integral setting,  analogous complexified
ideal-tetrahedron data are described by shape parameters and gluing
equations. This is the standard geometric input behind the volume conjecture,
quantum hyperbolic invariants, and complex Chern--Simons state integrals
\cite{Kashaev:1997,MurakamiMurakami:2001,Gukov:2005,
BaseilhacBenedetti:2004,DimofteGaroufalidis:2013,
murakami2010introduction,costantino20076j}. In the present real theorem, this
comparison is only a guide: the actual null reconstruction data remain real
framed parabolic data together with their nilpotent logarithm vectors
$\kappa_i$.

\section{Discussion and Outlook}
\label{sec:conclusion_discussion}

{Theorems~\ref{theorem:unified_minkowski} and
\ref{theorem:character_geometric_chambers} relate finite Lorentzian
tetrahedra to their based face holonomies and characterize the
corresponding loci in the relative $\SL(2,\mathbb R)$ character variety
by polynomial inequalities. The reconstruction allows spacelike,
timelike, and null faces, is unique up to ambient isometry, and recovers
the prescribed Levi--Civita face holonomies exactly. Its essential global
ingredient is strict copositivity of the signed inverse Gram form on the
outward branch. This controls the whole tetrahedron, not just its
vertices. The degenerate Gram strata distinguish null candidate vertex lines,
 rank-three normal configurations, and lower-rank collapse.  }

The chamber description also raises global questions.
{The polynomial tests allow one to investigate nonemptiness,
connected components, and their changes as the boundary conjugacy
classes vary. The boundary of a finite reconstruction locus can arise
from failure of the strict domain condition even when the Gram matrix
and all vertex principal minors remain nonzero. It would be useful to
determine which limiting geometries occur at these different boundaries.
Extending exact holonomy reconstruction to null vertex lines or to
branches passing through infinity (\eg similar to the reconstruction in \cite{Haggard:2015ima}) would also require suitable boundary
and transport conventions.}

{Projective duality gives a geometric route to volume observables.
On the locus where the projective dual tetrahedron is a standard
hyperideal AdS tetrahedron, one can seek expressions for its volume,
co-volume, and their Schl\"afli variations in the original trace
coordinates \cite{milnor1994schlafli,luo20083,luo2014volume,bobenko2015discrete}. This would permit a direct comparison with the tetrahedral
geometry in modular-double $b$-$6j$ asymptotics
\cite{TeschnerVartanov:2014,Liu:2025tzv}. Such a comparison requires an
explicit relation between the original face holonomies and the variables
of the dual tetrahedron.}

The reconstruction chambers provide a natural classical domain for
quantizing Lorentzian curved tetrahedra.
{On the smooth locus represented by irreducible flat connections,
the relative character variety carries the Goldman--Atiyah--Bott
symplectic structure
\cite{Goldman:1984,AtiyahBott:1983,GuruprasadHuebschmannJeffreyWeinstein1997}
and is the reduced Chern--Simons phase space. The trace polynomials are
functions of Wilson-loop observables, so the theorem identifies which
classical boundary data admit a finite tetrahedral interpretation.
The compact construction using quantum-group intertwiners and coherent
states provides a useful comparison
\cite{HanHsiaoPan:2024,HsiaoPan:2025}. For the Lorentzian theory, a central
problem is to 
construct quantum states describing finite tetrahedra for which, in the
semiclassical regime, the spectral values of the relevant geometric operators
satisfy the reconstruction conditions,
and to implement their domain inequalities in a
noncommutative observable algebra.}

Four-holed-sphere flat connections also occur as local boundary phase
spaces in four-dimensional spinfoam models with nonzero cosmological constant
\cite{Han:2021tzw,Han:2025mkc}.
{When the relevant boundary data take values in the Lorentzian real
form, the present theorem supplies a test for a finite tetrahedral cell
with spacelike, timelike, or null faces. Applying it to several cells
requires compatibility of their shared geometry.}
A full spinfoam application must additionally impose gluing, shape
matching, and four-dimensional critical-point equations
\cite{Han:2024dpa,Pan:2025geo}. In such  cases, configurations in local boundary phase spaces will be integrated out so that semiclassical geometric solutions are determined by the critical-point equations.
{A concrete next step is to determine whether these
geometric solutions lie in the reconstruction loci found here and how
the domain inequalities constrain the allowed causal configurations.}

\section*{Acknowledgements}
The authors thank Shuang Ming and Diandian Wang for multiple inspiring and useful discussions, especially on classical and quantum hyperbolic geometry, $b$-$6j$ symbols and related topics. 
QP acknowledges the support from the Shuimu Tsinghua Scholar Program of Tsinghua University. HL acknowledges the support from the National Natural Science Foundation of China (Grant No. 12505081) and the start-up funding from Westlake University.

\appendix
\renewcommand\thesection{\Alph{section}}

\section{Coordinate-support domain obstructions}
\label{subsec:character_copositivity_counterexample}

{This appendix explains why the global domain condition in
Theorem~\ref{theorem:unified_minkowski} is needed in addition to
closure and the local Gram conditions. We first organize failures
of the domain condition according to their location in the coefficient
simplex. We then give an exact rational closing holonomy quadruple
whose candidate vertices are finite but whose domain condition fails
for every choice of vertex rays.}

\subsection{Coordinate-support obstruction}
\label{subsec:coordinate_support_obstructions}

{This subsection explains where the domain condition
\eqref{eq:strict_inverse_gram_copositivity} can fail in the coefficient
simplex, complementing the finite tests of
Lemma~\ref{lemma:finite_domain_admissibility_criterion}.}

{Let $N_1,\ldots,N_4\in\mathbb V_\sigma$ be linearly independent
normals with Gram matrix $G$, and let $W_i$ be their dual basis.
Choose signs $\varepsilon_i\in\{\pm1\}$ selecting the rays
$\mathbb R_{>0}(\varepsilon_iW_i)$, and let
$$
S_\varepsilon:=\diag(\varepsilon_1,\ldots,\varepsilon_4)\,,
\qquad
A:=\sigma S_\varepsilon G^{-1}S_\varepsilon\,,
$$
where $\sigma=-\sgnop(\det G)$.
This is the signed quadratic form used in
\eqref{eq:character_signed_vertex_gram}.
We do not assume that $A_{ii}>0$.}

{For $0\neq\lambda\in\mathbb R_{\geq0}^4$, put
$$
X(\lambda):=\sum_i\lambda_i\varepsilon_iW_i\,,
\qquad
\sigma\ambip{X(\lambda)}{X(\lambda)}
=\lambda^\top A\lambda\,.
$$
Positive, zero, and negative values correspond respectively to rays
in the selected domain, null rays, and rays outside that domain.}

  For each nonempty  {index set}
$I\subseteq\{1,2,3,4\}$ define
\be
\Delta_I:=\left\{\lambda\in\mathbb R_{\geq0}^4\ \middle|\
\lambda_j=0\ (j\notin I)\,,\quad
\sum_{i\in I}\lambda_i=1\right\}\,,
\qquad
c_I:=\min_{\lambda\in\Delta_I}\lambda^{\top}A\lambda\,.
\label{eq:appendix_support_minima}
\ee
{Here $I$ specifies the coordinates allowed to be nonzero, and
$\dim\Delta_I=|I|-1$. The $2^4-1=15$ nonempty index sets correspond to
four vertices, six edges, four triangular faces, and the full simplex.}
Strict copositivity is equivalent to $c_I>0$ for every nonempty $I$.

{Since $\Delta_J\subseteq\Delta_I$ for $J\subseteq I$, we have
$c_I\leq c_J$. The table considers $c_I\leq0$ with $c_J>0$ for every
nonempty proper subset $J\subsetneq I$. The failure therefore occurs
in the relative interior of $\Delta_I$. Each $I$ is tested separately,
although the table groups the outcomes by $|I|$.}

\begin{center}
\begingroup
\setlength{\tabcolsep}{4pt}
\renewcommand{\arraystretch}{1.2}
{\begin{tabular}{@{}c l l l@{}}
$|I|$ & $\Delta_I$ & $c_I=0$ & $c_I<0$\\
\hline
$1$ & vertex & null candidate vertex ray & candidate ray outside the domain\\
$2$ & edge & null ray in edge interior & negative interval in edge interior\\
$3$ & triangular face & null ray in face interior & negative region in face interior\\
$4$ & full simplex & excluded & negative interior region (dS only)
\end{tabular}}
\endgroup
\end{center}

{If $|I|=4$ and $c_I=0$, positivity on proper coordinate faces
places a minimizer $\lambda_*$ in the interior. Constrained minimization gives
$$
A\lambda_*=\nu\mathbf1\,,
\qquad
\nu=\lambda_*^\top A\lambda_*=0\,,
$$
contradicting the invertibility of $A$.
The table concerns domain obstructions, not additional reconstruction
statements.}

{On the closing $\SL(2,\mathbb R)$ character image with
$\Delta_{\mathbf t}\neq0$, the identity $m_i=-4P_i^2\leq0$ excludes
the negative vertex case. When every $P_i\neq0$, Proposition
\ref{prop:character_polynomial_domain_criterion} gives explicit domain
tests, with a possible additional interior obstruction only in dS.}

\subsection{An exact finite character obstruction}
\label{subsec:exact_finite_character_obstruction}

{We now show that} the strict domain condition in
\eqref{eq:character_dS_chamber}--\eqref{eq:character_AdS_chamber} is not a
consequence of closure, Gram nondegeneracy, and the four nonzero {vertex} principal
minors.  We give {an exact rational counterexample with four finite candidate vertices.}
{The following rational
$\SL(2,\mathbb R)$ quadruple has a nondegenerate de Sitter Gram
matrix and four finite candidate vertices, but no choice of
vertex rays gives a copositive domain form.
All inequalities below are verified by exact rational arithmetic.}

{For rational $a,b,c$ with $a\neq0$, put
}\begin{equation}
{M(a,b,c):=
\begin{pmatrix}
a&b\\
c&(1+bc)/a
\end{pmatrix}\in\SL(2,\mathbb Q)\,.
\label{eq:counterexample_Mabc}
}\end{equation}
  {Every terminating decimal in the next display denotes the corresponding
exact rational number, not a rounded matrix entry.  Define
}\begin{equation}
{\begin{aligned}
H_1&:=M(0.39595562,-1.83345258,1.55212433)\,,\\
H_2&:=M(-1.64287192,1.83989348,1.36514782)\,,\\
H_3&:=M(-15.87627484,0.97269495,-7.96947054)\,,\\
H_4&:=(H_3H_2H_1)^{-1}\,.
\end{aligned}
\label{eq:counterexample_exact_H}
}\end{equation}
  {The identity $a(1+bc)/a-bc=1$ proves } $\det H_i=1$   {for
$i=1,2,3$, and hence also for $H_4$.  By definition,
}\begin{equation}
{H_4H_3H_2H_1=\id
\label{eq:counterexample_exact_closure}
}\end{equation}
  {exactly.  Direct rational comparison gives
}\begin{equation}
{\begin{gathered}
-4.266<\trspin(H_1)<-4.265\,,
\qquad -3.781<\trspin(H_2)<-3.780\,,\\
-15.452<\trspin(H_3)<-15.450\,,
\qquad 9.807<\trspin(H_4)<9.808\,,
\end{gathered}
}\end{equation}
  {so all four matrices are noncentral.
}

{Construct the exact rational traces $\mathbf t,x,y,z$ and the matrix
$\mathcal G_{\mathbf t}$ from
\eqref{eq:adapted_fricke_coordinates} and
\eqref{eq:character_connected_gram}.  Integer arithmetic gives the exact {rational bounds}
}\begin{equation}
{-33328378<\Delta_{\mathbf t}<-33328377
\label{eq:counterexample_delta_interval}
}\end{equation}
 and
  \begin{equation}
{\begin{aligned}
-1469380&<m_1<-1469379\,,\\
-2470776&<m_2<-2470775\,,\\
-49013&<m_3<-49012\,,\\
-\frac{11}{10^7}&<m_4<-\frac{27}{25000000}\,.
\end{aligned}
\label{eq:counterexample_minor_intervals}
}\end{equation}
  {Every comparison here and below is exact: the terminating decimal endpoints
mean the corresponding rationals and the inequalities are checked by
clearing positive denominators.  Thus $\Delta_{\mathbf t}<0$ and all four
$m_i$ are strictly negative.  This is a finite nondegenerate de Sitter
candidate.
}

{Set $Q:=\mathcal G_{\mathbf t}^{-1}$.  Exact evaluation via
$Q=\operatorname{adj}(\mathcal G_{\mathbf t})/\Delta_{\mathbf t}$ gives
}\begin{equation}
{\begin{aligned}
0.04408&<Q_{11}<0.04409\,,
&0.07413&<Q_{22}<0.07414\,,\\
0.001470&<Q_{33}<0.001471\,,
&3.27\mathbin{\times}10^{-14}&<Q_{44}<3.28\mathbin{\times}10^{-14}\,,\\
0.006023&<Q_{41}<0.006024\,,
&0.004935&<Q_{42}<0.004936\,,\\
0.002018&<Q_{43}<0.002019\,,
\end{aligned}
\label{eq:counterexample_Q_intervals}
}\end{equation}
  {and
}\begin{equation}
{\begin{aligned}
0.00003628&<Q_{41}^2-Q_{44}Q_{11}<0.00003629\,,\\
0.00002435&<Q_{42}^2-Q_{44}Q_{22}<0.00002436\,,\\
0.000004074&<Q_{43}^2-Q_{44}Q_{33}<0.000004075\,.
\end{aligned}
\label{eq:counterexample_binary_obstructions}
}\end{equation}

{Since \eqref{eq:counterexample_delta_interval} gives
} $\sigma_{\mathcal G}=+1$,   {a ray-sign choice
$\varepsilon\in\{\pm1\}^4$ has domain matrix
$A_\varepsilon=S_\varepsilon QS_\varepsilon$.  Copositivity of a symmetric
$2\times2$ principal block
\mbox{%
$\left(\begin{smallmatrix}a&b\\b&c\end{smallmatrix}\right)$
}%
with
{$a,c>0$} requires
}\begin{equation}
{b+\sqrt{ac}\geq0\,.
\label{eq:binary_copositive_necessary}
}\end{equation}

{Indeed, if $b+\sqrt{ac}<0$, then
\be
\begin{pmatrix}\sqrt c&\sqrt a\end{pmatrix}
\begin{pmatrix}a&b\\b&c\end{pmatrix}
\begin{pmatrix}\sqrt c\\\sqrt a\end{pmatrix}
=2\sqrt{ac}\bigl(b+\sqrt{ac}\bigr)<0\,.
\ee
}
  If
$\varepsilon_j\neq\varepsilon_4$, then the $(4,j)$ entry of
$A_\varepsilon$ is $-Q_{4j}$.  Equations
\eqref{eq:counterexample_Q_intervals}--
\eqref{eq:counterexample_binary_obstructions} give
\begin{equation}
{-Q_{4j}+\sqrt{Q_{44}Q_{jj}}<0\,,
\qquad j=1,2,3\,,
}\end{equation}
 {so \eqref{eq:binary_copositive_necessary} fails.  Hence copositivity would
force $\varepsilon_1=\varepsilon_2=\varepsilon_3=\varepsilon_4$, in which
case $A_\varepsilon=Q$.
}

{Finally take the exact rational vector
}\begin{equation}
{\lambda=\frac{1}{1000}(293,230,477,0)^{\top}
\in\mathbb R_{\geq0}^4\,.
}\end{equation}
{Exact rational evaluation gives
}\begin{equation}
{-0.001233<\lambda^{\top}Q\lambda<-0.001232\,.
\label{eq:counterexample_witness_interval}
}\end{equation}
 Thus $Q$ is not copositive either.  No one of the sixteen sign choices is
copositive, and in particular the uniquely selected outward branch is not
strictly copositive.  This closing rational character point proves that the
finite domain-admissibility hypothesis
is not implied by closure, the determinant sign, and the four
displayed vertex principal-minor tests.

\section{{Numerical reconstruction and
domain-obstruction tests}}
\label{subsec:global_branch_examples}

 {This appendix presents two numerical tests of
Theorem~\ref{theorem:unified_minkowski}.
The first supports reconstruction of a finite anti-de Sitter
tetrahedron with mixed causal face types. The second passes the
local Gram and branch-sign tests but fails the global domain
condition numerically. The holonomies are defined exactly,
while the reconstruction tests are evaluated numerically.
No proof in the paper depends on these computations.}

We work in the fixed based tangent model
\begin{equation}
(T_v\mathcal M,\eta_T)\simeq(\mathbb R^{1,2},\mathrm{diag}(-1,1,1))
\end{equation}
with the cross-product convention of Section \ref{sec:unified_notation}.  In this basis
\begin{equation}
J(u^0,u^1,u^2)=
\begin{pmatrix}
0&u^2&-u^1\\
u^2&0&-u^0\\
-u^1&u^0&0
\end{pmatrix}\,.
\label{eq:J_matrix_examples}
\end{equation}
For each example we prescribe $O_1,O_2,O_3$ explicitly and set
\begin{equation}
O_4:=(O_3O_2O_1)^{-1}\,.
\label{eq:example_O4_by_closure}
\end{equation}
Thus the based closure relation $O_4O_3O_2O_1=\id$ is exact.
{The exponential definitions and
\eqref{eq:example_O4_by_closure} specify the exact input data.
Matrices marked with $\approx$ are decimal approximations. From the fixed lines of the $O_i$, after the
branch choice indicated in each example, we
form the candidate Gram matrix $G$ using
\eqref{eq:reconstructed_gram}, including the exceptional entry
$G_{24}=\tanip{n_2}{O_1n_4}$.

\subsection{{An anti-de Sitter reconstruction test}}
 \label{subsec:mixed_ads_reconstruction_example}

{Consider the outward representatives
}\begin{equation}
{n_1=(-1,-1,0)\,,\qquad
n_2=(\cosh0.4,\,\sinh0.4,\,0)\,,\qquad
n_3=(0,0,1)\,,
\label{eq:mixed_reconstruction_normals}
}\end{equation}
{with squared norms $0,-1,+1$, respectively.  For the null face we use the
scale convention $\Theta_1=-1$, so its nilpotent logarithm vector is
$\kappa_1=\Theta_1n_1=-n_1=(1,1,0)$.  Define
}\begin{equation}
{\begin{aligned}
O_1&=\exp\big(-J(n_1)\big)\,,&
O_2&=\exp\big(-0.4J(n_2)\big)\,,\\
O_3&=\exp\big(-0.4J(n_3)\big)\,,&
O_4&=(O_3O_2O_1)^{-1}\,.
\end{aligned}
\label{eq:mixed_reconstruction_O}
}\end{equation}
{Thus $O_4O_3O_2O_1=\id$ exactly.  In particular,
}\begin{equation}
{O_1=
\begin{pmatrix}
1.5&-0.5&-1\\
0.5&0.5&-1\\
-1&1&1
\end{pmatrix}\,,
\qquad
O_4\approx
\begin{pmatrix}
1.416232395&0.459043464&0.891623965\\
0.745912349&1.076449112&0.630589044\\
0.670320046&-0.227987306&1.182095915
\end{pmatrix}\,.
}\end{equation}
The first three inputs are, respectively, nontrivial parabolic, elliptic,
and hyperbolic elements.  Moreover,
$\trvec(O_4)\approx3.674777422>3$, so
$O_4$ lies numerically in the nontrivial hyperbolic class.
Its selected candidate unit spacelike fixed representative is
\begin{equation}
n_4\approx(0.483412839,\,0.879436977,\,-0.678438337)\,,
\qquad
\tanip{n_4}{n_4}\approx1\,.
\end{equation}
{The exceptional transport equality is numerically
}\begin{equation}
{\tanip{n_2}{O_1n_4}
=\tanip{n_2}{O_3^{-1}n_4}
\approx-0.483412839\,,
}\end{equation}
where the equality follows exactly from closure.  Hence
the resulting numerical candidate outward Gram matrix is
\begin{equation}
{{G}\approx
\begin{pmatrix}
0&0.670320046&0&-0.396024138\\
0.670320046&-1&0&-0.483412839\\
0&0&1&-0.678438337\\
-0.396024138&-0.483412839&-0.678438337&1
\end{pmatrix}\,.
\label{eq:mixed_reconstruction_gram}
}\end{equation}
{The determinant, vertex principal determinants, and transported triples are
}\begin{align}
{\det {G}}&{\approx0.170979016>0\,,\nonumber}\\
{\bigl(\det G_{\hat i}\bigr)_{i=1}^4}
&{\approx(-0.773409396,\,-0.156835118,\,-0.035837481,\,-0.449328964)\,,
\label{eq:mixed_reconstruction_determinants}}\\
{(\chi_1,\chi_2,\chi_3,\chi_4)
}&{\approx(0.879436977,\,0.396024138,\,0.189307899,\,0.670320046)\,.
\label{eq:mixed_reconstruction_triples}
}\end{align}
At the displayed precision, the representatives have the
common-positive candidate branch-sign pattern.
The identities
{$\det G_{\hat i}=-\chi_i^2$} are numerically satisfied
entry by entry.

{For the anti-de Sitter sign, the global domain form is
}\begin{equation}
{A:=-G^{-1}}
\approx
\begin{pmatrix}
4.523417062&0.996279782&2.857192783&4.211425899\\
0.996279782&0.917276995&1.053346928&1.552605256\\
2.857192783&1.053346928&0.209601632&1.782920519\\
4.211425899&1.552605256&1.782920519&2.627977256
\end{pmatrix}\,.
\label{eq:mixed_reconstruction_signed_inverse}
\end{equation}

An 80-digit numerical evaluation gives $A_{ij}>0.2096$ for all entries,
with the smallest at $(3,3)$, approximately $0.209601631751$.
Thus, numerically, for every $0\neq\lambda\in\mathbb R_{\geq0}^4$,
\begin{equation}
{\lambda^\top {A}\lambda
=\sum_{i,j}{A_{ij}}\lambda_i\lambda_j
>0.2096\left(\sum_i\lambda_i\right)^2>0\,.
\label{eq:mixed_reconstruction_copositivity}
}\end{equation}

{The numerical tests above support membership in the anti-de Sitter reconstruction chamber and illustrate} the unique labeled finite
domain-admissible tetrahedron supplied by Theorem
\ref{theorem:unified_minkowski}, realizing the exact holonomies in
\eqref{eq:mixed_reconstruction_O}.

{Using a numerical Gram factorization of $G$
in the ambient metric $\diag(-1,-1,1,1)$, the four candidate
vertex rays may be normalized individually as follows.}
\begin{equation}
{\begin{aligned}
V_1&\approx(0.163167,\,1.035641,\,-0.304735,\,-0.079454)\,,\\
V_2&\approx(-0.806957,\,0.728893,\,-0.403552,\,-0.140039)\,,\\
V_3&\approx(-0.067939,\,2.362880,\,1.895378,\,-0.997677)\,,\\
V_4&\approx(-0.066996,\,1.085967,\,0.265838,\,0.336366)\,.
\end{aligned}
\label{eq:mixed_reconstruction_vertices}
}\end{equation}
To the displayed precision, they satisfy
$\ambip{V_i}{V_i}=-1$, while the opposite outward supports are
\begin{equation}
\bigl(\ambip{V_i}{N_i}\bigr)_{i=1}^4
\approx(-0.470183,\,-1.044118,\,-2.184252,\,-0.616864)\,.
\end{equation}
{Here $N_i$ denotes the corresponding ambient normal
in the chosen numerical Gram factorization.}

\subsection{An anti-de Sitter domain-obstruction test}
\label{subsec:mixed_ads_domain_obstruction}

{Choose the normal representatives}
$$
\begin{aligned}
n_1:=\kappa_1&=(-1\,,-1\,,0)\,,\\
n_2&=(\cosh0.8\,,-\sinh0.8\,,0)\,,\\
n_3&=\left(\sinh0.8\,,
\frac{\cosh0.8}{\sqrt2}\,,
\frac{\cosh0.8}{\sqrt2}\right)\,.
\end{aligned}
$$
{Their squared norms are $0$, $-1$, and $+1$, respectively.
We use the null scale convention $\Theta_1=1$.}
The prescribed holonomies are
$$
\begin{aligned}
O_1&=\exp\big(J(\kappa_1)\big)\,,
&
O_2&=\exp\big(J(n_2)\big)\,,\\
O_3&=\exp\big(0.8\,J(n_3)\big)\,,
&
O_4&=(O_3O_2O_1)^{-1}\,.
\end{aligned}
$$
Thus $O_1$ is parabolic, $O_2$ is elliptic, and $O_3$ is
hyperbolic.
{The value $\trvec(O_4)\approx2.301884$ indicates that
$O_4$ is elliptic. Its selected fixed representative is}
$$
n_4\approx(-1.448075\,,0.271183\,,1.011623)\,,
\qquad
\tanip{n_4}{n_4}\approx-1\,.
$$

{The reconstructed Gram matrix is}
$$
G\approx
\begin{pmatrix}
0&2.225541&-0.057603&-1.719258\\
2.225541&-1&-2.027674&1.357598\\
-0.057603&-2.027674&1&2.499205\\
-1.719258&1.357598&2.499205&-1
\end{pmatrix}\,.
$$
The determinant and vertex principal determinants are
$$
\begin{aligned}
\det G&\approx2.232581>0\,,\\
\bigl(\det G_{\hat i}\bigr)_{i=1}^4
&\approx
(-4.245038\,,-2.457511\,,-2.480216\,,-4.429824)\,.
\end{aligned}
$$
The four transported triples are
$$
(\chi_1\,,\chi_2\,,\chi_3\,,\chi_4)
\approx
(2.060349\,,1.567645\,,1.574870\,,2.104715)\,.
$$
{These signs indicate anti-de Sitter inertia, four finite
candidate vertices, and the positive-triple outward branch.
It remains to check the global domain condition.}

{Set $A:=-G^{-1}$. Its $(1,2)$ principal block is}
$$
\begin{pmatrix}
A_{11}&A_{12}\\
A_{12}&A_{22}
\end{pmatrix}
\approx
\begin{pmatrix}
1.901404&-2.310531\\
-2.310531&1.100749
\end{pmatrix}\,.
$$
{Evaluation at the nonnegative coefficient vector}
$$
\lambda=\left(\frac9{20}\,,\frac{11}{20}\,,0\,,0\right)^\top
$$
{gives}
$$
\lambda^\top A\lambda\approx-0.425702<0\,.
$$
{This gives numerical evidence that $A$ is not copositive.
Because $\lambda$ has only its first two components nonzero,
the obstruction occurs on the projective edge joining the
first two candidate vertex rays. Thus the local Gram and
branch-sign tests pass, while the global domain-containment
test fails numerically.}

\bibliographystyle{alpha}
\bibliography{GSF-short}

\end{document}